\documentclass{jfm}
\usepackage{xcolor}
\usepackage{graphicx}
\usepackage{epstopdf, epsfig}
\usepackage{amsmath}
\usepackage{bm}

\usepackage[T1]{fontenc}
\usepackage[english]{babel}
\usepackage{float}
\usepackage{amssymb}
\usepackage{algorithm,algorithmic}
\usepackage[unicode=true]{hyperref}

\makeatletter
\newtheorem{prop}{\protect\propositionname}
\newtheorem{defn}{\protect\definitionname}
\newtheorem{thm}{\protect\theoremname}
\newtheorem{example}{\protect\examplename}
\newtheorem{rem}{\protect\remarkname}

\makeatother

\providecommand{\definitionname}{Definition}
\providecommand{\examplename}{Example}
\providecommand{\propositionname}{Proposition}
\providecommand{\remarkname}{Remark}
\providecommand{\theoremname}{Theorem}

\usepackage{xcolor}
\usepackage[normalem]{ulem}

\shorttitle{Objective barriers to active transport}
\shortauthor{G. Haller, S. Katsanoulis, M. Holzner, B. Frohnapfel and D. Gatti}

\title{Objective barriers to the transport of dynamically active vector fields}

\author{George Haller\aff{1}\corresp{\email{georgehaller@ethz.ch}},
  Stergios Katsanoulis\aff{1},
  Markus Holzner\aff{2},
 Bettina Frohnapfel\aff{3}
 \and Davide Gatti\aff{3}}

\affiliation{%
\aff{1} Institute for Mechanical Systems, ETH Z\"{u}rich, Z\"{u}rich, Switzerland
\aff{2} WSL Swiss Federal Research Institute, Birmensdorf, Switzerland%
\aff{3} Institute of Fluid Mechanics, Karlsruhe Institute of Technology,
Karlsruhe, Germany
}

\begin{document}

\maketitle

\begin{abstract}
We derive a theory for material surfaces that maximally inhibit the
diffusive transport of a dynamically active 
vector field, such as the linear momentum, the angular momentum or
the vorticity, in general fluid flows. These special material surfaces (\emph{Lagrangian active barriers}) provide physics-based, observer-independent boundaries of
dynamically active coherent structures. We find that Lagrangian active barriers
evolve from invariant surfaces of an associated steady and incompressible
\emph{barrier equation}, whose right-hand side is the time-averaged pullback of the viscous stress terms
in the evolution equation for the dynamically active vector field. Instantaneous
limits of these barriers mark objective \emph{Eulerian
active barriers} to the short-term diffusive transport of the dynamically active vector field. We obtain that in unsteady
Beltrami flows, Lagrangian and Eulerian active barriers  coincide exactly with purely advective transport barriers bounding
observed coherent structures. In more general flows, active 
barriers can be identified by applying Lagrangian coherent
structure (LCS) diagnostics, such as the finite-time Lyapunov
exponent and the polar rotation angle, to the appropriate active barrier equation. In comparison to their
passive counterparts, these \emph{active LCS diagnostics} require no significant
fluid particle separation and hence provide substantially higher-resolved
Lagrangian and Eulerian coherent structure boundaries from temporally shorter
velocity data sets. We illustrate these results and their physical interpretation on two-dimensional, homogeneous, isotropic turbulence and on a three-dimensional turbulent channel flow.
\end{abstract}

\begin{keywords}
Keywords will be added upon submission
\end{keywords}

\section{Introduction}

Fluid transport is often the simplest to describe through its barriers.
Indeed, transport barriers are routinely invoked in discussions of
transport in classical fluid dynamics (Ottino 1989), geophysics (Weiss
\& Provenzale 1989), reactive flows (Rosner 2000) and plasma fusion
(Dinklage 2005).

Despite their broadly recognized significance, transport barriers
have remained loosely defined and little understood. The only generally
agreed definition is the one of MacKay, Meiss \& Percival (1984) who
define transport barriers in two-dimensional (2D), time-periodic flows as invariant
curves of the Poincar\'e (or stroboscopic) map for fluid particle motions.
This definition extends to three-dimensional (3D) steady flows, identifying
advective transport barriers as 2D material surfaces whose intersection
with a section transverse to the flow is an invariant curve for the
first-return map defined for that section (Ottino 1989, MacKay 1994).
In many 3D steady flows, however, trajectories may rarely if ever
return to the physically relevant Poincar\'e sections, such as the cross-stream
sections of pipe flows.

This lack of returns obliges one to look for barriers to advective
transport among all material surfaces\,\textendash \,an ill-defined objective,
given that all material surfaces are barriers to advective transport.
Indeed, none of them can be crossed by other material trajectories
by the uniqueness of trajectories through any point at a given time
in a smooth velocity field. Some material surfaces are nevertheless
perceived as organizers of advective transport because they preserve
their coherence, i.e., do not develop smaller scales (filamentation) in their evolution. These distinguished surfaces
are generally referred to as Lagrangian coherent structures (or LCS;
see Haller 2015). In the absence of a universally accepted notion
of material coherence, however, different LCS definitions continue
to coexist and highlight different material surfaces as advective
transport barriers (Hadjighasem et al. 2017). Beyond their diversity,
most LCS criteria have also been criticized for being purely kinematic
with no regard to relevant physical quantities, such
as the linear momentum and the vorticity. The need for developing LCS methods for the transport of such physical quantities has recently been stressed by Balasuriya, Ouellette \& Rypina (2018).

Parallel to the development of different LCS criteria, several different
Eulerian criteria for coherent vortices have been put forward (see
Epps 2017 and G\"unther \& Theisel 2018 for recent reviews). Most of
these approaches also set out to find sustained (Lagrangian) swirling
motion of fluid particles, but hope to achieve this goal by studying
local properties of instantaneous (Eulerian) velocity snapshots. As
this is a hopeless undertaking for unsteady flows, these approaches
invariably divert from their originally stated objective and postulate
coherence principles for the instantaneous velocity field, rather
than for particle motion. One can then a posteriori interpret the
resulting velocity-dependent inequalities (such as the $Q$-, $\Delta$-,
$\lambda_{2}$ - and $\lambda_{ci}$-criteria reviewed recently in
Pedergnana et al. 2020) as physical, but their actual connection to
flow physics is unclear due to the conceptual gaps in their derivations
and their dependence on the observer,

Unsurprisingly, therefore, the resulting vortex criteria often yield
erroneous results even for simple flows in which the coherent swirling
regions can be identified unambiguously from Poincar\'e maps (see Pedergnana
et al. 2020 for recent demonstrations). This has resulted in the practice of plotting a
few level sets of $Q$, $\Delta$, $\lambda_{2}$ or $\lambda_{ci}$, as opposed to verifying the inequalities imposed on these quantities by the appropriate criteria 
(see, e.g., Dubief \& Delcayre 2000, McMullan \& Page 2012, Anghan
et al. 2014, Gao et al. 2015, Jantzen et al. 2019). These level surfaces
are selectively chosen to match expectations or produce visually pleasing
images. As a further ad hoc element in this procedure, the level surfaces
are not objective: they depend on the frame of reference, even though
truly unsteady flows have no distinguished frame of reference (Lugt
1979). The experimental detectability or physical relevance of these
surfaces is, therefore, unclear. Arguably, as long as this practice
continues, there is little hope for a commonly accepted definition for coherent vortices. 

A way out of this conundrum is to identify coherent structures based
on the transport of physical quantities of interest to the fluid mechanics
community, but use mathematical deductions that are free from ad hoc
assumptions, user-defined thresholds and tunable parameters. Specifically,
one may seek the boundaries of coherent structures or vortices based
on their transport-extremizing properties. Unlike the notions of coherence
and swirling, the notion of transport through a surface is physically
well-understood, quantitative and frame-independent, when properly
phrased. These features allow for a systematic, quantitative comparison of all surfaces to find minimizers (barriers) of transport among them. This in turn offers a way to quantify the general view in fluid mechanics that coherent structures influence transport processes in turbulent flows (Robinson 1991, Hutchins \& Marusic 2007).

As a first step in this direction, Haller, Karrasch \& Kogelbauer
(2018, 2019) formalize the definition of transport barriers for passively
advected diffusive scalars. They then locate transport barriers as
material surfaces that inhibit the diffusive transport of a weakly
diffusive scalar more than neighboring material surfaces do. Katsanoulis
et al. (2019) use these results to locate vortex boundaries in 2D
flows as outermost closed barriers to the diffusive transport of the
scalar vorticity. These results, however, do not cover barriers
to the transport of dynamically active vector fields, such as momentum and vorticity, in 3D. There
are also examples, such as the 2D decaying channel flow shown in Fig.
 \ref{fig: steady 2D channel flow example-0}, in which the passive-scalar-based
approach to vorticity transport only captures the walls as perfect
transport barriers in a finite-time analysis. The remaining observed
barriers to the redistribution of the normalized vorticity (i.e.,
all horizontal lines) are only captured by the approach over an infinitely
long time interval. 
\begin{figure}
\centering{}\includegraphics[width=0.85\textwidth]{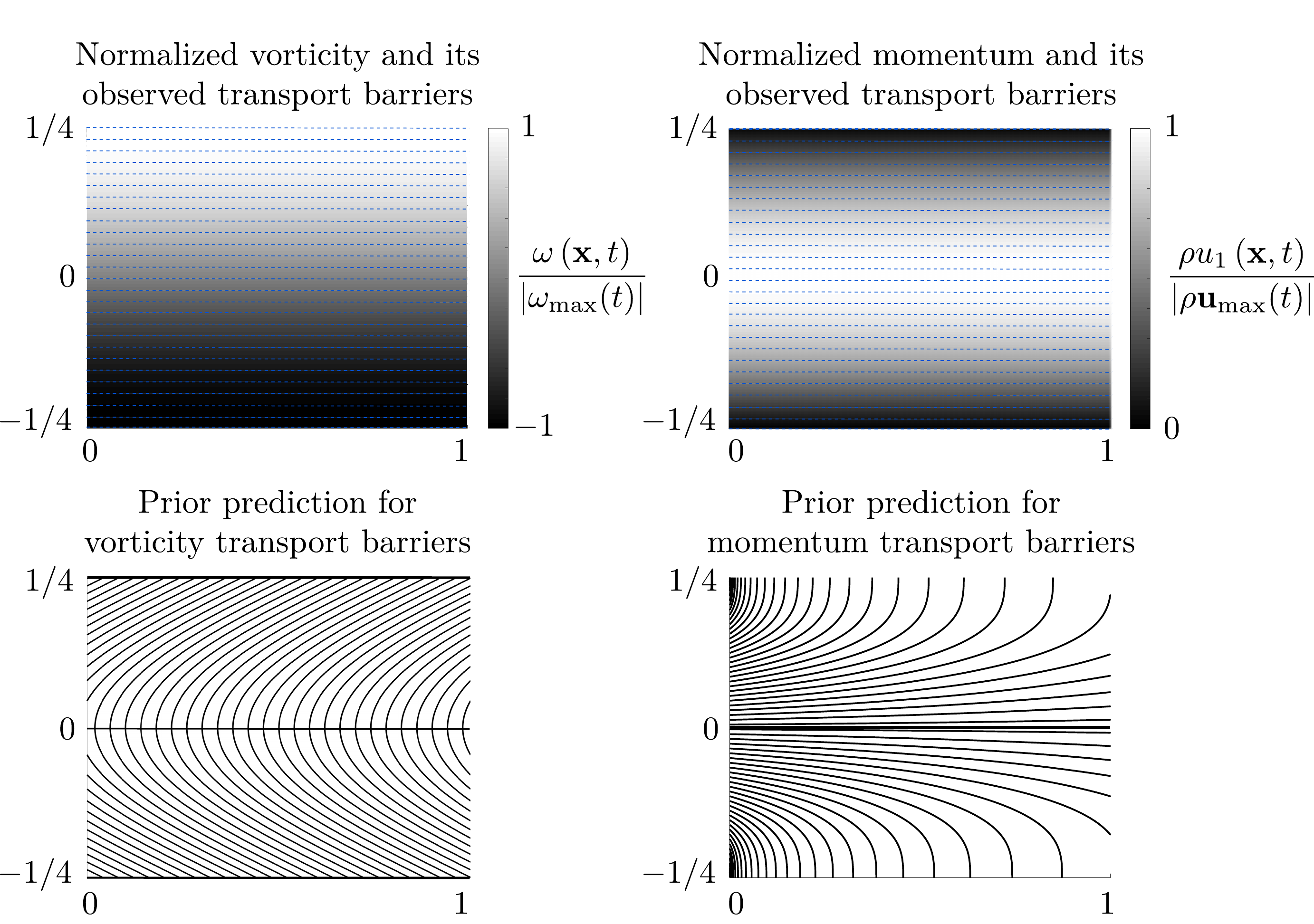}\caption{Vorticity and linear momentum,  normalized by their maxima at an arbitrary
time instance in a decaying planar channel flow. These plots remain
steady in time, with all horizontal lines (some shown dotted) acting
as barriers to the vertical redistribution of the vorticity and linear momentum.
Also shown are prior predictions for perfect barriers to vorticity
transport in this flow by Haller, Karrasch \& Kogelbauer (2019) on the left  and
for perfect barriers to momentum transport by Meyers  \&
Meneveau (2013) on the right. The latter barrier trajectories are released uniformly across the entry cross section of the channel. See Appendix A for details.}
\label{fig: steady 2D channel flow example-0}
\end{figure}

More broadly speaking, there has been a lack of methods to identify barriers to the transport
of dynamically active quantities, i.e., scalar, vector or tensor fields whose evolution impacts the evolution of the underlying fluid velocity field. A notable
exception is the work of Meyers \& Meneveau (2013), who locate momentum-
and energy-transport barriers as tubes tangent to a flux vector field
formally associated with these dynamically active scalar fields. While
insightful, this approach also has several heuristic elements. The
construct depends on the frame of reference and the choice of a transport
direction. The flow data is assumed statistically stationary with
a well-defined mean velocity field. The proposed flux vector introduced
in this fashion is non-unique: any divergence-free vector field could
be added to it. Finally, the flux vector differs from the classic
momentum and energy flux that it purports to represent. All these
features of the approach prevent the detection of most observed barriers
to momentum redistribution already in simple 2D flows, such as our 2D decaying
channel flow example in Fig.  \ref{fig: steady 2D channel flow example-0}.
Indeed, the only horizontal barrier captured by this approach is the
symmetry axis of the channel. 

In the present work, we seek to fill the gaps in previous approaches
by extending the transport-barrier-detection approach of Haller, Karrasch
\& Kogelbauer (2018, 2019) to active transport in 3D. In this extension,
we seek material barriers to the diffusive (or viscosity-induced)
transport of an arbitrary\emph{ dynamically active vector field},
by which we mean a vector field whose evolution impacts the evolution of the underlying
fluid velocity field. We then seek transport barriers
as special material surfaces across which the net diffusive transport
of the active vector field pointwise vanishes. When applied to the 2D channel flow
example shown in Fig.  \ref{fig: steady 2D channel flow example-0},
the approach we develop here returns the observed material barriers
(all horizontal lines) as barriers to the spatial redistribution of
vorticity and momentum (see Example \ref{ex:2D analytic example} in section \ref{subsec:active barriers in 2D-Navier=002013Stokes-flows}).
This example and more complex examples discussed later illustrate
that material barriers to active transport can be used to define boundaries
of dynamical coherent structures (i.e., time-varying structures observed
in dynamically active vector fields) in a frame-independent fashion. 

The outline of this paper is as follows. In section 2, we introduce
our set-up and notation for a dynamically active vector field. We
then discuss in section 3 the shortcomings of available flux definitions
when applied to active transport through material surfaces, and introduce
an objective notion of diffusive transport for active vector fields.
In section 4, we identify surfaces blocking this diffusive transport
and define active transport barriers more formally. Section 5 describes
the instantaneous, Eulerian limits of these active barriers, and section
6 derives the equations for both Lagrangian and Eulerian active barriers to
the diffusive transport of linear momentum, angular momentum and vorticity.
In section 7, we work out solutions of these barrier equations analytically
for 2D Navier\textendash Stokes flows and 3D directionally steady
Beltrami flows. Section 8 discusses computational aspects of active
transport barriers and introduces active versions of passive LCS-detection
tools that generally enable a higher-resolved identification of coherent
structures from finite-time flow data than their passive counterparts
do. Section 9 shows such computations and their physical implications for 2D homogeneous, isotropic turbulence and for a 3D turbulent channel flow.   We summarize our conclusions in
section 10. Appendix A illustrates on a simple example the challenges
of defining active barriers with an observable footprint. Appendix
B motivates the need for a new definition for diffusive flux through
material surfaces. Finally, Appendices C and D contain the detailed
proofs of our technical results.

\section{Set-up\label{sec:Set-up}}

We consider a 3D flow with velocity field $\mathbf{u}(\mathbf{x},t)$
and density $\rho(\mathbf{x},t)$, known at spatial locations $\mathbf{x}\in U\in\mathbb{R}^{3}$
in a bounded set $U$ at times $t\in\left[t_{1},t_{2}\right]$. The
equation of motion for such a flow is of the general form
\begin{equation}
\rho\frac{D\mathbf{u}}{Dt}=-\boldsymbol{\nabla}p+\boldsymbol{\nabla}\cdot\mathbf{T}_{vis}+\mathbf{q},\label{eq:main continuum eq. of motion}
\end{equation}
where $D/Dt$ denotes the material derivative, $p(\mathbf{x},t)$
is the (equilibrium) pressure, $\mathbf{T}_{vis}(\mathbf{x},t)=\mathbf{T}_{vis}^{T}(\mathbf{x},t)$
is the viscous stress tensor and $\mathbf{q}(\mathbf{x},t)$ denotes
the external body forces (see Gurtin, Fried \& Anand 2013). 

Material trajectories generated by the velocity field \textbf{$\mathbf{u}$}
are solutions of the differential equation $\dot{\mathbf{x}}=\mathbf{u}(\mathbf{x},t)$.
We denote the time-$t$ position of a trajectory starting from $\mathbf{x}_{0}$
at time $t_{0}$ by $\mathbf{x}(t;t_{0},\mathbf{x}_{0})$. The flow
map induced by $\mathbf{u}$ is defined as the mapping $\mathbf{F}_{t_{0}}^{t}\colon\mathbf{x}_{0}\mapsto\mathbf{x}(t;t_{0},\mathbf{x}_{0})$.
A \emph{material surface} $\mathcal{M}(t)\subset U$ is a time-dependent
two-dimensional manifold transported by the flow map from its initial
position $\mathcal{M}_{0}:=\mathcal{M}(t_{0})$ as
\begin{equation}
\mathcal{M}(t)=\mathbf{F}_{t_{0}}^{t}\left[\mathcal{M}(t_{0})\right].\label{eq:material surface def}
\end{equation}

Let $\mathbf{f}(\mathbf{x},t)$ be another smooth vector field defined
on the same spatiotemporal domain $U\times\left[t_{0},t_{1}\right]$.
We will  be interested in $\mathbf{f}$ fields that are \emph{dynamically
active vector fields}, i.e., their evolution impacts the evolution of the velocity field $\mathbf{u}$. Such a vector field $\mathbf{f}$ is typically defined as a function   of $\mathbf{u}$ and its derivatives.
The simplest physical examples of  active vector fields are the
linear momentum $\mathbf{f}:=\rho\mathbf{u}$ and the vorticity $\mathbf{f}:=\boldsymbol{\omega}=\boldsymbol{\nabla}\times\mathbf{u}$. Both of these examples of $\mathbf{f}$ are \emph{frame-dependent}
(\emph{non-objective}) vector fields, because they do not transform
properly under general frame changes of the form
\begin{equation}
\mathbf{x}=\mathbf{Q}(t)\mathbf{y}+\mathbf{b}(t),\quad\mathbf{QQ}^{T}=\mathbf{I},\quad\mathbf{Q}(t)\in SO(3),\quad\mathbf{b}(t)\in\mathbb{R}^{3},\label{eq:observer change}
\end{equation}
where both $\mathbf{Q}(t)$ and $\mathbf{b}(t)$ are smooth in time.
Indeed, evaluating the definition of these vectors in the $\mathbf{y}$-frame
gives transformed vector fields $\tilde{\mathbf{f}}(\mathbf{y},t)$
for which\footnote{Specifically, $\rho\tilde{\mathbf{u}}=\mathbf{Q}^{T}\left(\rho\mathbf{u}-\dot{\mathbf{Q}}\mathbf{y}-\dot{\mathbf{b}}\right)$
and $\tilde{\boldsymbol{\omega}}=\mathbf{Q}^{T}\left(\boldsymbol{\omega}-\dot{\mathbf{q}}\right)$,
where the vorticity of the frame change, $\dot{\mathbf{q}}$, is defined
by the requirement that $\frac{1}{2}\dot{\mathbf{q}}\times\mathbf{e}=\dot{\mathbf{Q}}\mathbf{Q}^{T}\mathbf{e}$
for all vectors $\mathbf{e}\in\mathbb{R}^{3}.$}
\begin{equation}
\tilde{\mathbf{f}}(\mathbf{y},t)\neq\mathbf{Q}^{T}(t)\mathbf{f}(\mathbf{x},t).\label{eq:non-objectivity of u}
\end{equation}
It is, therefore, a challenge to describe the transport of $\mathbf{f}$
through a material surface in an intrinsic, observer-independent fashion. 

We assume that the evolution of $\mathbf{f}$ is governed by a partial
differential equation of the form
\begin{equation}
\frac{D}{Dt}\mathbf{f}=\mathbf{h}_{vis}+\mathbf{h}_{nonvis},\qquad\partial_{\mathbf{T}_{vis}}\mathbf{h}_{vis}\neq\mathbf{0},\quad\partial_{\mathbf{T}_{vis}}\mathbf{h}_{nonvis}=\mathbf{0}.\label{eq:DuDt}
\end{equation}
The function $\mathbf{h}_{vis}(\mathbf{x},t,\mathbf{u},\mathbf{f},\mathbf{T}_{vis})$
contains all the terms arising from diffusive forces (i.e., viscous
Cauchy stresses), while $\mathbf{h}_{nonvis}(\mathbf{x},t,\mathbf{u},\mathbf{f})$
has no explicit dependence on those forces. Instead, $\mathbf{h}_{nonvis}$
contains terms originating from the pressure, external forces and
possible inertial effects. For instance, as we will see in section
\ref{sec:Active-barrier-equations-derivations}, when $\mathbf{f}$
is the linear momentum of an incompressible Navier\textendash Stokes
flow with kinematic viscosity $\nu$, then we have \textbf{$\mathbf{h}_{vis}:=\rho\nu\Delta\mathbf{u}$}.
Or if, for the same class of flows, \textbf{$\mathbf{f}$ }equals
the vorticity $\boldsymbol{\omega}=\boldsymbol{\nabla}\times\mathbf{u}$
, then we have $\mathbf{h}_{vis}:=\nu\Delta\boldsymbol{\omega}$.

We finally assume that $\mathbf{h}_{vis}$ is an objective vector
field, i.e., under any observer change of the form \eqref{eq:observer change},
we obtain the transformed vector field $\mathbf{\tilde{h}}_{vis}$
in the form
\begin{equation}
\mathbf{\tilde{h}}_{vis}(\mathbf{y},t,\tilde{\mathbf{u}},\tilde{\mathbf{f}},\tilde{\mathbf{T}}_{vis})=\mathbf{Q}^{T}(t)\mathbf{h}_{vis}(\mathbf{x},t,\mathbf{u},\mathbf{f},\mathbf{T}_{vis}).\label{eq:objectivity condition}
\end{equation}
In all examples of $\mathbf{f}$ considered in this paper, this objectivity condition will hold, but one can certainly define dynamically active vector fields (e.g., $\mathbf{f}:=|\mathbf{u}|\mathbf{u}$) that do not satisfy condition \eqref{eq:objectivity condition}. With its dependence on inertial effects, the vector field $\mathbf{h}_{nonvis}$
is not objective.

\section{Active transport through material surfaces}

We seek to quantity the diffusive transport of the active vector field
$\mathbf{f}(\mathbf{x},t)$ through a material surface $\mathcal{M}(t)$
with a smoothly oriented unit normal vector field $\mathbf{n}(\mathbf{x},t)$.
While there is broad agreement on the notion of the flux of a passive
scalar field through a surface (see, e.g., Batchelor 2000), different
notions of the flux of an active vector field coexist. For instance,
the vorticity flux through $\mathcal{M}(t)$ (see, e.g., Childress
2009) is defined as 
\begin{equation}
\mathrm{Flux}_{\boldsymbol{\omega}}\left(\mathcal{M}(t)\right)=\int_{\mathcal{M}(t)}\boldsymbol{\omega}\cdot\mathbf{n\,}dA,\label{eq:vorticity flux}
\end{equation}
which measures the degree to which $\boldsymbol{\omega}$ is transverse
to $\mathcal{M}(t)$ on average, as opposed to the rate at which vorticity
is transported through $\mathcal{M}(t)$. Another broadly used quantity
is the linear momentum flux through $\mathcal{M}(t)$ (see, e.g.,
Bird et al. 2007), defined as 
\begin{equation}
\mathrm{\mathbf{Flux}}_{\rho\mathbf{u}}\left(\mathcal{M}(t)\right)=\int_{\mathcal{M}(t)}\rho\mathbf{u}\left(\mathbf{u}\cdot\mathbf{n}\right)\,dA.\label{eq:classic momentum flux}
\end{equation}
This expression is originally conceived for non-material surfaces,
formally measuring the rate at which $\rho\mathbf{u}$ is carried
through $\mathcal{M}(t)$ by trajectories. However, no such convective
flux is possible when $\mathcal{M}(t)$ is a material surface, which
can never be crossed by material trajectories. As a consequence, $\mathrm{\mathbf{Flux}}_{\rho\mathbf{u}}\left(\mathcal{M}(t)\right)$
does not capture the full flux through material surfaces (see Appendix
B for a simple example).

Beyond the issues already mentioned for $\mathrm{Flux}_{\boldsymbol{\omega}}$
and $\mathrm{\mathbf{Flux}}_{\rho\mathbf{u}}$, these flux notions
have further common shortcomings for the purposes of defining an intrinsic
flux through material surfaces. First, one expects a flux of a quantity
through a surface to have the units of that quantity divided by time
and multiplied by the surface area. This not the case for either $\mathrm{Flux}_{\boldsymbol{\omega}}$
or $\mathrm{\mathbf{Flux}}_{\rho\mathbf{u}}$. Second, as the mass
flux and the diffusive flux of a tracer through a material surface
are objective (Haller, Karrasch \& Kogelbauer 2018, 2019), one expects
a truly intrinsic flux of a vector field through a material surface
to be objective as well: it should remain unchanged under all observer
changes of the form \eqref{eq:observer change}. A direct calculation
shows that neither $\mathrm{Flux}_{\boldsymbol{\omega}}$ nor $\mathrm{\mathbf{Flux}}_{\rho\mathbf{u}}$
are objective, which is the result of the frame-dependence of $\boldsymbol{\omega}$
and $\mathbf{u}$ (see, e.g., Haller 2015). 

As a consequence of this frame-dependence, specific values of $\mathrm{Flux}_{\boldsymbol{\omega}}$
and $\mathrm{\mathbf{Flux}}_{\rho\mathbf{u}}$ carry no intrinsic
meaning in general unsteady fluid flows, because such flows have no
distinguished frames of reference (Lugt 1979). This prevents us from
locating intrinsic (and hence observer-independent) barriers to the
transport of vorticity and momentum using these fluxes. Specifically,
the classic notion of a vortex tube\footnote{i.e., a cylindrical surface $\mathcal{A}(t)$ with pointwise zero
vorticity flux $\boldsymbol{\omega}(\mathbf{x},t)\cdot\mathbf{n}(\mathbf{x},t)$,
which implies $\mathrm{Flux}_{\boldsymbol{\omega}}\left(\mathcal{A}(t)\right)=0$. }, defined via $\mathrm{Flux}_{\boldsymbol{\omega}}$, is not objective:
observers rotating relative to each other will identify different
surfaces as vortex tubes. This holds even for inviscid flows, in which
all vortex tubes are material surfaces (see Batchelor 2000).

To address these shortcomings of commonly used vector-field-flux definitions,
we introduce the \emph{diffusive flux} of $\mathbf{f}\left(\mathbf{x},t\right)$
through $\mathcal{M}(t)$ by integrating the diffusive component of
the surface-normal material derivative of $\mathbf{f}\left(\mathbf{x},t\right)$
over $\mathcal{M}(t)$:
\begin{equation}
\Phi\left(\mathcal{M}(t)\right)=\left[\int_{\mathcal{M}(t)}\frac{D\mathbf{f}}{Dt}\cdot\mathbf{n}\,dA\right]_{vis}=\int_{\mathcal{M}(t)}\mathbf{h}_{vis}\cdot\mathbf{n}\,dA.\label{eq:internal flux}
\end{equation}
Physically, the diffusive flux $\Phi$ measures the extent to which the
diffusive component of the rate-of-change of $\mathbf{f}$ along trajectories forming the surface $\mathcal{M}(t)$ is non-tangent to  $\mathcal{M}(t)$. Trajectories do not need to cross the material
surface $\mathcal{M}(t)$ to generate diffusive flux.

The diffusive flux $\Phi$ has the physical units expected for the
flux of $\mathbf{f}$: the units of \textbf{$\mathbf{f}$ }multiplied
by area and divided by time. Under an observer change of the form
\eqref{eq:observer change}, the transformation formula $\mathbf{n}=\mathbf{Q}\tilde{\mathbf{n}}$
for unit normals and the assumption \eqref{eq:objectivity condition}
on the active vector field\textbf{ $\mathbf{f}$ }imply that
\begin{equation}
\tilde{\Phi}\left(\tilde{\mathcal{M}}(t)\right)=\int_{\tilde{\mathcal{M}}(t)}\tilde{\mathbf{h}}_{vis}\cdot\tilde{\mathbf{n}}\,d\tilde{A}=\int_{\mathcal{M}(t)}\left(\mathbf{Q}^{T}\mathbf{h}_{vis}\right)\cdot\left(\mathbf{Q}^{T}\mathbf{n}\right)dA=\int_{\mathcal{M}(t)}\mathbf{h}_{vis}\cdot\mathbf{n}dA=\Phi\left(\mathcal{M}(t)\right),
\end{equation}
 and hence the diffusive flux of $\mathbf{f}$ is also objective,
i.e., invariant under all observer changes.

With this dimensionally correct and objective notion of the flux at
hand, we can now define the diffusive transport of $\mathbf{f}\left(\mathbf{x},t\right)$
through $\mathcal{M}(t)$ over a time interval $\left[t_{0},t_{1}\right]$
as the time-integral of $\Phi\left(\mathcal{M}(t)\right)$ over $\left[t_{0},t_{1}\right]$.
To compare the overall ability of surfaces to withstand the diffusive
transport of $\mathbf{f}\left(\mathbf{x},t\right)$ over different
time intervals, we will work with the time-normalized total diffusive
transport, given by the \emph{diffusive transport functional }
\begin{align}
\psi_{t_{0}}^{t_{1}}\left(\mathcal{M}_{0}\right) & =\frac{1}{t_{1}-t_{0}}\int_{t_{0}}^{t_{1}}\Phi\left(\mathcal{M}(t)\right)dt\nonumber \\
 & =\frac{1}{t_{1}-t_{0}}\int_{t_{0}}^{t_{1}}\int_{\mathcal{M}(t)}\mathbf{h}_{vis}\cdot\mathbf{n}dA\,dt.\label{eq:psi flux definition-1}
\end{align}
The time integration of this functional is carried out along trajectories
forming the evolving material surface $\mathcal{M}(t)$.  We view
$\psi_{t_{0}}^{t_{1}}$ purely as a function of $\mathcal{M}_{0}\equiv\mathcal{M}\left(t_{0}\right)$,
because later positions of the material surface $\mathcal{M}(t)$
are fully determined by the initial position $\mathcal{M}_{0}$ through
the relationship \eqref{eq:material surface def}. The functional
$\psi_{t_{0}}^{t_{1}}$ can also be viewed as the time-averaged
diffusive flux of the vector field $\mathbf{f}$ through $\mathcal{M}(t)$
over the time interval $\left[t_{0},t_{1}\right]$. As for any diffusion-induced
transport, $\psi_{t_{0}}^{t_{1}}\left(\mathcal{M}_{0}\right)$ is
expected to be small if the material surface $\mathcal{M}(t)$ remains
coherent, i.e., does not develop smaller scales (filamentation) during
its evolution.

To obtain a more explicit formula for $\psi_{t_{0}}^{t_{1}}\left(\mathcal{M}_{0}\right)$
while keeping our notation simple, we now introduce some notation.
For an arbitrary time-dependent Lagrangian vector field $\mathbf{v}(\mathbf{x}_{0},t)$,
we let 
\begin{equation}
\mathbf{\overline{v}}(\mathbf{x}_{0})=\frac{1}{t_{1}-t_{0}}\int_{t_{0}}^{t_{1}}\mathbf{v}\left(\mathbf{x}_{0},t\right)\,dt
\end{equation}
denote the temporal average of $\mathbf{v}(\mathbf{x}_{0},t)$ over
the time interval $[t_{0},t_{1}].$ We will also denote by $\left(\mathbf{F}_{t_{0}}^{t}\right)^{*}\mathbf{w}$
the pull-back of an Eulerian vector field $\mathbf{w}(\mathbf{x},t)$
under the flow map $\mathbf{F}_{t_{0}}^{t}$ to the initial configuration
at $t_{0}$, defined as
\begin{equation}
\left(\mathbf{F}_{t_{0}}^{t}\right)^{*}\mathbf{w}(\mathbf{x}_{0})=\left[\boldsymbol{\nabla}\mathbf{F}_{t_{0}}^{t}\left(\mathbf{x}_{0}\right)\right]^{-1}\mathbf{w}\left(\mathbf{F}_{t_{0}}^{t}\left(\mathbf{x}_{0}\right),t\right).
\end{equation}
With this notation, we obtain the following result:
\begin{thm}
\label{prop: psi formula}Under the assumptions \eqref{eq:DuDt}-\eqref{eq:objectivity condition}
on the dynamically active vector field $\mathbf{f}$, the diffusive
transport functional $\psi_{t_{0}}^{t_{1}}$ of \textbf{$\mathbf{f}$
}can be calculated as
\begin{equation}
\psi_{t_{0}}^{t_{1}}\left(\mathcal{M}_{0}\right)=\int_{\mathcal{M}_{0}}\mathbf{b}_{t_{0}}^{t_{1}}\cdot\mathbf{n}_{0}\,dA_{0},\label{eq:psi final}
\end{equation}
with the objective Lagrangian vector field 
\begin{equation}
\mathbf{b}_{t_{0}}^{t_{1}}:=\overline{\det\nabla\mathbf{F}_{t_{0}}^{t}\left(\mathbf{F}_{t_{0}}^{t}\right)^{*}\mathbf{h}_{vis}}.\label{eq:bdef-1}
\end{equation}
As a consequence, the diffusive transport, $\psi_{t_{0}}^{t_{1}}\left(\mathcal{M}_{0}\right)$,
is objective.
\end{thm}
\begin{proof}
Using the classic surface-element deformation formula 
\begin{equation}
\mathbf{n}dA=\det\nabla\mathbf{F}_{t_{0}}^{t}\left[\nabla\mathbf{F}_{t_{0}}^{t}\right]^{-T}\mathbf{n}_{0}dA_{0},\label{eq:surface_element_formula}\end{equation}
(see Gurtin, Fried \& Anand 2013) in eq. \eqref{eq:psi flux definition-1},
we obtain
\begin{align}
\psi_{t_{0}}^{t_{1}}\left(\mathcal{M}_{0}\right) & =\frac{1}{t_{1}-t_{0}}\int_{t_{0}}^{t_{1}}\int_{\mathcal{M}_{0}}\left.\mathbf{h}_{vis}\right|_{\mathbf{x}=\mathbf{F}_{t_{0}}^{t}(\mathbf{x}_{0})}\cdot\left(\det\nabla\mathbf{F}_{t_{0}}^{t}\left[\nabla\mathbf{F}_{t_{0}}^{t}\right]^{-T}\mathbf{n}_{0}dA_{0}\right)\,dt\nonumber \\
 & =\int_{\mathcal{M}_{0}}\left\{ \frac{1}{t_{1}-t_{0}}\int_{t_{0}}^{t_{1}}\det\nabla\mathbf{F}_{t_{0}}^{t}\left[\nabla\mathbf{F}_{t_{0}}^{t}\right]^{-1}\left.\mathbf{h}_{vis}\right|_{\mathbf{x}=\mathbf{F}_{t_{0}}^{t}\left(\mathbf{x}_{0}\right)}\cdot\mathbf{n}_{0}dA_{0}\right\} \,dt\nonumber \\
 & =\int_{\mathcal{M}_{0}}\mathbf{b}_{t_{0}}^{t_{1}}\cdot\mathbf{n}_{0}\,dA_{0},
\end{align}
with $\mathbf{b}_{t_{0}}^{t_{1}}(\mathbf{x}_{0})$ defined in \eqref{eq:bdef-1}.
The vector field $\mathbf{b}_{t_{0}}^{t_{1}}(\mathbf{x}_{0})$ is
objective in the Lagrangian sense (see Ogden 1984), because under
assumption \eqref{eq:objectivity condition}, an observer change of
the form \eqref{eq:observer change} gives
\begin{align}
\mathbf{b}_{t_{0}}^{t_{1}} & =\overline{\det\boldsymbol{\nabla}\mathbf{F}_{t_{0}}^{t}\left(\mathbf{F}_{t_{0}}^{t}\right)^{*}\mathbf{h}_{vis}}=\overline{\det\boldsymbol{\nabla}\mathbf{F}_{t_{0}}^{t}\left[\boldsymbol{\nabla}\mathbf{F}_{t_{0}}^{t}\right]^{-1}\mathbf{Q}(t)\tilde{\mathbf{h}}_{vis}}\nonumber \\
 & =\overline{\det\left[\mathbf{Q}(t)\boldsymbol{\tilde{\nabla}}\mathbf{\tilde{F}}_{t_{0}}^{t}\mathbf{Q}^{T}(t_{0})\right]\left[\mathbf{Q}(t)\boldsymbol{\tilde{\nabla}}\mathbf{\tilde{F}}_{t_{0}}^{t}\mathbf{Q}^{T}(t_{0})\right]^{-1}\mathbf{Q}(t)\tilde{\mathbf{h}}_{vis}}\nonumber \\
 & =\overline{\det\boldsymbol{\tilde{\nabla}}\mathbf{\tilde{F}}_{t_{0}}^{t}\mathbf{Q}(t_{0})\left[\boldsymbol{\tilde{\nabla}}\mathbf{\tilde{F}}_{t_{0}}^{t}\right]^{-1}\tilde{\mathbf{h}}_{vis}}\nonumber \\
 & =\mathbf{Q}(t_{0})\tilde{\mathbf{b}}_{t_{0}}^{t_{1}}.\label{eq:objectivity of b field}
\end{align}
As a result, we have
\begin{align}
\tilde{\psi}_{t_{0}}^{t_{1}}\left(\mathcal{\tilde{M}}_{0}\right)& =\int_{\mathcal{\tilde{M}}_{0}}\tilde{\mathbf{b}}_{t_{0}}^{t_{1}}\cdot\tilde{\mathbf{n}}_{0}\,d\tilde{A}_{0}=\int_{\mathcal{\tilde{M}}_{0}}\left(\mathbf{Q}^{T}(t_{0})\mathbf{b}_{t_{0}}^{t_{1}}\right)\cdot\left(\mathbf{Q}^{T}(t_{0})\mathbf{n}_{0}\right)\,d\tilde{A}_{0}=\int_{\mathcal{M}_{0}}\mathbf{b}_{t_{0}}^{t_{1}}\cdot\mathbf{n}_{0}\,dA_{0}\nonumber\\
& =\psi_{t_{0}}^{t_{1}}\left(\mathcal{M}_{0}\right),
\end{align}
proving the objectivity of $\psi_{t_{0}}^{t_{1}}\left(\mathcal{M}_{0}\right)$
. \\
\end{proof}
Theorem \ref{prop: psi formula} shows that $\psi_{t_{0}}^{t_{1}}\left(\mathcal{M}_{0}\right)$
can be calculated as the (algebraic) flux of the objective Lagrangian
vector field $\mathbf{b}_{t_{0}}^{t_{1}}(\mathbf{x}_{0})$ through
the initial surface $\mathcal{M}_{0}$. Following MacKay (1994), we
also define the \emph{geometric flux }of $\mathbf{b}_{t_{0}}^{t_{1}}$
through $\mathcal{M}_{0}$ as 
\begin{equation}
\Psi_{t_{0}}^{t_{1}}\left(\mathcal{M}_{0}\right)=\int_{\mathcal{M}_{0}}\left|\mathbf{b}_{t_{0}}^{t_{1}}\cdot\mathbf{n}_{0}\right|\,dA_{0}.\label{eq:geometric flux}
\end{equation}
This geometric flux cannot vanish due to global cancellations, and
hence is a better measure of the overall permeability (non-invariance)
of the surface $\mathcal{M}_{0}$ under the vector field $\mathbf{b}_{t_{0}}^{t_{1}}$
than the algebraic flux $\psi_{t_{0}}^{t_{1}}\left(\mathcal{M}_{0}\right)$.

\section{Lagrangian active barriers\label{sec:Material-(Lagrangian)-barriers}}

We seek diffusive transport barriers as material surfaces along which
the integrand in the diffusive transport functional $\psi_{t_{0}}^{t_{1}}$
vanishes pointwise. Therefore, the net transport of $\mathbf{f}$
due to viscous forces in the fluid is zero through any subset of such
a barrier.\textbf{ }Technically speaking, such surfaces are global
minimizers of the Lagrangian geometric flux $\Psi_{t_{0}}^{t_{1}}$. 

We note from \eqref{eq:psi final} that the integrand $\Psi_{t_{0}}^{t_{1}}\left(\mathcal{M}_{0}\right)$
can only vanish pointwise if $\mathcal{M}_{0}$ is everywhere tangent
to $\mathbf{b}_{t_{0}}^{t_{1}}(\mathbf{x}_{0})$. Therefore, diffusive
transport barrier surfaces evolve materially from initial surfaces
to which the temporally averaged pull-back of $\mathbf{h}_{vis}$
is everywhere tangent (see Fig. \ref{fig:qualitative geometry of barrier surfaces}).
\begin{figure}
\centering{}\includegraphics[width=0.5\textwidth]{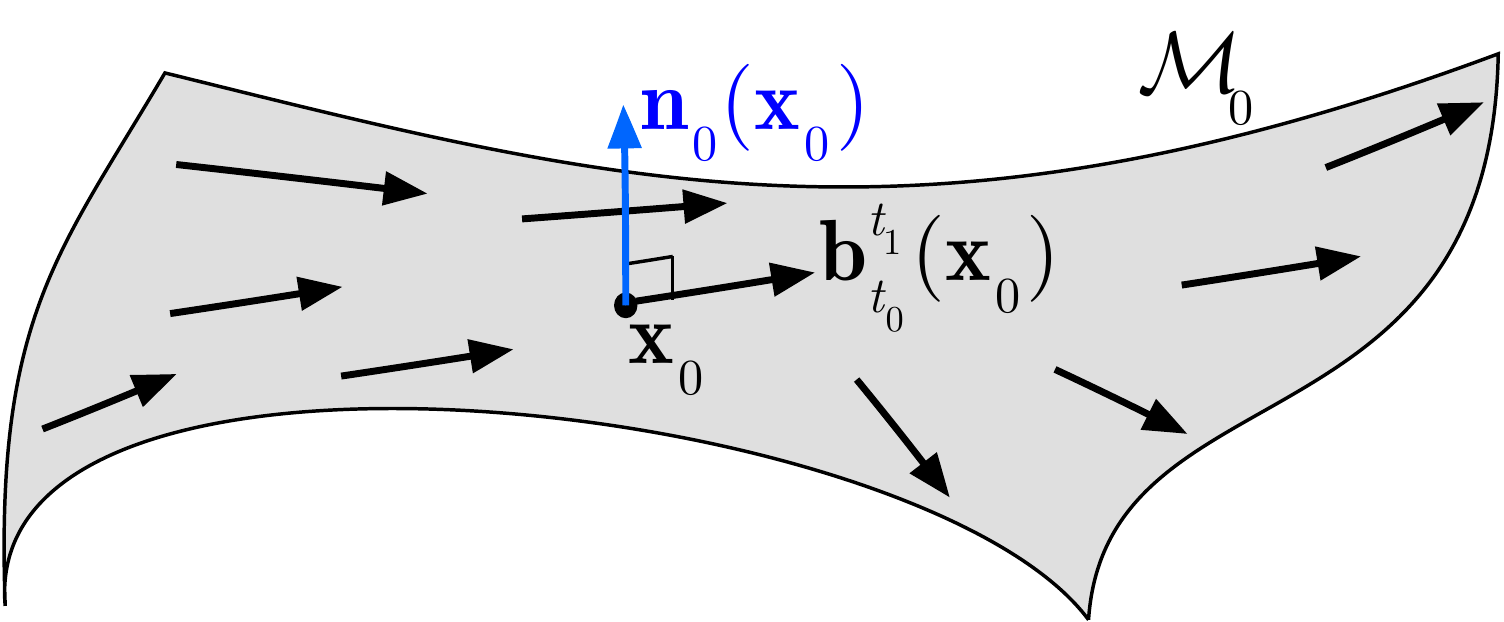}\caption{The normal vector field $\mathbf{n}_{0}(\mathbf{x}_{0})$ of any initial
material barrier $\mathcal{M}_{0}$ must be orthogonal to the barrier
vector field $\mathbf{b}_{t_{0}}^{t_{1}}(\mathbf{x}_{0})$. Therefore,
$\mathcal{M}_{0}$ must be a two-dimensional invariant manifold of
the vector field $\mathbf{b}_{t_{0}}^{t_{1}}=\overline{\det\nabla\mathbf{F}_{t_{0}}^{t}\left(\mathbf{F}_{t_{0}}^{t}\right)^{*}\mathbf{h}_{vis}}$.
\label{fig:qualitative geometry of barrier surfaces}}
\end{figure}
We conclude that if $s\in\mathbb{R}$ parametrizes the streamlines $\mathbf{x}_{0}(s)$ of $\mathbf{b}_{t_{0}}^{t_{1}}(\mathbf{x}_{0})$
and differentiation with respect to $s$ is denoted by a prime, then any 2D
streamsurface (i.e., invariant manifold) of the 3D autonomous differential
equation,
\begin{equation}
\mathbf{x}_{0}^{\prime}=\mathbf{b}_{t_{0}}^{t_{1}}(\mathbf{x}_{0}),\label{eq:barrier equation}
\end{equation}
is a diffusive transport barrier candidate. For this reason, we refer
to eq. \eqref{eq:barrier equation} as the\emph{ barrier equation},
and to $\mathbf{b}_{t_{0}}^{t_{1}}(\mathbf{x}_{0})$ as the corresponding\emph{
barrier vector field.} By the objectivity of the vector field $\mathbf{b}_{t_{0}}^{t_{1}}(\mathbf{x}_{0})$,
the barrier equation \eqref{eq:barrier equation} is objective. Indeed,
after a frame change of the form \eqref{eq:observer change}, we obtain
the transformed barrier equation $\mathbf{Q}(t_{0})\tilde{\mathbf{y}}_{0}^{\prime}=\mathbf{Q}(t_{0})\tilde{\mathbf{b}}_{t_{0}}^{t}(\mathbf{y}_{0})$,
which gives $\tilde{\mathbf{y}}_{0}^{\prime}=\tilde{\mathbf{b}}_{t_{0}}^{t}(\mathbf{y}_{0})$. 

Any smooth curve of initial conditions for the differential equation
\eqref{eq:barrier equation}, however, generates a 2D streamsurface
of trajectories for eq. \eqref{eq:observer change}. Of these infinitely
many barrier candidates, we would like to find only the barrier surfaces
with an observable impact on the transport of $\mathbf{f}$. To this
end, we formally define active transport barriers as follows:
\begin{defn}
\label{def:barrier definition}A \emph{diffusive transport barrier}
for the vector field $\mathbf{f}$ over the time interval $[t_{0},t_{1}]$
is a material surface $\mathcal{B}(t)\subset U$ whose initial position
$\mathcal{B}_{0}=\mathcal{B}(t_{0})$ is a structurally stable (i.e.,
persistent under small, smooth perturbations of \textbf{$\mathbf{u}$}),
2D invariant manifold of the autonomous dynamical system \eqref{eq:barrier equation}. 

The required dimensionality of $\mathcal{B}(t)$ ensures that it divides
locally the space into two 3D regions with minimal diffusive transport
between them. The required structural stability of $\mathcal{B}(t)$
ensures that conclusions reached about transport barriers for one
specific velocity field \textbf{$\mathbf{u}$ }remain valid under
small perturbations of \textbf{$\mathbf{u}$ }as well (see Guckenheimer
\& Holmes 1983).
\end{defn}
While a general classification of structurally stable invariant manifolds
in 3D dynamical systems is not available, structurally stable 2D surfaces
in 3D, steady volume-preserving flows are known to be families of
neutrally stable 2D tori, 2D stable and unstable manifolds of structurally
stable fixed points or of structurally stable periodic orbits (see,
e.g., MacKay 1994). Such structurally stable fixed points and periodic
orbits are either hyperbolic or are contained in no-slip boundaries
and become hyperbolic after a rescaling of time (Surana, Grunberg
\& Haller 2006). In view of these results, the three possible active
barrier geometries for volume-preserving barrier equations in 3D are
shown in Fig. \ref{fig:active  barrier geometries}.
\begin{figure}
\centering{}\includegraphics[width=0.9\textwidth]{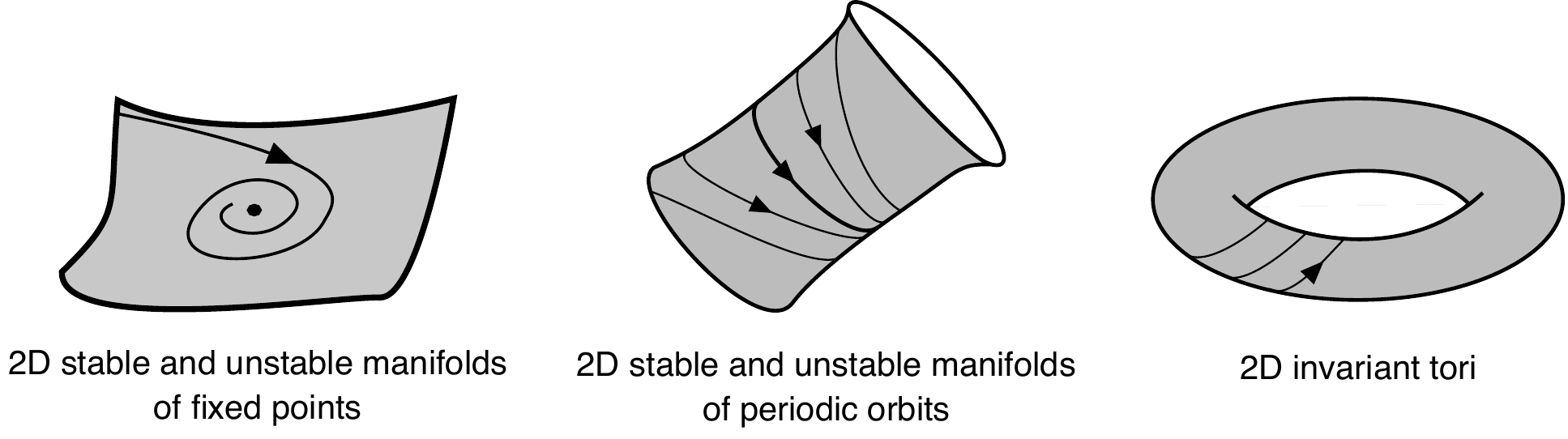}\caption{Possible geometries of material barriers to diffusive transport. Curves
with arrows indicate qualitative sketches of trajectories of the barrier
equation \eqref{eq:barrier equation}, for which these barriers are
structurally stable, two-dimensional invariant manifolds. \label{fig:active  barrier geometries}}
\end{figure}

As we shall see, the barrier-equations for momentum, angular momentum and vorticity are always volume-preserving for incompressible flows and hence the possible active barriers fall in the three categories shown in Fig. \ref{fig:active  barrier geometries}. For compressible flows, the barrier equations are generally not volume-preserving but the three barrier geometries shown in Fig. \ref{fig:active  barrier geometries} nevertheless frequently arise in such flows as well. Invariant tori in compressible barrier equations, however, must necessarily be isolated attractors or repellers, as opposed to members of neutrally stable torus families.

\section{Eulerian active barriers}

Our treatment of active barriers has so far been fundamentally Lagrangian,
targeting material surfaces that render the diffusive transport functional
$\psi_{t_{0}}^{t_{1}}$ zero. Taking the $t_{1}\to t_{0}\equiv t$
limit in our arguments yields that instantaneous diffusive-flux minimizing
surfaces (\emph{Eulerian active barriers}) are structurally
stable, 2D invariant manifolds of the instantaneous barrier equation
\begin{align}
\mathbf{x}^{\prime} & =\mathbf{b}_{t}^{t}(\mathbf{x})=\mathbf{h}_{vis}\left(\mathbf{x},t,\mathbf{u}(\mathbf{x},t),\mathbf{f}(\mathbf{x},t),\mathbf{T}_{vis}(\mathbf{x},t)\right),\label{eq:Eulerian barrier equation}
\end{align}
with $t$ fixed and prime still denoting differentiation with respect
to the dummy parameter $s$. 

The active barriers extracted from \eqref{eq:Eulerian barrier equation}
can be calculated from instantaneous velocity data without Lagrangian
advection, yet they inherit the objectivity of Lagrangian barriers.
These instantaneous barriers, therefore, extend the notion of objective
Eulerian coherent structures (Serra \& Haller 2016) and instantaneous
passive diffusion barriers (Haller, Karrasch \& Kogelbauer 2018, 2019)
to the transport of active vector fields. 

\section{Active barrier equations for momentum and vorticity\label{sec:Active-barrier-equations-derivations}}

We now derive material barrier equations for different active vector
fields. In each case, the instantaneous limits of these equations
can directly be obtained by replacing $\mathbf{F}_{t_{0}}^{t}$ with
the identity map and omitting the averaging operation in time.

\subsection{Barriers to linear momentum transport}

Setting $\mathbf{f}:=\rho\mathbf{u}$, we can rewrite eq. \eqref{eq:main continuum eq. of motion}
as
\begin{equation}
\frac{D\mathbf{f}}{Dt}=\boldsymbol{\nabla}\cdot\mathbf{T}_{vis}-\boldsymbol{\nabla}p+\mathbf{q}-\frac{D\rho}{Dt}\mathbf{u},\label{eq:u equation for linear momentum}
\end{equation}
and hence obtain
\begin{equation}
\mathbf{h}_{vis}=\boldsymbol{\nabla}\cdot\mathbf{T}_{vis},\qquad\mathbf{h}_{nonvis}=-\boldsymbol{\nabla}p+\mathbf{q}-\frac{D\rho}{Dt}\mathbf{u},\label{eq: internal-external partition for linear momentum}
\end{equation}
for the viscous and non-viscous terms in \eqref{eq:DuDt}. The viscous
stress tensor and its divergence are objective (Gurtin, Fried \& Anand
2013), and hence the $\mathbf{h}_{vis}$ function in \eqref{eq: internal-external partition for linear momentum}
satisfies the objectivity condition \eqref{eq:objectivity condition}.
Accordingly, the barrier equations \eqref{eq:barrier equation} and
\eqref{eq:Eulerian barrier equation} for the diffusive transport
of linear momentum become
\begin{align}
\mathbf{x}_{0}^{\prime} & =\overline{\det\nabla\mathbf{F}_{t_{0}}^{t}\left(\mathbf{F}_{t_{0}}^{t}\right)^{*}\left[\boldsymbol{\nabla}\cdot\mathbf{T}_{vis}\right]},\label{eq:NS Lagrangian momentum barrier eq}\\
\mathbf{x}^{\prime} & =\boldsymbol{\nabla}\cdot\mathbf{T}_{vis}.\label{eq:NS Eulerian momentum barrier eq}
\end{align}

Specifically, in the case of incompressible Navier\textendash Stokes
flows with kinematic viscosity $\nu$, we have the constitutive law
$\boldsymbol{\nabla}\cdot\mathbf{T}_{vis}=\nu\rho\Delta\mathbf{u}$
in the general momentum equation \eqref{eq:u equation for linear momentum};
we also observe that $\det\nabla\mathbf{F}_{t_{0}}^{t}\equiv1$ holds
by incompressibility. We then obtain the following: 
\begin{thm}
\label{thm: 3D incompressible N-S,  linear momentum barrier eq }For
incompressible, uniform-density Navier\textendash Stokes flows, the
material and instantaneous barrier equations \eqref{eq:NS Lagrangian momentum barrier eq}
and \eqref{eq:NS Eulerian momentum barrier eq} for linear momentum
take the specific forms
\begin{align}
\mathbf{x}_{0}^{\prime} & =\nu\rho\,\overline{\left(\mathbf{F}_{t_{0}}^{t}\right)^{*}\Delta\mathbf{u}},\label{eq:incompressible NS Lagrangian momentum barrier eq}\\
\mathbf{x}^{\prime} & =\nu\rho\,\Delta\mathbf{u}.\label{eq:incompressible NS Eulerian momentum barrier eq}
\end{align}
 
\end{thm}
Each of the eqs. \eqref{eq:incompressible NS Lagrangian momentum barrier eq}-\eqref{eq:incompressible NS Eulerian momentum barrier eq}
defines a 3D, autonomous (or steady) dynamical system with respect
to the time-like variable\textbf{ }$s\in\mathbb{R}$, and hence can
be analyzed via tools developed for steady flows in the chaotic advection
literature (Aref et al. 2017). For the purpose of finding active transport
barriers, all the relevant information about the unsteadiness of \textbf{$\mathbf{u}\left(\mathbf{x},t\right)$
}over the time interval $\left[t_{0},t_{1}\right]$ is encoded into
eq. \eqref{eq:incompressible NS Lagrangian momentum barrier eq} through
the pull-back and the temporal averaging operations. The instantaneous
version \eqref{eq:incompressible NS Eulerian momentum barrier eq}
of these equations only contains the physical time $t$ as a parameter;
it is, therefore, also a steady ODE with respect to the variable $s$
parametrizing its streamlines. Both dynamical systems in \eqref{eq:incompressible NS Lagrangian momentum barrier eq}-\eqref{eq:incompressible NS Eulerian momentum barrier eq}
are volume-preserving because $\Delta\mathbf{u}$ is divergence-free
for incompressible flows. Therefore, the three possible active barrier
geometries arising from the analysis of these barrier equations are
those shown in Fig. \ref{fig:active  barrier geometries}.

In order to solve for trajectories of eq. \eqref{eq:incompressible NS Lagrangian momentum barrier eq}
accurately over a domain $U$ with boundary $\partial U$, one must
be aware of any special boundary condition that $\mathbf{b}_{t_{0}}^{t_{1}}(\mathbf{x}_{0})$
may have to satisfy along $\partial U$. We assume for simplicity
that\textbf{ $\mathbf{u}(\mathbf{x},t)$ }is incompressible and $\partial U$
is a no-slip boundary. Then, after projecting the Navier\textendash Stokes
equation \eqref{eq:main continuum eq. of motion} at a point $\mathbf{x}\in\partial U$
onto a local orthogonal basis $(\mathbf{e}_{1},\mathbf{e}_{2},\mathbf{e}_{3})$,
with $\mathbf{e}_{3}$ normal to the wall, we obtain
\begin{equation}
\left(\begin{array}{c}
0\\
0\\
0
\end{array}\right)=\nu\left(\begin{array}{c}
0\\
0\\
\Delta\mathbf{u}\cdot\mathbf{e}_{3}
\end{array}\right)+\left(\begin{array}{c}
\left(\mathbf{q}-\frac{1}{\rho}\boldsymbol{\nabla}p\right)\cdot\mathbf{e}_{1}\\
\left(\mathbf{q}-\frac{1}{\rho}\boldsymbol{\nabla}p\right)\cdot\mathbf{e}_{2}\\
\left(\mathbf{q}-\frac{1}{\rho}\boldsymbol{\nabla}p\right)\cdot\mathbf{e}_{3}
\end{array}\right).\label{eq:BC for momentum barriers}
\end{equation}
Therefore, if the wall-normal pressure gradient balances out the external
body forces along $\partial U$ (as is often assumed in CFD simulations),
then $\Delta\mathbf{u}$ satisfies a no-penetration boundary condition
along the no-slip boundary $\partial U$, because $\Delta\mathbf{u}\cdot\mathbf{e}_{3}$
must vanish at each boundary point by \eqref{eq:BC for momentum barriers}.
Given that such a boundary $\partial U$ is invariant under the flow
map $\mathbf{F}_{t_{0}}^{t}$, we obtain that the pull-back of $\Delta\mathbf{u}$
under the flow map must also be tangent to the boundary. Consequently,
any no-slip boundary $\partial U$ with a vanishing boundary-normal
resultant force is an invariant manifold for the barrier equations
\eqref{eq:incompressible NS Lagrangian momentum barrier eq}-\eqref{eq:incompressible NS Eulerian momentum barrier eq}\emph{.}\textbf{\emph{ }}

\subsection{Barriers to angular momentum transport}

To analyze angular momentum barriers, we take the cross product of
eq. \eqref{eq:main continuum eq. of motion} with a vector $\mathbf{r}=\mathbf{x-\hat{x}}$,
where $\mathbf{\hat{x}}\in U$ marks a fixed reference point. Setting
then $\mathbf{f}:=\mathbf{r}\times\rho\mathbf{u},$ we obtain an evolution
equation for $\mathbf{f}$ in the form
\begin{equation}
\frac{D\mathbf{f}}{Dt}=\left(\mathbf{x-\hat{x}}\right)\times\frac{D\rho}{Dt}\mathbf{u}-\left(\mathbf{x-\hat{x}}\right)\times\boldsymbol{\nabla}p+\left(\mathbf{x-\hat{x}}\right)\times\mathbf{\mathbf{q}}+\left(\mathbf{x-\hat{x}}\right)\times\boldsymbol{\nabla}\cdot\mathbf{T}_{vis},\label{eq:angular momentum evolution eq}
\end{equation}
implying 
\begin{equation}
\mathbf{h}_{vis}=\left(\mathbf{x-\hat{x}}\right)\times\boldsymbol{\nabla}\cdot\mathbf{T}_{vis},\qquad\mathbf{h}_{nonvis}=\left(\mathbf{x-\hat{x}}\right)\times\left[-\boldsymbol{\nabla}p+\mathbf{q}+\frac{D\rho}{Dt}\mathbf{u}\right],\label{eq: internal-external partition for angular momentum}
\end{equation}
for the viscous and non-viscous terms in \eqref{eq:DuDt}. Under a
frame-change of the form \eqref{eq:observer change}, this $\mathbf{h}_{vis}$
satisfies 
\begin{equation}
\mathbf{h}_{vis}=\left(\mathbf{x-\hat{x}}\right)\times\boldsymbol{\nabla}\cdot\mathbf{T}_{vis}=\mathbf{Q}(t)\left(\mathbf{y-\hat{y}}\right)\times\mathbf{Q}(t)\tilde{\boldsymbol{\nabla}}\cdot\tilde{\mathbf{T}}_{vis}=\left(\mathbf{y-\hat{y}}\right)\times\tilde{\boldsymbol{\nabla}}\cdot\tilde{\mathbf{T}}_{vis}=\tilde{\mathbf{h}}_{vis},\label{eq:objectivity for angular momentum barriers}
\end{equation}
where we have used the objectivity of $\boldsymbol{\nabla}\cdot\mathbf{T}_{vis}$.
We conclude from \eqref{eq:objectivity for angular momentum barriers}
that the objectivity condition \eqref{eq:objectivity condition} is
satisfied for this choice of $\mathbf{f}$, and hence our formulation
is applicable. We, therefore, obtain, as in the case of linear momentum,
the following result:
\begin{thm}
\label{thm: 3D incompressible N-S,  angular momentum barrier eq }For
incompressible, uniform-density Navier\textendash Stokes flows, the
material and instantaneous barrier equations \eqref{eq:barrier equation}
and \eqref{eq:Eulerian barrier equation} for angular momentum take
the specific form
\end{thm}
\begin{align}
\mathbf{x}_{0}^{\prime} & =\nu\rho\,\overline{\left(\mathbf{F}_{t_{0}}^{t}\right)^{*}\left[\left(\mathbf{x-\hat{x}}\right)\times\Delta\mathbf{u}\right]},\label{eq:incompressible NS Lagrangian angular momentum barrier eq}\\
\mathbf{x}^{\prime} & =\nu\rho\,\left(\mathbf{x-\hat{x}}\right)\times\Delta\mathbf{u}.\label{eq:incompressible NS Eulerian angular momentum barrier eq}
\end{align}
These equations again define 3D, steady, volume-preserving dynamical
systems with respect to the time-like independent variable\textbf{
}$s\in\mathbb{R}$. As in the case of barriers to the transport of
linear momentum, we find that in the presence of zero boundary-normal
resultant force, eq. \eqref{eq:angular momentum evolution eq} implies
any no-slip boundary $\partial U$ to be an invariant manifold for
the two dynamical systems in \eqref{eq:incompressible NS Lagrangian angular momentum barrier eq}-\eqref{eq:incompressible NS Eulerian angular momentum barrier eq}.

\subsection{Barriers to vorticity transport }

To obtain the evolution equation for the active vector field $\mathbf{f}:=\boldsymbol{\omega}$,
we divide eq. \eqref{eq:main continuum eq. of motion} by $\rho$,
take the curl of both sides and use the relation $\boldsymbol{\nabla}\times\boldsymbol{\nabla}p=\mathbf{0}$
to obtain the general vorticity transport equation
\begin{equation}
\frac{D\mathbf{f}}{Dt}=\left(\boldsymbol{\nabla}\mathbf{u}\right)\mathbf{f}-\left(\boldsymbol{\nabla}\cdot\mathbf{u}\right)\mathbf{f}+\frac{1}{\rho^{2}}\boldsymbol{\nabla}\rho\times\boldsymbol{\nabla}p+\boldsymbol{\nabla}\times\left(\frac{1}{\rho}\mathbf{\mathbf{q}}\right)+\nu\boldsymbol{\nabla}\times\left(\frac{1}{\rho}\boldsymbol{\nabla}\cdot\mathbf{T}_{vis}\right).\label{eq:vorticity tranport eq.-1}
\end{equation}
Consequently, our general formulation \eqref{eq:DuDt} applies with
\begin{equation}
\mathbf{h}_{vis}=\nu\boldsymbol{\nabla}\times\left(\frac{1}{\rho}\boldsymbol{\nabla}\cdot\mathbf{T}_{vis}\right),\qquad\mathbf{h}_{nonvis}=\left(\boldsymbol{\nabla}\mathbf{u}\right)\mathbf{f}-\left(\boldsymbol{\nabla}\cdot\mathbf{u}\right)\mathbf{f}+\frac{1}{\rho^{2}}\boldsymbol{\nabla}\rho\times\boldsymbol{\nabla}p+\boldsymbol{\nabla}\times\left(\frac{1}{\rho}\mathbf{\mathbf{q}}\right).
\end{equation}
Following the derivation of the transformation formula for vorticity
under an observer change \eqref{eq:observer change} (see, e.g., Truesdell
\& Rajagopal 2009), we obtain that $\mathbf{h}_{vis}=\mathbf{Q}(t)\tilde{\mathbf{h}}_{vis}$.
Therefore, the objectivity condition \eqref{eq:objectivity condition}
is satisfied for this choice of $\mathbf{f}$, and hence our formulation
is applicable. The barrier equations \eqref{eq:barrier equation}
and \eqref{eq:Eulerian barrier equation} for diffusive vorticity
transport then become
\begin{align}
\mathbf{x}_{0}^{\prime} & =\nu\,\overline{\det\nabla\mathbf{F}_{t_{0}}^{t}\left(\mathbf{F}_{t_{0}}^{t}\right)^{*}\left[\boldsymbol{\nabla}\times\left(\frac{1}{\rho}\boldsymbol{\nabla}\cdot\mathbf{T}_{vis}\right)\right]},\label{eq:NS Lagrangian vorticity barrier eq}\\
\mathbf{x}^{\prime} & =\nu\,\boldsymbol{\nabla}\times\left(\frac{\boldsymbol{\nabla}\cdot\mathbf{T}_{vis}}{\rho}\right).\label{eq:NS Eulerian vorticity barrier eq}
\end{align}
Specifically, as in the case of linear and angular momentum barriers,
we obtain:
\begin{thm}
\label{thm: 3D incompressible N-S,  vorticity barrier eq}For incompressible,
uniform-density Navier\textendash Stokes flows, the material and instantaneous
barrier equations \eqref{eq:NS Lagrangian momentum barrier eq} and
\eqref{eq:NS Eulerian vorticity barrier eq} for vorticity take the
specific form
\begin{align}
\mathbf{x}_{0}^{\prime} & =\nu\,\overline{\left(\mathbf{F}_{t_{0}}^{t}\right)^{*}\Delta\boldsymbol{\omega}},\label{eq:incompressible NS Lagrangian vorticity barrier eq}\\
\mathbf{x}^{\prime} & =\nu\,\Delta\boldsymbol{\omega}.\label{eq:incompressible NS Eulerian vorticity barrier eq}
\end{align}
\end{thm}
As in the case of the linear and angular momenta, the active barrier
equations \eqref{eq:incompressible NS Lagrangian vorticity barrier eq}-\eqref{eq:incompressible NS Lagrangian momentum barrier eq}
define 3D, autonomous, volume-preserving dynamical systems with respect
to the time-like, evolutionary variable\textbf{ }$s\in\mathbb{R}$,
and hence can be analyzed by adopting tools available such equations
(see section \ref{sec:Practical-implementation-on}).

As for boundary conditions for trajectories of the equations \eqref{eq:incompressible NS Lagrangian vorticity barrier eq}
along a no-slip boundary $\partial U$ in the incompressible case
with $\rho_{0}(\mathbf{x})\equiv1$, the vorticity-transport equation
along the wall $\partial U$ takes the form
\begin{equation}
\left(\begin{array}{c}
\left(\frac{D}{Dt}\boldsymbol{\omega}\right)\cdot\mathbf{e}_{1}\\
\left(\frac{D}{Dt}\boldsymbol{\omega}\right)\cdot\mathbf{e}_{2}\\
0
\end{array}\right)=\nu\left(\begin{array}{c}
\boldsymbol{\Delta\omega}\cdot\mathbf{e}_{1}\\
\boldsymbol{\Delta\omega}\cdot\mathbf{e}_{2}\\
\boldsymbol{\Delta\omega}\cdot\mathbf{e}_{3}
\end{array}\right)\Delta\boldsymbol{\omega}+\left(\begin{array}{c}
\boldsymbol{\nabla\times q}\cdot\mathbf{e}_{1}\\
\boldsymbol{\nabla\times q}\cdot\mathbf{e}_{2}\\
\boldsymbol{\nabla\times q}\cdot\mathbf{e}_{3}
\end{array}\right),\label{eq:BC for omega}
\end{equation}
with the vectors $\mathbf{e}_{i}$ defined as in formula \eqref{eq:BC for momentum barriers}.
Consequently, whenever the curl of non-potential body forces is normal
to a no-slip boundary $\partial U$, the vector field $\Delta\boldsymbol{\omega}$
satisfies a no-penetration boundary condition along $\partial U$,
given that $\boldsymbol{\Delta\omega}\cdot\mathbf{e}_{3}$
must then vanish by \eqref{eq:BC for omega}. As we have already noted
in relation to formula \eqref{eq:BC for momentum barriers}, this
in turn implies that $\partial U$ is an invariant manifold for the
two flows in eqs. \eqref{eq:incompressible NS Lagrangian vorticity barrier eq}-\eqref{eq:incompressible NS Eulerian vorticity barrier eq}.

\section{Active transport barriers in special classes of flows\label{sec:Active-transport-barriers-in-special-classes-of-flows}}

In order to illustrate the feasibility of the active barriers we have
constructed, we now identify them in classes of explicit Navier\textendash Stokes
solutions, with the details of the calculations relegated to Appendices
C and D. 

\subsection{2D Navier\textendash Stokes flows viewed as 3D Navier\textendash Stokes
	flows with symmetry \label{subsec:active barriers in 2D-Navier=002013Stokes-flows}}

We define the planar variable $\boldsymbol{\hat{\mathbf{x}}}=\left(x_{1},x_{2}\right)\in\mathbb{R}^{2}$
and assume that a solution of the 3D incompressible Navier\textendash Stokes
equation is of the form
\begin{equation}
\mathbf{u}(\mathbf{x},t)=\left(\hat{\mathbf{u}}(\hat{\mathbf{x}},t),w(\hat{\mathbf{x}},t)\right),\quad p(\mathbf{x},t)=p(\hat{\mathbf{x}},t),\qquad\mathbf{x}=\left(\boldsymbol{\hat{\mathbf{x}}},x_{3}\right)\in\mathbb{R}^{3},\label{eq:2D ansatz}
\end{equation}
with the two-dimensional velocity field $\hat{\mathbf{u}}(\hat{\mathbf{x}},t)$
and the scalar functions $w(\hat{\mathbf{x}},t)$ and $p(\hat{\mathbf{x}},t)$
(see, e.g., Majda \& Bertozzi 2002). Under this 2D-symmetry ansatz,
substitution of \textbf{$\mathbf{u}$ }and\textbf{ $p$ }into the
3D Navier\textendash Stokes equation gives 
\begin{align}
\partial_{t}\mathbf{\hat{\mathbf{u}}}+\left(\boldsymbol{\nabla}_{\mathbf{\hat{\mathbf{x}}}}\hat{\mathbf{u}}\right)\hat{\mathbf{u}} & =-\frac{1}{\rho}\boldsymbol{\nabla}_{\mathbf{\hat{\mathbf{x}}}}p+\nu\Delta_{\hat{\mathbf{x}}}\hat{\mathbf{u}},\label{eq:3D Navier Stokes for symmetric flows-uv}\\
\partial_{t}w+\boldsymbol{\nabla}_{\mathbf{\hat{\mathbf{x}}}}w\cdot\hat{\mathbf{u}} & =\nu\Delta_{\hat{\mathbf{x}}}w,\label{eq:3D Navier Stokes for symmetric flows-w}
\end{align}
with the subscript $\hat{\mathbf{x}}$ referring to the 2D version
of the differential operators involved. Therefore, the symmetry ansatz
\eqref{eq:2D ansatz} for a 3D Navier-Stokes solution is valid if
$w(\hat{\mathbf{x}},t)$ is chosen as a solution of the advection-diffusion
equation appearing in \eqref{eq:3D Navier Stokes for symmetric flows-w}.
This advection-diffusion equation, however, coincides with the 2D
vorticity transport equation, which is solved by
\begin{equation}
w(\boldsymbol{\hat{\mathbf{x}}},t)=\hat{\omega}(\boldsymbol{\hat{\mathbf{x}}},t),\label{eq:w is the vorticity}
\end{equation}
with $\hat{\omega}(\boldsymbol{\hat{\mathbf{x}}},t)$ denoting the
scalar vorticity field of the 2D Navier\textendash Stokes solution
$\hat{\mathbf{u}}(\boldsymbol{\hat{\mathbf{x}}},t)$. In the following,
we will choose the third component of $\mathbf{u}$ as in eq. \eqref{eq:w is the vorticity}
and use the notation
\begin{equation}
\mathbf{J}=\left(\begin{array}{rc}
0 & 1\\
-1 & 0
\end{array}\right)
\end{equation}
for the two-dimensional canonical symplectic matrix $\mathbf{J}$.
With this notation, we obtain the following results on active barriers
to momentum transport in eq. \eqref{eq:2D ansatz}.
\begin{thm}
	\label{thm:2D momentum barriers}For 2D incompressible, uniform-density
	Navier\textendash Stokes flows, the material and instantaneous barrier
	equations \eqref{eq:incompressible NS Lagrangian momentum barrier eq}
	and \eqref{eq:incompressible NS Eulerian momentum barrier eq} for
	linear momentum are autonomous Hamiltonian systems of the form
	\begin{align}
	\hat{\mathbf{x}}_{0}^{\prime} & =\nu\rho\,\mathbf{J}\boldsymbol{\nabla}_{0}\,\overline{\hat{\omega}\left(\hat{\mathbf{F}}_{t_{0}}^{t}\left(\hat{\mathbf{x}}_{0}\right),t\right)},\label{eq:2D incompressible NS Lagrangian momentum barrier eq}\\
	\hat{\mathbf{x}}^{\prime} & =\nu\rho\,\mathbf{J}\boldsymbol{\nabla}\,\hat{\omega}\left(\hat{\mathbf{x}},t\right),\label{eq:2D incompressible NS Eulerian momentum barrier eq}
	\end{align}
	respectively. Therefore, time-$t_{0}$ positions of material active
	barriers to linear momentum transport in these flows are structurally
	stable level curves of the time-averaged Lagrangian vorticity $\overline{\hat{\omega}\left(\hat{\mathbf{F}}_{t_{0}}^{t}\left(\mathbf{\hat{x}}_{0}\right),t\right)}$
	viewed as a Hamiltonian. Similarly, instantaneous active barriers
	to linear momentum transport at time $t$ are structurally stable
	level curves of the vorticity $\hat{\omega}(\boldsymbol{\hat{\mathbf{x}}},t)$.
\end{thm}
\begin{proof}
	See Appendix C.
\end{proof}
While streamlines in general are not objective, the streamlines of
the vorticity $\hat{\omega}(\boldsymbol{\hat{\mathbf{x}}},t)$ are
Eulerian-objective and streamlines of the time-averaged Lagrangian
vorticity $\overline{\hat{\omega}\left(\hat{\mathbf{F}}_{t_{0}}^{t}\left(\mathbf{\hat{x}}_{0}\right),t\right)}$
are Lagrangian-objective (see Ogden 1984). This is consistent with
the more general result established in eq. \eqref{eq:objectivity of b field}
for the objectivity of all active barriers. 

Active barriers to vorticity transport in \eqref{eq:2D ansatz} also
turn out to be trajectories of autonomous Hamiltonian systems. To
state this result, we will use the notation
\begin{equation}
\delta\hat{\omega}\left(\hat{\mathbf{x}}_{0},t_{0},t_{1}\right):=\hat{\omega}\left(\hat{\mathbf{F}}_{t_{0}}^{t_{1}}\left(\hat{\mathbf{x}}_{0}\right),t_{1}\right)-\hat{\omega}\left(\hat{\mathbf{x}}_{0},t_{0}\right)\label{eq:deltaomegahat}
\end{equation}
for the Lagrangian vorticity-change function along trajectories over
the time interval $[t_{0},t_{1}]$.
\begin{thm}
	\label{thm:2D vorticity barriers}For 2D incompressible, uniform-density
	Navier\textendash Stokes flows, the material and instantaneous barrier
	equations \eqref{eq:incompressible NS Lagrangian vorticity barrier eq}
	and \eqref{eq:incompressible NS Eulerian vorticity barrier eq} for
	linear momentum are autonomous Hamiltonian systems of the form
	\begin{align}
	\hat{\mathbf{x}}_{0}^{\prime} & =\frac{\nu}{t_{1}-t_{0}}\,\mathbf{J}\boldsymbol{\nabla}_{0}\,\delta\hat{\omega}\left(\hat{\mathbf{x}}_{0},t_{0},t_{1}\right),\label{eq:2D incompressible NS Lagrangian vorticity barrier eq}\\
	\hat{\mathbf{x}}^{\prime} & =\nu\,\mathbf{J}\boldsymbol{\nabla}\frac{D}{Dt}\hat{\omega}\left(\hat{\mathbf{x}},t\right),\label{eq:2D incompressible NS Eulerian vorticity barrier eq}
	\end{align}
	respectively. Therefore, time-$t_{0}$ positions of material active
	barriers to linear momentum transport in these flows are structurally
	stable level curves of the Lagrangian vorticity-change function $\delta\hat{\omega}\left(\hat{\mathbf{x}}_{0},t_{0},t_{1}\right)$
	viewed as a Hamiltonian. Similarly, instantaneous active barriers
	to linear momentum transport at time $t$ are structurally stable
	level curves of the material derivative $\frac{D}{Dt}\hat{\omega}\left(\hat{\mathbf{x}},t\right)$,
	or equivalently, of the vorticity Laplacian $\Delta\hat{\omega}\left(\hat{\mathbf{x}},t\right).$
\end{thm}
\begin{proof}
	See Appendix C.
\end{proof}
While vorticity is not objective, the level curves of the Lagrangian
vorticity change $\delta\hat{\omega}\left(\hat{\mathbf{x}}_{0},t_{0},t_{1}\right)$
is objective. This follows directly from the objectivity of the barrier
equations that we have generally established, but can also be verified
directly using the definition of objectivity.
\begin{rem}
	\label{rem:used active LCS diagnostics instead of level curves}By
	Theorems \ref{thm:2D momentum barriers}-\ref{thm:2D vorticity barriers},
	outermost members of nested families of closed level curves of $\overline{\hat{\omega}\left(\hat{\mathbf{F}}_{t_{0}}^{t}\left(\mathbf{\hat{x}}_{0}\right),t\right)}$
	or $\delta\hat{\omega}\left(\hat{\mathbf{x}}_{0},t_{0},t_{1}\right)$
	can be used to define \emph{coherent material vortex boundaries. }These
	are constructed as maximal barriers to momentum or vorticity transport,
	depending on whether one isolates coherent vortices based on their
	role in momentum- or vorticity-transport, respectively. Similarly,
	to locate instantaneous Eulerian vortex boundaries, one identifies
	outermost members of nested families of closed level curves of $\hat{\omega}(\boldsymbol{\hat{\mathbf{x}}},t)$
	or $\frac{D}{Dt}\hat{\omega}\left(\hat{\mathbf{x}},t\right)$, respectively.
	These outermost contours give a clear conceptual meaning to vortex
	boundaries from an active transport perspective, but their identification
	from numerical data tends to be a sensitive process. Instead, active-transport-minimizing
	material and instantaneous vortex boundaries can simply be visualized
	via LCS-detection tools adopted to their appropriate 2D, steady barrier
	equations \eqref{eq:2D incompressible NS Lagrangian momentum barrier eq}-\eqref{eq:2D incompressible NS Eulerian momentum barrier eq}
	and \eqref{eq:2D incompressible NS Lagrangian vorticity barrier eq}-\eqref{eq:2D incompressible NS Eulerian vorticity barrier eq}
	(see section \ref{sec:Practical-implementation-on}).
\end{rem}
\begin{example}
	\label{ex:2D analytic example}We consider the spatially doubly-periodic
	Navier\textendash Stokes flow family described by Majda \& Bertozzi
	(2002) in the form
	\begin{align}
	\hat{\mathbf{u}}(\hat{\mathbf{x}},t) & =e^{-4\pi^{2}\ell\nu t}\hat{\mathbf{u}}_{0}(\hat{\mathbf{x}}),\quad p(\hat{\mathbf{x}},t)=e^{-4\pi^{2}\ell\nu t}p_{0}(\hat{\mathbf{x}}),\label{eq:2D flow family}\\
	\hat{\mathbf{u}}_{0}(\hat{\mathbf{x}}) & =\sum_{\left|\mathbf{k}\right|^{2}=\ell}\left(\begin{array}{c}
	a_{\mathbf{k}}k_{2}\sin\left(2\pi\mathbf{k}\cdot\hat{\mathbf{x}}\right)-b_{\mathbf{k}}k_{2}\cos\left(2\pi\mathbf{k}\cdot\hat{\mathbf{x}}\right)\\
	-a_{\mathbf{k}}k_{1}\sin\left(2\pi\mathbf{k}\cdot\hat{\mathbf{x}}\right)+b_{\mathbf{k}}k_{1}\cos\left(2\pi\mathbf{k}\cdot\hat{\mathbf{x}}\right)
	\end{array}\right),\nonumber 
	\end{align}
	where $\hat{\mathbf{u}}_{0}(\hat{\mathbf{x}})$ and $p_{0}(\hat{\mathbf{x}})$
	solve the steady planar Euler equation for some positive integer $\ell$.\footnote{This flow family contains our motivating example \eqref{eq:horizontal shear jet}
		in Appendix A with the choice $k_{1}=0$, $\ell=k_{2}=1$, $a_{(1,0)}=b_{(1,0)}=a_{(0,1)}=0$
		and $b_{(0,1)}=a$ if we let $x_{2}\to-x_{2}$. } In that case, we have 
	\begin{align}
	\Delta\mathbf{u} & =\left(\begin{array}{c}
	\Delta_{\hat{\mathbf{x}}}\hat{\mathbf{u}}\\
	\Delta_{\hat{\mathbf{x}}}\hat{\omega}
	\end{array}\right)=\left(\begin{array}{c}
	-4\pi^{2}\ell e^{-4\pi^{2}\ell\nu t}\hat{\mathbf{u}}_{0}(\hat{\mathbf{x}})\\
	\Delta_{\hat{\mathbf{x}}}\hat{\omega}(\hat{\mathbf{x}},t)
	\end{array}\right).
	\end{align}
	
	One can verify by direct substitution that $e^{-4\pi^{2}\ell\nu t}\hat{\mathbf{u}}_{0}\left(\hat{\mathbf{F}}_{t_{0}}^{t}\left(\mathbf{\hat{\mathbf{x}}}_{0}\right)\right)$
	is a solution of the equation of variations $\dot{\boldsymbol{\xi}}=e^{-4\pi^{2}\ell\nu t}\boldsymbol{\nabla}_{\mathbf{\hat{\mathbf{x}}}}\hat{\mathbf{u}}_{0}(\hat{\mathbf{x}}(t))\boldsymbol{\xi}$
	(whose fundamental matrix solution is $\nabla_{\hat{\mathbf{x}}_{0}}\mathbf{\hat{F}}_{t_{0}}^{t}\left(\mathbf{x}_{0}\right)$)
	for the differential equation $\dot{\mathbf{x}}=e^{-4\pi^{2}\ell\nu t}\hat{\mathbf{u}}_{0}(\hat{\mathbf{x}})$.
	As a consequence, we have 
	\[
	\left[\boldsymbol{\nabla}_{\mathbf{\hat{\mathbf{x}}}_{0}}\hat{\mathbf{F}}_{t_{0}}^{t}\left(\mathbf{\hat{\mathbf{x}}}_{0}\right)\right]^{-1}e^{-4\pi^{2}\ell\nu t}\hat{\mathbf{u}}_{0}\left(\hat{\mathbf{F}}_{t_{0}}^{t}\left(\mathbf{\hat{\mathbf{x}}}_{0}\right)\right)=e^{-4\pi^{2}\ell\nu t_{0}}\hat{\mathbf{u}}_{0}\left(\mathbf{\hat{\mathbf{x}}}_{0}\right),
	\]
	and hence, by Theorem \ref{thm:2D momentum barriers}, the material
	and instantaneous barrier equations for linear momentum take the specific
	form
	\begin{align}
	\hat{\mathbf{x}}_{0}^{\prime} & =\nu\rho e^{-4\pi^{2}\ell\nu t_{0}}\hat{\mathbf{u}}_{0}\left(\mathbf{\hat{\mathbf{x}}}_{0}\right),\nonumber \\
	x_{03}^{\prime} & =\nu\rho A(\mathbf{\hat{\mathbf{x}}}_{0},t_{1},t_{0}),\label{eq:momentum barrier eq for 2D flow-separable case}\\
	\hat{\mathbf{x}}^{\prime} & =\nu\rho e^{-4\pi^{2}\ell\nu t}\hat{\mathbf{u}}_{0}\left(\mathbf{\hat{\mathbf{x}}}\right),\nonumber \\
	x_{3}^{\prime} & =\nu\rho A(\mathbf{\hat{\mathbf{x}}},t,t),\nonumber 
	\end{align}
	for an appropriate function $A(\mathbf{\hat{\mathbf{x}}}_{0},t_{1},t_{0})$.
	Therefore, both material and instantaneous barriers to linear momentum
	transport in the 2D Navier\textendash Stokes flow family in eq. \eqref{eq:2D flow family}
	are structurally stable streamlines of the steady velocity field $\hat{\mathbf{u}}_{0}\left(\mathbf{\hat{\mathbf{x}}}_{0}\right)$. 
	
	As for vorticity barriers in this example, note that
	\begin{align}
	\hat{\omega} & =\partial_{x_{1}}u_{2}-\partial_{x_{2}}u_{1}=-2\pi\ell e^{-4\pi^{2}\ell\nu t}\hat{\omega}_{0},\\
	\hat{\omega}_{0}\left(\mathbf{\hat{\mathbf{x}}}\right) & =\sum_{\left|\mathbf{k}\right|^{2}=\ell}a_{\mathbf{k}}\cos\left(2\pi\mathbf{k}\cdot\hat{\mathbf{x}}\right)+b_{\mathbf{k}}\sin\left(2\pi\mathbf{k}\cdot\hat{\mathbf{x}}\right).
	\end{align}
	As the steady part of the vorticity field solves the steady planar
	Euler equation, trajectories of $\hat{\mathbf{u}}(\hat{\mathbf{x}},t)$
	remain confined to the steady streamlines of $\hat{\mathbf{u}}_{0}\left(\mathbf{\hat{\mathbf{x}}}_{0}\right)$.
	Since these trajectories also conserve the vorticity $\hat{\omega}_{0}$
	of the inviscid limit of the flow, the change in vorticity $\hat{\omega}\left(\hat{\mathbf{x}},t\right)$
	along trajectories of $\hat{\mathbf{u}}(\hat{\mathbf{x}},t)$ can
	be written as
	\begin{equation}
	\delta\hat{\omega}\left(\hat{\mathbf{x}}_{0},t_{0},t_{1}\right)=-2\pi\ell\left(e^{-4\pi^{2}\ell\nu t_{1}}-e^{-4\pi^{2}\ell\nu t_{0}}\right)\hat{\omega}_{0}\left(\mathbf{\hat{\mathbf{x}}}_{0}\right).
	\end{equation}
	Therefore, level curves of the vorticity change along trajectories
	coincide with those of the inviscid vorticity $\hat{\omega}_{0}\left(\mathbf{\hat{\mathbf{x}}}_{0}\right)$,
	which are in turn just the streamlines of $\hat{\mathbf{u}}_{0}\left(\mathbf{\hat{\mathbf{x}}}_{0}\right)$.
	Finally, we have 
	\begin{equation}
	\Delta\hat{\omega}\left(\hat{\mathbf{x}},t\right)=8\pi^{3}\ell^{2}e^{-4\pi^{2}\ell\nu t}\hat{\omega}_{0}\left(\mathbf{\hat{\mathbf{x}}}\right),
	\end{equation}
	and hence the level curves of $\Delta\hat{\omega}\left(\hat{\mathbf{x}},t\right)$
	also coincide with those of $\hat{\mathbf{u}}_{0}\left(\mathbf{\hat{\mathbf{x}}}\right)$. 
	
	\emph{We conclude that both material and instantaneous active barriers
		to vorticity and linear momentum transport coincide with the streamlines
		of $\hat{\mathbf{u}}_{0}\left(\mathbf{\hat{\mathbf{x}}}_{0}\right)$.
	}In particular, we obtain the correct active barrier distributions
	that we inferred for our motivational 2D channel-flow example in Fig.
	\ref{fig: steady 2D channel flow example-0} (see \eqref{eq:horizontal shear jet}
	in Appendix A), which is part of the solution family \eqref{eq:2D flow family}.
	Importantly, we obtain the same frame-indifferent conclusion about
	active barriers from any finite-time (or even instantaneous) analysis
	of the velocity field \eqref{eq:2D flow family} .
\end{example}

\subsection{Directionally steady Beltrami flows}\label{sec:directionally steady Beltrami flows}

Virtually all explicitly known, unsteady solutions of the 3D incompressible
Navier\textendash Stokes equations satisfy the \emph{strong Beltrami
	property}
\begin{equation}
\boldsymbol{\mathbf{\omega}}(\mathbf{x},t)=k(t)\mathbf{u}(\mathbf{x},t)\label{eq:Beltrami property}
\end{equation}
for some scalar function $k(t)$ (see Majda \& Bertozzi 2002). By
definition, for any such incompressible strong Beltrami flow, we obtain
\begin{align}
\Delta\boldsymbol{\omega} & =\boldsymbol{\nabla}\left(\boldsymbol{\nabla}\cdot\boldsymbol{\omega}\right)-\boldsymbol{\nabla}\times\left(\boldsymbol{\nabla}\times\boldsymbol{\omega}\right)=-k^{3}\mathbf{u},\nonumber \\
\Delta\mathbf{u} & =\frac{1}{k}\Delta\boldsymbol{\omega}=-k^{2}\mathbf{u}.
\end{align}
Recall that if a steady Euler flow is non-Beltrami, then it is integrable
(Arnold \& Keshin 1998). Therefore, only velocity field satisfying
the Beltrami property can generate complicated particle dynamics in
steady, inviscid flows.

We call an unsteady strong Beltrami flow with velocity field $\mathbf{u}(\mathbf{x},t)$
a \emph{directionally steady Beltrami flow} if
\begin{equation}
\mathbf{u}(\mathbf{x},t)=\alpha(t)\mathbf{u}_{0}(\mathbf{x}),\quad\boldsymbol{\mathbf{\mathbf{\omega}}}\left(\mathbf{x},t\right)=\boldsymbol{\nabla}\times\mathbf{u}(\mathbf{x},t)=k(t)\alpha(t)\mathbf{u}_{0}\left(\mathbf{x}\right)\label{eq:separable Beltrami flow definition}
\end{equation}
hold for some continuously differentiable scalar function $\alpha(t)$.
Note that any steady strong Beltrami flow $\mathbf{u}_{0}(\mathbf{x})$
(which necessarily admits $k(t)\equiv k=const.)$ solves the steady
Euler equation and generates a directionally steady Beltrami solution
$\mathbf{u}(\mathbf{x},t)=\exp\left(-\nu k^{2}t\right)\mathbf{u}_{0}(\mathbf{x})$
for the unsteady Navier\textendash Stokes equation under conservative
forcing (Majda and Bertozzi 2002). 

For all directionally steady Beltrami flows, we obtain the following
simple result on active transport barriers:
\begin{thm}
	\label{thm: Separable Beltrami momentum and vorticity barriers }
	Both material and instantaneous active barriers to the diffusive transport
	of linear momentum and vorticity in directionally steady Beltrami
	flows coincide exactly with structurally stable, 2D invariant manifolds
	of the steady component $\mathbf{u}_{0}(\mathbf{x})$ of the velocity
	field. These in turn coincide with 2D invariant manifolds of $\mathbf{u}(\mathbf{x},t)$
	defined in \eqref{eq:separable ODE}. 
\end{thm}
\begin{proof}
	See Appendix D.
\end{proof}
Theorem \ref{thm: Separable Beltrami momentum and vorticity barriers }
shows that invariant manifolds for the Lagrangian particle motion
in directionally steady Beltrami flows coincide with material and
instantaneous active barriers to linear momentum and vorticity transport.
This agrees with one's intuition: observed mass-transport barriers
in these flows are expected to coincide with barriers to vorticity
and momentum transport, given that momentum and vorticity are scalar
multiples of each other. Remarkably, as in the case of 2D flows analyzed
in the previous section, the exact barriers emerge from our analysis
independently of the choice of the finite-time interval $[t_{0},t_{1}]$,
including the case of instantaneous extraction with $t_{0}=t_{1}$. 
\begin{rem}
	\label{rem: remark on tori in Beltrami flows}In view of Theorem \ref{thm: Separable Beltrami momentum and vorticity barriers },
	when viewed as transport barriers to momentum and vorticity, both
	Lagrangian and Eulerian coherent vortex boundaries in directionally
	steady Beltrami flows coincide with outermost members of nested families
	of invariant tori identified form purely advective mixing studies
	(see, e.g., Dombre et al. 1986 and Haller 2001). This is in line with
	the expectation we stated earlier that the outermost members of a
	family of non-filamenting, closed material surfaces will also be outermost
	barriers to diffusive transport.
\end{rem}
\begin{example}
	\label{ex: time-dependent ABC flow} Examples of directionally steady
	Beltrami flows include the Navier-Stokes flow family (Ethier \& Steinman
	1994)
	\begin{equation}
	\mathbf{u}(\mathbf{x},t)=e^{-\nu d^{2}t}\mathbf{u}_{0}(\mathbf{x}),\qquad\mathbf{u}_{0}(\mathbf{x})=-a\left(\begin{array}{c}
	e^{ax_{1}}\sin\left(ax_{2}\pm dx_{3})+e^{ax_{3}}\cos\left(ax_{1}\pm dx_{2}\right)\right)\\
	e^{ax_{2}}\sin\left(ax_{3}\pm dx_{1})+e^{ax_{1}}\cos\left(ax_{2}\pm dx_{3}\right)\right)\\
	e^{ax_{3}}\sin\left(ax_{1}\pm dx_{2})+e^{ax_{2}}\cos\left(ax_{3}\pm dx_{1}\right)\right)
	\end{array}\right),
	\end{equation}
	and the viscous, unsteady version of the classic ABC flow $\mathbf{u}_{0}(\mathbf{x})$
	(Dombre et al. 1986), given by
	\begin{equation}
	\mathbf{u}(\mathbf{x},t)=e^{-\nu t}\mathbf{u}_{0}(\mathbf{x}),\qquad\mathbf{u}_{0}(\mathbf{x})=\left(\begin{array}{c}
	A\sin x_{3}+C\cos x_{2}\\
	B\sin x_{1}+A\cos x_{3}\\
	C\sin x_{2}+B\cos x_{1}
	\end{array}\right).\label{eq:time-dependent ABC flow}
	\end{equation}
	All lengths in these examples are non-dimensional. Further examples
	of 3D, unsteady but directionally steady Beltrami solutions are derived
	by Barbato, Berselli \& Grisanti (2007) and Antuono (2020).

For all these flows, Theorem \ref{thm: Separable Beltrami momentum and vorticity barriers }
guarantees that all material and instantaneous active barriers to
diffusive momentum and vorticity transport coincide with structurally
stable, 2D invariant manifolds of the flow generated by the steady
velocity field $\mathbf{u}_{0}(\mathbf{x})$. Such manifolds can be
captured via their intersections with Poincar\'e sections, with these
intersections appearing as invariant curves of the associated Poincar\'e
map, as first illustrated by Dombre et al. (1986) for one cross-section
of the ABC flow. A more complete set of Poincar\'e maps along three
orthogonal planes is shown in the right subplot Fig. \ref{fig:Poincare map ABC flow},
which reveals several families of 2D invariant tori, appearing as
spatially periodic cylinders.

As discussed in Remark \ref{rem: remark on tori in Beltrami flows},
these torus families form objectively defined coherent vortices, with
each torus acting as an internal barriers to both momentum- and vorticity-transport
within the vortex. Outermost members of these torus families provide
objective, active-transport-based coherent vortex boundaries. By their
invariance under the flow map, they remain perfectly coherent under
advection. For comparison, we also show in Fig. \ref{fig:Poincare map ABC flow}
other common Eulerian diagnostics applied to this flow: sectional
streamlines (computed from velocities projected onto the three faces
of the cube at time $t=0$ ); vorticity levels for $\mathbf{u}(\mathbf{x},t)$
at $t=0$; and levels of the parameter $Q=\left|\boldsymbol{W}\right|^{2}-\left|\mathbf{S}\right|^{2}$
at $t=0$, with the spin tensor $\mathbf{W}$ and the rate-of-strain
tensor $\mathbf{S}$ defined as
\begin{equation}
\mathbf{W}=\frac{1}{2}\left[\boldsymbol{\nabla}\mathbf{u}-\left(\boldsymbol{\nabla}\mathbf{u}\right)^{T}\right],\quad\mathbf{S}=\frac{1}{2}\left[\boldsymbol{\nabla}\mathbf{u}+\left(\boldsymbol{\nabla}\mathbf{u}\right)^{T}\right].\label{eq:spin and rate of strain}
\end{equation}
The $Q>0$ region is often used to define vortices, and hence the
white level sets are considered vortex boundaries by the\textbf{ $Q$}-criterion
of Hunt et al. (1988). The structures appearing in the latter three
plots change under an observer change and do not remain invariant
under advection by the flow map. 
\begin{figure}
	\centering{}\includegraphics[width=1\textwidth]{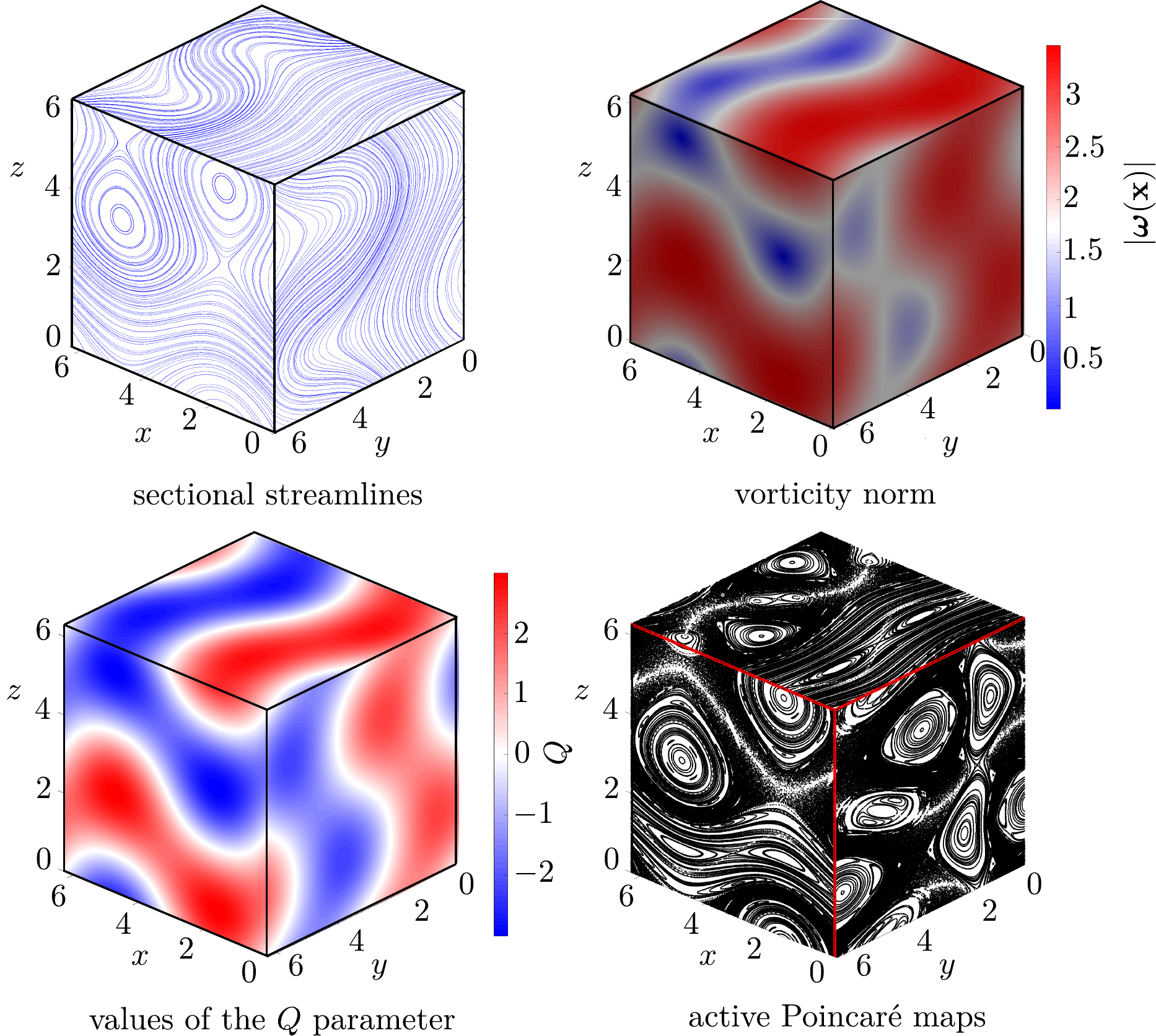}\caption{Three nonobjective diagnostics (sectional streamline plots, the vorticity
		norm and the $Q$-parameter) for the unsteady ABC flow \eqref{eq:time-dependent ABC flow}
		with $A=\sqrt{3}$, $B=\sqrt{2}$ and $C=1$, at time $t=0$. All
		three plots remain the same for all times, because the velocity field
		is directionally steady. Also shown are three objective, active Poincar\'e
		maps computed for the associated barrier equations. The Lagrangian
		and Eulerian barrier equations for this flow are given by \textbf{$\mathbf{x}_{0}^{\prime}=\mathbf{u}_{0}\left(\mathbf{x}_{0}\right)$}
		by Theorem \ref{thm: Separable Beltrami momentum and vorticity barriers },
		both for momentum and vorticity. Black dots on the active Poincar\'e
		sections indicate repeated return locations of barrier trajectories
		launched from the same section. Intersections of 2D, toroidal transport
		barriers with the three Poincar\'e sections are visible as invariant
		curves of these Poincar\'e maps. Outermost members of these torus families
		define objective coherent vortex boundaries. \label{fig:Poincare map ABC flow}}
\end{figure}

Further studies revealing the same invariant manifolds in the steady
ABC flow using finite-time Lyapunov exponents (FTLE) and the polar
rotation angle (PRA) were given by Haller (2001) and Farazmand \&
Haller (2016), each emphasizing different classes of barriers from
the complete collection revealed in Fig. \ref{fig:Poincare map ABC flow}.
The FTLE and the PRA are generally usable structure detection tools
along any cross section of an unsteady flow, whereas Poincar\'e maps
are only defined for trajectories returning to the same cross section
of a steady or time-periodic flow. In the next section, we will also
show the passive FTLE and PRA plots computed for the ABC flow \eqref{eq:time-dependent ABC flow},
as well as active versions of the FTLE and PRA applied to the barrier
equations of the ABC flow over the same time interval.
\end{example}

\section{Practical implementation of active barrier identification \label{sec:Practical-implementation-on}}

Here we discuss the computation of the barrier equations for momentum
and vorticity from velocity data sets. In addition, we introduce dynamically
active versions of three simple LCS techniques that can be used to extract
active transport barrier surfaces. While these LCS diagnostics enable a quick visualization of active barriers in an objective fashion, the more advanced LCS methods we cited in the Introduction are also directly applicable to the barrier equations.

\subsection{Computation from highly resolved numerical data}

All applications of our main results in Theorems \ref{thm: 3D incompressible N-S,  linear momentum barrier eq }-\ref{thm: 3D incompressible N-S,  vorticity barrier eq}
require the analysis of the associated 3D autonomous, divergence-free
dynamical systems that depend on the Laplacian of \textbf{$\mathbf{u}(\mathbf{x},t)$
}for momentum-transport barriers, or on the Laplacian of $\boldsymbol{\omega}\left(\mathbf{x},t\right)$
for vorticity-transport barriers. In direct numerical simulations
(DNS) of the Navier\textendash Stokes equation, the required Laplacians
can be computed spectrally with high accuracy,
as our numerical results in Section \ref{subsec:channelflow}
will illustrate. With these Laplacians at hand, one proceeds to find
invariant manifolds of the barrier equations in Theorems \ref{thm: 3D incompressible N-S,  linear momentum barrier eq }-\ref{thm: 3D incompressible N-S,  vorticity barrier eq},
which invariably involves computing trajectories of these equations.
In generating these trajectories numerically, it is usually helpful
to omit the (small) viscosity $\nu$ from the right-hand sides
of the barrier equations to speed up the simulation. This omission
of $\nu$ is equivalent to a rescaling of the time-like variable $s$
in the barrier equations, which does not alter the trajectories of
these autonomous differential equations.

For 2D incompressible Navier\textendash Stokes flows, Theorems \ref{thm:2D momentum barriers}-\ref{thm:2D vorticity barriers}
show the relevant barrier equations to be computed. The right-hand-sides
of these equations are autonomous Hamiltonian vector fields whose
trajectories coincide with the level curves of the corresponding Hamiltonians.
Strictly speaking, therefore, the numerical solution of these barrier
equations can be avoided by simply plotting the level curves of their
Hamiltonians, which can be computed by finite-differencing the velocity
field (but see also Remark \ref{rem:used active LCS diagnostics instead of level curves} in section \ref{subsec:active barriers in 2D-Navier=002013Stokes-flows}).

\subsection{Computation from experimental or lower-resolved numerical data\label{subsec:Computation-from-experimental}}

Taking second and third spatial derivatives of a velocity field obtained
from an already finalized numerical simulation or experiment is challenging.
An alternative is to work with the original material derivatives arising
in our definition of active transport, rather than with the Laplacians
of the velocity and the vorticity. More specifically, if we let $\mathbf{a}(\mathbf{x},t)=\frac{D\mathbf{u}}{Dt}(\mathbf{x},t)$
denote the Lagrangian particle acceleration along fluid trajectories,
then using the general momentum equation \eqref{eq:main continuum eq. of motion},
the active barrier equations \eqref{eq:NS Lagrangian momentum barrier eq}
and \eqref{eq:NS Eulerian momentum barrier eq} for the linear momentum
can be rewritten as
\begin{align}
\mathbf{x}_{0}^{\prime} & =\overline{\det\nabla\mathbf{F}_{t_{0}}^{t}\left(\mathbf{F}_{t_{0}}^{t}\right)^{*}\left[\rho\mathbf{a}+\boldsymbol{\nabla}p-\mathbf{q}\right]},\label{eq:NS Lagrangian momentum barrier eq-1}\\
\mathbf{x}^{\prime} & =\rho\mathbf{a}+\boldsymbol{\nabla}p-\mathbf{q}.\label{eq:NS Euerian momentum barrier eq-1}
\end{align}
These equations involve the Lagrangian acceleration, $\mathbf{a}(\mathbf{x},t)$,
which can be obtained from high-resolution numerical or experimental
data via the temporal differentiation of the velocity vector along
trajectories. 

Similarly, the most general active barrier equations \eqref{eq:NS Lagrangian vorticity barrier eq}-\eqref{eq:NS Eulerian vorticity barrier eq}
for vorticity can be rewritten as 
\begin{align}
\mathbf{x}_{0}^{\prime} & =\overline{\det\nabla\mathbf{F}_{t_{0}}^{t}\left(\mathbf{F}_{t_{0}}^{t}\right)^{*}\boldsymbol{\nabla}\times\left[\mathbf{a}+\frac{1}{\rho}\left(\boldsymbol{\nabla}p-\mathbf{q}\right)\right]},\label{eq:NS Lagrangian vorticity barrier eq-11}\\
\mathbf{x}^{\prime} & =\boldsymbol{\nabla}\times\left[\mathbf{a}+\frac{1}{\rho}\left(\boldsymbol{\nabla}p-\mathbf{q}\right)\right].\label{eq:NS Eulerian vorticity barrier eq-11}
\end{align}
In particular, for incompressible, constant density, Newtonian fluids
subject only to potential body forces, the material and instantaneous
barrier equations for vorticity in \eqref{eq:NS Lagrangian vorticity barrier eq-11}-\eqref{eq:NS Eulerian vorticity barrier eq-11}
simplify to
\begin{align}
\mathbf{x}_{0}^{\prime} & =\overline{\left(\mathbf{F}_{t_{0}}^{t}\right)^{*}\boldsymbol{\nabla}\times\mathbf{a}},\label{eq:NS Lagrangian vorticity barrier eq-1}\\
\mathbf{x}^{\prime} & =\boldsymbol{\nabla}\times\mathbf{a},\label{eq:NS Eulerian vorticity barrier eq-1-1}
\end{align}
given that $\boldsymbol{\nabla}\times\left[\frac{1}{\rho}\left(\boldsymbol{\nabla}p-\mathbf{q}\right)\right]=\frac{1}{\rho}\boldsymbol{\nabla}\times\left[\boldsymbol{\nabla}p-\mathbf{q}\right]\equiv\mathbf{0}$ holds for such flows.
\subsection{Passive vs. active Poincar\'e maps}

Passive Poincar\'e maps for 3D steady flows map initial conditions of
trajectories launched from a selected 2D section to their first return
to the section, if such a return exists. We refer to a Poincar\'e map computed
for the 3D steady barrier equations \eqref{eq:barrier equation} or
\eqref{eq:Eulerian barrier equation} as \emph{active Poincar\'e map
	(}see Fig. \ref{fig:Poincare map ABC flow} for an example). This
two-dimensional mapping generally does not preserve the standard 2D
area, but preserves a general area form, which makes the active Poincar\'e
map a 2D symplectic map (Meiss 1992). One-dimensional invariant curves
of 2D symplectic maps satisfy the only available formal definition
of advective transport barriers by MacKay, Meiss \& Percival (1984),
as we noted in the Introduction. Structurally stable invariant curves
of 2D symplectic maps include stable and unstable manifolds of hyperbolic
fixed points and Kolmogorov\textendash Arnold-Moser (KAM) curves,
i.e., nested families of closed curves satisfying non-resonance and
twist-conditions (Arnold 1978). 

In contrast to active Poincar\'e maps, the mapping relating subsequent
returns of trajectories to a selected section in the general unsteady
velocity field $\mathbf{u}(\mathbf{x},t$) is not well-defined as
a single Poincar\'e map. Rather, this map will be different for different
initial times $t_{0}$. Therefore, passive Poincar\'e maps are generally
inapplicable to LCS detection in $\mathbf{u}(\mathbf{x},t)$, whereas
active Poincar\'e maps are well-defined on barrier-equation trajectories
that return to a cross section. In case they do not, the active versions of the FTLE and  PRA fields introduced next provide alternative tools to
uncover structurally stable invariant manifolds in the the barrier
equations.

\subsection{Passive FTLE vs. active FTLE (aFTLE)\label{subsec:Passive-FTLE-vs.-aFTLE}}

We fix a time interval $[t_{0},t_{1}]$ over which we would like to
identify LCSs as coherent material surfaces in the advective transport
induced by the unsteady velocity field $\mathbf{u}(\mathbf{x},t)$.
With the notation of section \ref{sec:Set-up}, the right Cauchy\textendash Green
strain tensor $\mathbf{C}_{t_{0}}^{t_{1}}\left(\mathbf{x}_{0}\right)$
is defined as 
\begin{equation}
\mathbf{C}_{t_{0}}^{t_{1}}\left(\mathbf{x}_{0}\right):=\left[\boldsymbol{\nabla}\mathbf{F}_{t_{0}}^{t_{1}}\left(\mathbf{x}_{0}\right)\right]^{T}\boldsymbol{\nabla}\mathbf{F}_{t_{0}}^{t_{1}}\left(\mathbf{x}_{0}\right),\label{eq:CG tensor}
\end{equation}
with the superscript $T$ referring to the transpose. Then, if $\lambda_{\mathrm{max}}\left(\mathbf{C}_{t_{0}}^{t_{1}}\right)$
denotes the maximal eigenvalues of the symmetric, positive definite
tensor $\mathbf{C}_{t_{0}}^{t_{1}}$, then the (passive) FTLE field
of $\mathbf{u}(\mathbf{x},t)$ over the $\left[t_{0},t_{1}\right]$
time interval is defined as
\begin{equation}
\mathrm{FTLE}_{t_{0}}^{t_{1}}(\mathbf{x}_{0})=\frac{1}{2\left(t_{1}-t_{0}\right)}\log\lambda_{\mathrm{max}}\left(\mathbf{C}_{t_{0}}^{t_{1}}\left(\mathbf{x}_{0}\right)\right).\label{eq:passive FTLE}
\end{equation}

Two-dimensional ridges of $\mathrm{FTLE}_{t_{0}}^{t_{1}}(\mathbf{x}_{0})$
are quick indicators of the time $t_{0}$ locations of hyperbolic
LCS. They signal locally most repelling material surfaces when $t_{1}>t_{0}$
and locally most attracting material surfaces when $t_{1}<t_{0}$.
Valleys of $\mathrm{FTLE}_{t_{0}}^{t_{1}}(\mathbf{x}_{0})$ tend to
indicate elliptic (vortical) LCSs, whereas trenches of $\mathrm{FTLE}_{t_{0}}^{t_{1}}(\mathbf{x}_{0})$
signal parabolic (jet-type) LCSs. The minimal and maximal value of
$t_{0}$ and $t_{1}$ are governed by the length of the available
data and the scales relative to which we wish to determine the
LCSs in the flow. The flow-map gradient involved in the definition
of $\lambda_{\mathrm{max}}\left(\mathbf{C}_{t_{0}}^{t_{1}}\left(\mathbf{x}_{0}\right)\right)$
can be computed by finite-differencing a set of trajectories, launched
from a regular grid of initial conditions, with respect to those initial
conditions. The $\mathrm{FTLE}_{t_{0}}^{t_{1}}(\mathbf{x}_{0})$ is
a simple but objective LCS diagnostic, with its strengths and limitations
reviewed in Haller (2015). 

For $t_{1}=t_{0}\equiv t$, the instantaneous
of limit of the FTLE field is the maximal rate-of-strain eigenvalue
\begin{equation}
\mathrm{FTLE}_{t}^{t}(\mathbf{x})=\lambda_{\mathrm{max}}\left(\mathbf{S}(\mathbf{x},t)\right),\label{eq:instantaneous limit of passive FTLE}
\end{equation}
with the rate-of-strain tensor $\mathbf{S}(\mathbf{x},t)$ defined
in \eqref{eq:spin and rate of strain}, as noted by Serra \& Haller (2016) and Nolan, Serra, \& Ross (2020).  This eigenvalue field can,
in principle, be used to detect objective Eulerian coherent structures
(OECS) as instantaneous limits of LCS. In practice, the field $\mathrm{FTLE}_{t}^{t}(\mathbf{x})$
often provides insufficient spatial detail, but the eigenvector field
of $\mathbf{S}(\mathbf{x},t)$ can be used to define and extract OECS
(see Serra \& Haller 2016).

In contrast to passive FTLE, by \emph{active }FTLE (aFTLE) we mean here the
implementation of the  FTLE diagnostic on the steady material
barrier equation \eqref{eq:barrier equation}, including its steady
instantaneous version \eqref{eq:Eulerian barrier equation}. We again
select a physical time interval $[t_{0},t_{1}]$ over which we would
like to locate barriers to the active transport of the vector field\textbf{
}$\mathbf{f}(\mathbf{x},t)$ in the velocity field $\mathbf{u}(\mathbf{x},t)$.
Let $\tilde{\mathbf{x}}_{0}(s;0,\mathbf{x}_{0})$ denote the trajectory
of the barrier ODE \eqref{eq:barrier equation} starting at the dummy
time $s=0$ from the initial location $\mathbf{x}_{0}$. The corresponding
autonomous flow map for this barrier ODE will be denoted by the \emph{active
	flow map} $\boldsymbol{\mathcal{F}}_{t_{0},t_{1}}^{s}\colon\mathbf{x}_{0}\mapsto\tilde{\mathbf{x}}_{0}(s;0,\mathbf{x}_{0}).$
The associated \emph{active Cauchy\textendash Green strain tensor}
for the barrier equation \eqref{eq:barrier equation} can then be
defined as 
\begin{equation}
\boldsymbol{\mathcal{C}}_{t_{0},t_{1}}^{s}\left(\mathbf{x}_{0}\right):=\left[\boldsymbol{\nabla}\boldsymbol{\mathcal{F}}_{t_{0},t_{1}}^{s}\left(\mathbf{x}_{0}\right)\right]^{T}\boldsymbol{\nabla}\boldsymbol{\mathcal{F}}_{t_{0},t_{1}}^{s}\left(\mathbf{x}_{0}\right).\label{eq:active CG tensor}
\end{equation}
Again, if $\lambda_{\mathrm{max}}\left(\boldsymbol{\mathcal{C}}_{t_{0},t_{1}}^{s}\right)$
denotes the maximal eigenvalue of the symmetric, positive definite
tensor $\boldsymbol{\mathcal{C}}_{t_{0},t_{1}}^{s}$, then the aFTLE
field of $\mathbf{u}(\mathbf{x},t)$ over the $\left[t_{0},t_{1}\right]$
time interval, with respect to the vector field $\mathbf{f}(\mathbf{x},t)$,
is defined as
\begin{equation}
\mathrm{aFTLE}_{t_{0},t_{1}}^{s}(\mathbf{x}_{0};\mathbf{f})=\frac{1}{2s}\log\lambda_{\mathrm{max}}\left(\boldsymbol{\mathcal{C}}_{t_{0},t_{1}}^{s}\left(\mathbf{x}_{0}\right)\right).\label{eq:aFTLE -- Lagrangian}
\end{equation}
Here the time-like parameter $s$ governs the level of accuracy and
spatial resolution in the visualization of active transport barriers.
The only limitation to the choice of $s$ is that the trajectories
of the barrier equation \eqref{eq:barrier equation} may ultimately
leave the spatial domain $U$ over which the barrier equation is known.
This is, however, unrelated to the physical time that the trajectories
of $\mathbf{u}(\mathbf{x},t)$ spend in the domain $U$. 

For instance, in our 2D turbulence simulation to be analyzed in section
\ref{subsec:Two-dimensional-turbulence}, the
maximal possible spatial detail for LCS from $\mathrm{FTLE}_{t_{0}}^{t_{1}}(\mathbf{x}_{0})$
is limited by the length of the time interval $[t_{0},t_{1}]=[0,50],$
given that this is the temporal length of the available data set.
In contrast, on the same data set, $\mathrm{aFTLE}_{t_{0},t_{1}}^{s}(\mathbf{x}_{0};\mathbf{f})$
can be computed for arbitrarily large $\left|s\right|,$ because the
barrier vector field is known globally for $U=\mathbb{R}^{2}$. Similarly,
in our 3D turbulent channel flow example in section \ref{subsec:channelflow},
trajectories of the barrier equation tend to stay in the finite channel
domain for much longer (non-dimensional) dummy times than the non-dimensional
residence time of fluid trajectories in the same channel. 

As a consequence,
aFTLE has the potential to provide much finer spatial detail for active
barriers than one is able to obtain for LCSs in the same data set
from the passive FTLE. Figure \ref{fig:ABC flow passive vs active FTLE and PRA} shows 
this substantial refinement obtained from the vorticity-based aFTLE
relative to the passive FTLE computed over the same time interval
$[t_{0},t_{1}]=[0,5]$ for the unsteady ABC flow \eqref{eq:time-dependent ABC flow}.
(For this particular flow, the linear-momentum-based aFTLE and aPRA
would give identical results by Theorem \ref{thm: Separable Beltrami momentum and vorticity barriers }.)
In addition, aFTLE is always guaranteed to converge under increasing
$s$, as illustrated in Fig. \ref{fig:ABC flow passive vs active FTLE and PRA},
while the convergence of $\mathrm{FTLE}_{t_{0}}^{t_{1}}(\mathbf{x}_{0})$
is generally not guaranteed in an unsteady flow with time-varying
structures.

\begin{figure}
	\centering{}\includegraphics[width=1\textwidth]{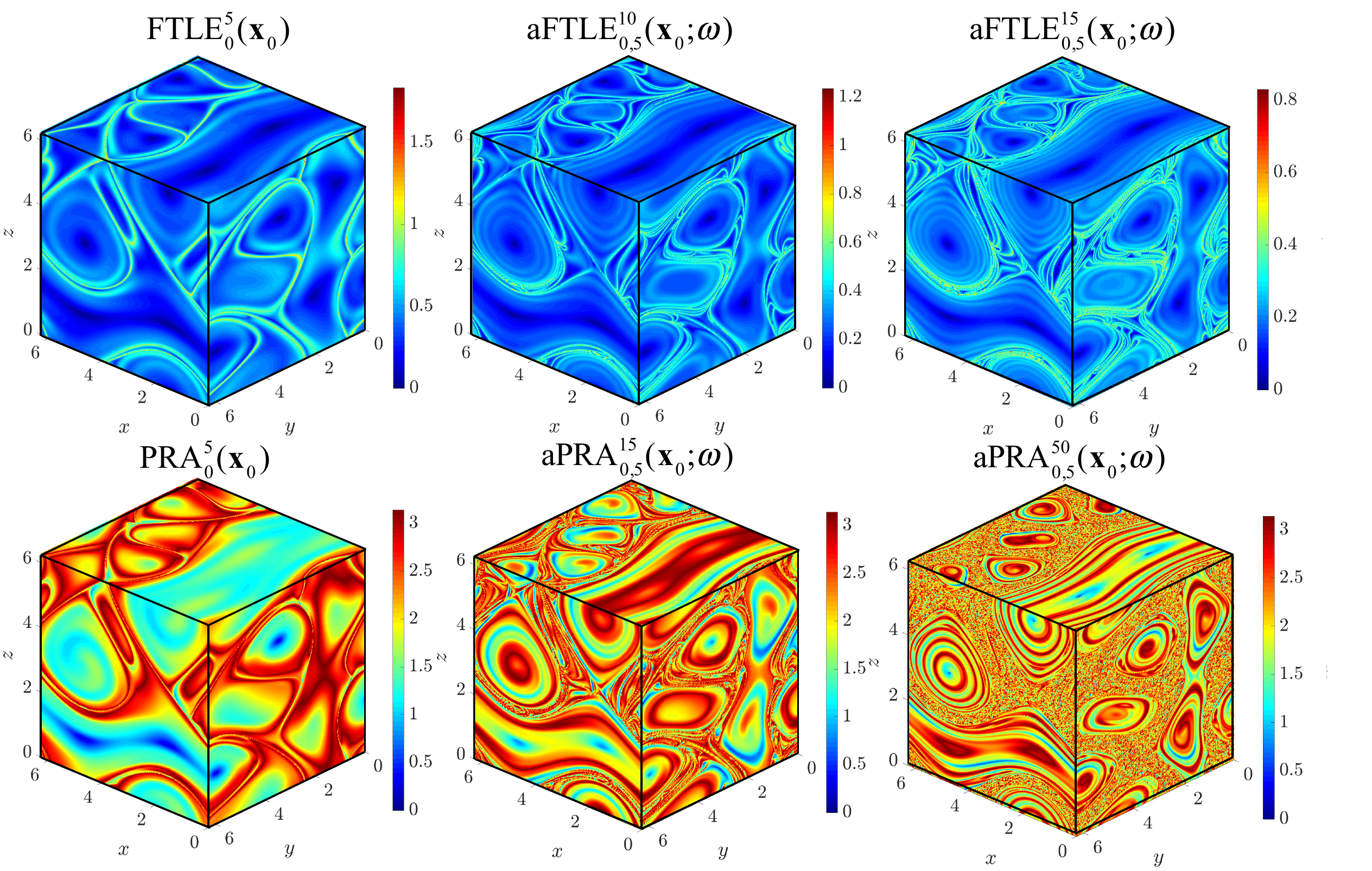}\caption{Passive and vorticity-based active versions of the FTLE and PRA diagnostics
		for the unsteady ABC flow \eqref{eq:time-dependent ABC flow}, computed
		over the same time interval ($[t_{0},t_{1}]=[0,5]$) and with the
		same spatial resolution ($300^{3}$ grid points in the spatial domain
		$[0,2\pi]^{3}$). Two values for the barrier-time $s$ were selected
		to illustrate the increasing spatial resolution and convergence of
		hyperbolic barriers by the aFTLE and of elliptic barriers by the aPRA
		under increasing $s$-times. With the exception of the passive PRA,
		all diagnostics shown here are objective. \label{fig:ABC flow passive vs active FTLE and PRA}}
\end{figure}

The $t_{1}=t_{0}\equiv t$ limit of the aFTLE field in eq. \eqref{eq:aFTLE -- Lagrangian}
is
\begin{equation}
\mathrm{aFTLE}_{t,t}^{s}(\mathbf{x};\mathbf{f})=\frac{1}{2s}\log\lambda_{\mathrm{max}}\left(\boldsymbol{\mathcal{C}}_{t,t}^{s}\left(\mathbf{x}\right)\right).\label{eq:aFTLE field -- Eulerian}
\end{equation}
Here $\boldsymbol{\mathcal{C}}_{t,t}^{s}\left(\mathbf{x}\right)$
is simply computed from the autonomous flow map $\boldsymbol{\mathcal{F}}_{t,t}^{s}(\mathbf{x})$
of the instantaneous barrier equation \eqref{eq:Eulerian barrier equation},
with the instantaneous time $t$ playing the role of a constant parameter
in the computation. Again, the time-like evolutionary variable $s$
in this computation can be arbitrarily large in norm, as long as the
trajectories generated by the barrier flow map $\boldsymbol{\mathcal{F}}_{t,t}^{s}(\mathbf{x})$
stay in the domain $U$ over which the barrier vector field $\mathbf{b}_{t}^{t}(\mathbf{x})$
is known. This guarantees convergence and higher resolution in the
detection of instantaneous objective barriers from $\mathrm{aFTLE}_{t,t}^{s}(\mathbf{x};\mathbf{f})$
when compared with $\mathrm{FTLE}_{t}^{t}(\mathbf{x})$. The only
practical limitation to resolving the details of active barriers via
aFTLE is the spatial resolution of the available data.

\subsection{Passive PRA vs. active PRA (aPRA)\label{subsec:Passive-PRA-vs.-aPRA}}

By the polar decomposition theorem (Gurtin, Fried \& Anand 2013),
the deformation gradient $\boldsymbol{\nabla}\mathbf{F}_{t_{0}}^{t_{1}}(\mathbf{x}_{0})$
can be uniquely decomposed as
\begin{equation}
\boldsymbol{\nabla}\mathbf{F}_{t_{0}}^{t_{1}}=\mathbf{R}_{t_{0}}^{t_{1}}\mathbf{U}_{t_{0}}^{t_{1}},\label{eq:polar decomposition}
\end{equation}
with the proper orthogonal \emph{rotation tensor }$\mathbf{R}_{t_{0}}^{t_{1}},$
the symmetric and the positive definite \emph{right stretch tensor}
$\mathbf{U}_{t_{0}}^{t_{1}}$ . The decomposition \eqref{eq:polar decomposition}
means that a general deformation can locally always be viewed as triaxial
stretching and compression followed by a rigid-body rotation. One
can verify by direct substitution into \eqref{eq:polar decomposition}
that $\mathbf{R}_{t_{0}}^{t_{1}}$ and $\mathbf{U}_{t_{0}}^{t_{1}}$
must be of the form
\begin{equation}
\mathbf{U}_{t_{0}}^{t_{1}}=\left[\mathbf{C}_{t_{0}}^{t_{1}}\right]^{1/2},\qquad\mathbf{R}_{t_{0}}^{t_{1}}=\boldsymbol{\nabla}\mathbf{F}_{t_{0}}^{t_{1}}\left[\mathbf{U}_{t_{0}}^{t_{1}}\right]^{-1},\label{eq:polar decomposition tensors}
\end{equation}
with $\mathbf{C}_{t_{0}}^{t}$ defined in \eqref{eq:CG tensor}. The
first equation in \eqref{eq:polar decomposition tensors} shows that
$\mathbf{U}_{t_{0}}^{t}$ can be computed using the singular-value-decomposition
of $\mathbf{C}_{t_{0}}^{t}$. With $\mathbf{U}_{t_{0}}^{t_{1}}$ at
hand, one can compute the rotation tensor $\mathbf{R}_{t_{0}}^{t_{1}}$
from the second equation of \eqref{eq:polar decomposition tensors}. 

Farazmand \& Haller (2016) show that $\mathbf{R}_{t_{0}}^{t_{1}}(\mathbf{x}_{0})$
rotates material elements around an axis of rotation by the \emph{polar
	rotation angle} (PRA) satisfying
\begin{equation}
\mathrm{PRA}_{t_{0}}^{t_{1}}(\mathbf{x}_{0})=\cos^{-1}\left[\frac{1}{2}\left(\mathrm{{tr}\,}\mathbf{R}_{t_{0}}^{t_{1}}(\mathbf{x}_{0})-1\right)\right]=\cos^{-1}\left[\frac{1}{2}\left(\sum_{i=1}^{3}\left\langle \boldsymbol{\xi}_{i}(\mathbf{x}_{0}),\boldsymbol{\eta}_{i}(\mathbf{x}_{0})\right\rangle -1\right)\right],\label{eq:cos theta}
\end{equation}
with $\boldsymbol{\xi}_{i}(\mathbf{x}_{0})$ and $\boldsymbol{\eta}_{i}(\mathbf{x}_{0})$
denoting the right and left singular vectors of $\boldsymbol{\nabla}\mathbf{F}_{t_{0}}^{t}(\mathbf{x}_{0})$.
For 2D flows viewed as 3D flows with a symmetry, the intermediate
eigenvalue of $\mathbf{C}_{t_{0}}^{t}$ is always one, which simplifies
$\mathrm{PRA}_{t_{0}}^{t_{1}}(\mathbf{x}_{0})$ to
\begin{equation}
\mathrm{PRA}_{t_{0}}^{t_{1}}(\mathbf{x}_{0})=\cos^{-1}\left\langle \boldsymbol{\xi}_{1}(\mathbf{x}_{0}),\boldsymbol{\eta}_{1}(\mathbf{x}_{0})\right\rangle =\cos^{-1}\left\langle \boldsymbol{\xi}_{2}(\mathbf{x}_{0}),\boldsymbol{\eta}_{2}(\mathbf{x}_{0})\right\rangle ,\quad\mathbf{x}_{0}\in\mathbb{R}^{2}.\label{eq:cos theta 2D}
\end{equation}

Farazmand \& Haller (2016) propose $\mathrm{PRA}_{t_{0}}^{t_{1}}(\mathbf{x}_{0})$
as a diagnostic tool for elliptic (rotational) LCS. They find that
nested circular or toroidal level sets of $\mathrm{PRA}_{t_{0}}^{t_{1}}(\mathbf{x}_{0})$
indeed highlight elliptic LCS significantly sharper than FTLE valleys
do. They also show, however, that these level sets are only objective
for 2D flows. Similarly to FTLE calculations for $\mathbf{u}(\mathbf{x},t)$,
the spatial scales resolved by the passive PRA in 2D flows are limited
by the length of the time interval $[t_{0},t_{1}].$ For 3D flows,
an additional limitation of the PRA is the non-objectivity of its
level surfaces. The instantaneous limit $t_{0}=t_{1}\equiv t$ of
the PRA gives $\mathrm{PRA}_{t}^{t}(\mathbf{x})\equiv0$, and hence
this diagnostic is unable to detect instantaneous limits of elliptic
OECS. 

In contrast, using the \emph{active rotation tensor}
\begin{equation}
\boldsymbol{\mathcal{R}}_{t_{0,}t_{1}}^{s}=\boldsymbol{\nabla}\boldsymbol{\mathcal{F}}_{t_{0},t_{1}}^{s}\left[\boldsymbol{\mathcal{C}}_{t_{0},t_{1}}^{s}\right]^{-1/2},\label{eq:polar decomposition tensors-1}
\end{equation}
the corresponding \emph{active }PRA (aPRA) is obtained in 3D as 
\begin{equation}
\mathrm{aPRA}_{t_{0,}t_{1}}^{s}(\mathbf{x}_{0};\mathbf{f})=\cos^{-1}\left[\frac{1}{2}\left(\mathrm{{tr}\,}\boldsymbol{\mathcal{R}}_{t_{0,}t_{1}}^{s}(\mathbf{x}_{0})-1\right)\right]=\cos^{-1}\left[\frac{1}{2}\left(\sum_{i=1}^{3}\left\langle \boldsymbol{\xi}_{i}^{a}(\mathbf{x}_{0}),\boldsymbol{\eta}_{i}^{a}(\mathbf{x}_{0})\right\rangle -1\right)\right],\label{eq:cos theta-1}
\end{equation}
with $\boldsymbol{\xi}_{i}^{a}(\mathbf{x}_{0})$ and $\boldsymbol{\eta}_{i}^{a}(\mathbf{x}_{0})$
denoting the right and left singular vectors of the active deformation
gradient $\boldsymbol{\nabla}\boldsymbol{\mathcal{F}}_{t_{0},t_{1}}^{s}$.
For 2D flows, the corresponding formula is
\begin{equation}
\mathrm{aPRA}_{t_{0,}t_{1}}^{s}(\mathbf{x}_{0};\mathbf{f})=\cos^{-1}\left\langle \boldsymbol{\xi}_{1}^{a}(\mathbf{x}_{0}),\boldsymbol{\eta}_{1}^{a}(\mathbf{x}_{0})\right\rangle =\cos^{-1}\left\langle \boldsymbol{\xi}_{2}^{a}(\mathbf{x}_{0}),\boldsymbol{\eta}_{2}^{a}(\mathbf{x}_{0})\right\rangle ,\quad\mathbf{x}_{0}\in\mathbb{R}^{2}.\label{eq:cos theta-1 2D}
\end{equation}

Unlike for the passive PRA defined in \eqref{eq:cos theta}, the spatial
scales resolved by the aPRA can be gradually refined by increasing
the time-like parameter $s$ in $\mathrm{aPRA}_{t_{0,}t_{1}}^{s}$.
As in the case of the aFTLE, this increase is possible as long
as the underlying trajectories $\tilde{\mathbf{x}}_{0}\left(s;0,\mathbf{x}_{0}\right)$
of the barrier equation for $\mathbf{f}$ stay in the spatial domain\textbf{
	$U$ }where $\mathbf{u}(\mathbf{x},t)$ is known. As for aFTLE, the
spatial resolution of the active barriers discoverable by aPRA is
only limited by the resolution of the available velocity data. Figure
\ref{fig:ABC flow passive vs active FTLE and PRA} illustrates the
substantial refinement and convergence for increasing $s$-values
obtained from aPRA relative to the passive PRA computed over the same
time interval $[t_{0},t_{1}]=[0,5]$ for the unsteady ABC flow \eqref{eq:time-dependent ABC flow}.

Another major advantage of $\mathrm{aPRA}_{t_{0,}t_{1}}^{s}$ over
$\mathrm{PRA}_{t_{0}}^{t_{1}}$ is the objectivity of $\mathrm{aPRA}_{t_{0,}t_{1}}^{s}$,
which follows from the objectivity of the barrier vector field $\mathbf{b}_{t_{0}}^{t_{1}}(\mathbf{x}_{0})$.
Additionally, structures revealed by $\mathrm{aPRA}_{t_{0,}t_{1}}^{s}$
always converge as $s$ is increased, because $\mathrm{aPRA}_{t_{0,}t_{1}}^{s}$
operates on a steady flow, even though $\mathbf{u}(\mathbf{x},t)$
is unsteady. Finally, unlike $\mathrm{PRA}_{t_{0}}^{t_{1}}(\mathbf{x}_{0})$,
its active version, $\mathrm{aPRA}_{t_{0,}t_{1}}^{s}$, has a non-degenerate
instantaneous limit, $\mathrm{aPRA}_{t_{,}t}^{s}(\mathbf{x};\mathbf{f})$, which is just the PRA computed for the Eulerian barrier equation \eqref{eq:Eulerian barrier equation} over the barrier-time interval $[0,s]$. This limit enables the detection of instantaneous limits of active elliptic LCSs as active elliptic OECSs.

\subsection{Relationship between active and passive LCS diagnostics}\label{relationship between active and passive LCS diagnostics}

Active LCS diagnostics are applied to  barrier vector fields, whereas passive LCS diagnostics are applied to the underlying velocity field. As a consequence, active barriers highlighted by active LCS methods generally differ from passive barriers (coherent structures) detected by passive LCS methods. This is no surprise, given that these two types of barriers are constructed from different principles. 

As an extreme case, all Eulerian and Lagrangian active  barriers coincide with their passive counterparts in directionally steady Beltrami flows (see section \ref{sec:directionally steady Beltrami flows}). Therefore, the closer a generic flow is to a Beltrami flow in a given region, the closer its active and passive barriers will be to each other in that region. More broadly, the more correlated the velocity field is with its Laplacian (i.e., with the diffusive force field), the closer the Lagrangian and Eulerian momentum barriers are expected to be to their passive counterparts. Similarly, the more correlated the velocity field is with the vorticity Laplacian (i.e., with the curl of the diffusive force field), the closer the Lagrangian and Eulerian vorticity barriers will be to their passive counterparts.

As another extreme case, inviscid flows have barrier equations with identically vanishing right-hand sides. This is because there is no viscous transport in such flows and hence active barriers are not well-defined. As a consequence,  aFTLE and aPRA will identically vanish  for such flows, while  passive FTLE and PRA will only vanish if the inviscid flow has a spatially independent velocity field. Therefore, the more inviscid the flow is in a region, the more its active and passive barriers will differ from each other in that region.

A notable case between Beltrami and inviscid flows is a Lamb--Oseen velocity field modeling a vortex decaying due to viscosity (Saffman et al. 1992). Along each cylindrical streamsurface surrounding the origin in this flow, the viscous force is a constant negative multiple of the velocity at any given time. This immediately implies that all Eulerian active and passive barriers to momentum transport coincide with the cylindrical streamsurfaces of Lamb--Oseen vortices, even though their velocity field is not Beltrami. Indeed, we do find in our upcoming 2D and 3D turbulence examples that strong enough vortices have very similar overall signatures in the active and the passive LCS diagnostic fields, with the former field providing more detail. This stands in contrast to hyperbolic mixing regions outside those vortices, in which active and passive LCS diagnostics may differ substantially.

\section{Active barriers in specific unsteady flows}

In this section, we illustrate the numerical implementation of our results  and the use of active LCS diagnostics (see section \ref{sec:Practical-implementation-on}) on
2D homogeneous, isotropic turbulence and a 3D turbulent channel flow. The scripts we have used to compute active barriers in these examples can be downloaded from https://github.com/LCSETH?tab=repositories.

\subsection{Two-dimensional homogeneous, isotropic turbulence\label{subsec:Two-dimensional-turbulence}}

Here we evaluate our 2D results from section \ref{subsec:active barriers in 2D-Navier=002013Stokes-flows}
on active barriers in a 2D turbulence simulation over a spatially
periodic domain $U=[0,2\pi]\times[0,2\pi]$. Since all computations will be two-dimensional in this section, we drop the hat from the notation we used in section \ref{subsec:active barriers in 2D-Navier=002013Stokes-flows} for 2D variables. 

Obtained from a pseudo-spectral code applied to the 2D, incompressible Navier-Stokes
equation (see Farazmand, Kevlahan \& Protas 2011), the spatial coordinates
are resolved using $1024^{2}$ Fourier modes with $2/3$ dealiasing.
The viscosity is $\nu=2\times10^{-5}$. This data set comprises $251$
equally spaced velocity field snapshots spanning the time interval
$[0,50]$. Whenever a numerical integration scheme is required, i.e.,
advection of particles and integration of the barrier fields, the
Runge-Kutta 4 algorithm is employed. The same data set was already
analyzed by Katsanoulis et al. (2019), who located vortex boundaries
as barriers to the diffusive transport of vorticity using the theory
of constrained diffusion barriers from Haller, Karrasch \& Kogelbauer
(2019). In contrast, here we use appropriate 2D, steady barrier equations
\eqref{eq:2D incompressible NS Lagrangian momentum barrier eq}-\eqref{eq:2D incompressible NS Eulerian momentum barrier eq}
and \eqref{eq:2D incompressible NS Lagrangian vorticity barrier eq}-\eqref{eq:2D incompressible NS Eulerian vorticity barrier eq}
(see also Remark \ref{rem:used active LCS diagnostics instead of level curves})
to visualize Lagrangian and objective Eulerian coherent vortices as
regions bounded by maximal barriers to active transport.

\subsubsection{Eulerian active barriers\label{subsec:Eulerian-barriers-to}}

For the instantaneous barrier calculations, we use the first snapshot
of the dataset at time $t=0$ and we compute the right-hand side of
eqs. \eqref{eq:2D incompressible NS Eulerian momentum barrier eq}
and \eqref{eq:2D incompressible NS Eulerian vorticity barrier eq}
using a grid of $1024\times1024$ points. In our experience, this
grid spacing is much smaller than the size of the coherent vortices
in this flow. As a consequence, the results do not change appreciably
under further grid refinements, as long as one targets structurally
stable objects in the Lagrangian particle dynamics, as we do (see
Definition \ref{def:barrier definition}). For vorticity
barriers, we use the 2D version of eq. \eqref{eq:NS Eulerian vorticity barrier eq-1-1}
to illustrate the computational procedure for barriers in lower-resolved
data. We then proceed to compute the aFTLE and aPRA for both the momentum
and vorticity barrier fields from eqs. \eqref{eq:aFTLE field -- Eulerian}
and \eqref{eq:cos theta-1 2D} using a central finite-differencing
scheme for the active flow map gradient required in eq. \eqref{eq:active CG tensor}.

We focus on the region $\left[2.8,4.9\right]\times\left[1,3\right]$ of the full computational domain
to illustrate the level of spatial detail we obtain from instantaneous velocity data (see fig. \ref{fig: 2D turbulence - FTLE-aFTLE Eulerian comparison}). We
note the striking differences in the quality of the delineated structures
between the instantaneous limit of the passive FTLE and the momentum-based
aFTLE of figure \ref{fig: 2D turbulence - FTLE-aFTLE Eulerian comparison}.
Advective LCSs tend to have relatively weak signatures in the instantaneous limit of the FTLE field (see formula \eqref{eq:instantaneous limit of passive FTLE}) which is given by the dominant rate-of-strain eigenvalue field.
In contrast, active barriers remain sharply defined in the aFTLE fields, which offer increasing refinement of the flow features under increasing $s$-times. The only limitation to this refinement is the resolution
of the available data. This is apparent in the vorticity-based aFTLE
in figure \ref{fig: 2D turbulence - FTLE-aFTLE Eulerian comparison}c,
where the improvement is more modest, given that higher-order spatial
derivatives need to be computed from the same data set. 

\noindent 
\begin{figure}
	\begin{centering}
		\includegraphics[width=1\textwidth]{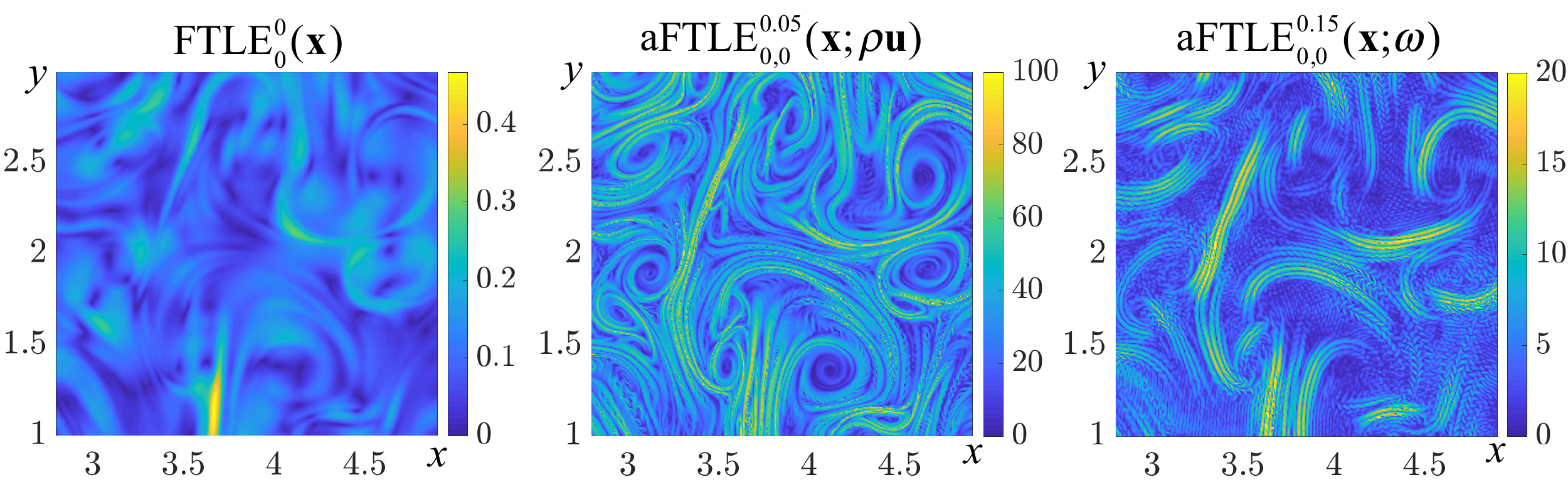}
		\par\end{centering}
	\caption{Comparison of the $t=0$ instantaneous limits of the passive FTLE, the aFTLE with respect to $\rho\mathbf{u}$ with $s=0.05$ and the aFTLE with respect to $\boldsymbol{\mathbf{\omega}}$ with $s=0.15$ in our 2D turbulence example. \label{fig: 2D turbulence - FTLE-aFTLE Eulerian comparison}}
\end{figure}

Figure \ref{fig: 2D turbulence - FTLE-aFTLE Eulerian comparison_momentum_closeup} focuses on momentum-based active barriers in one of the vortical regions revealed by  figure \ref{fig: 2D turbulence - FTLE-aFTLE Eulerian comparison}.  The aFTLE provides a clear demarcation of the main vortex, which becomes
even more pronounced for longer $s$-times, revealing secondary vortices
around its neighborhood. In contrast, none of these vortices are present
in the passive FTLE in figure \ref{fig: 2D turbulence - FTLE-aFTLE Eulerian comparison_momentum_closeup}.
A similar result emerges when the same region is analyzed using the
aPRA field in the same figure. Specifically, the effect of progressive refinement with increasing
$s$-times is more prominent here as a number
of elliptic structures become visible in the main vortical region. In contrast, the instantaneous limit of the passive PRA returns identically zero values, as the instantaneous limit of all polar rotation angles is zero by definition.

\noindent 
\begin{figure}
	\begin{centering}
		\includegraphics[width=1\textwidth]{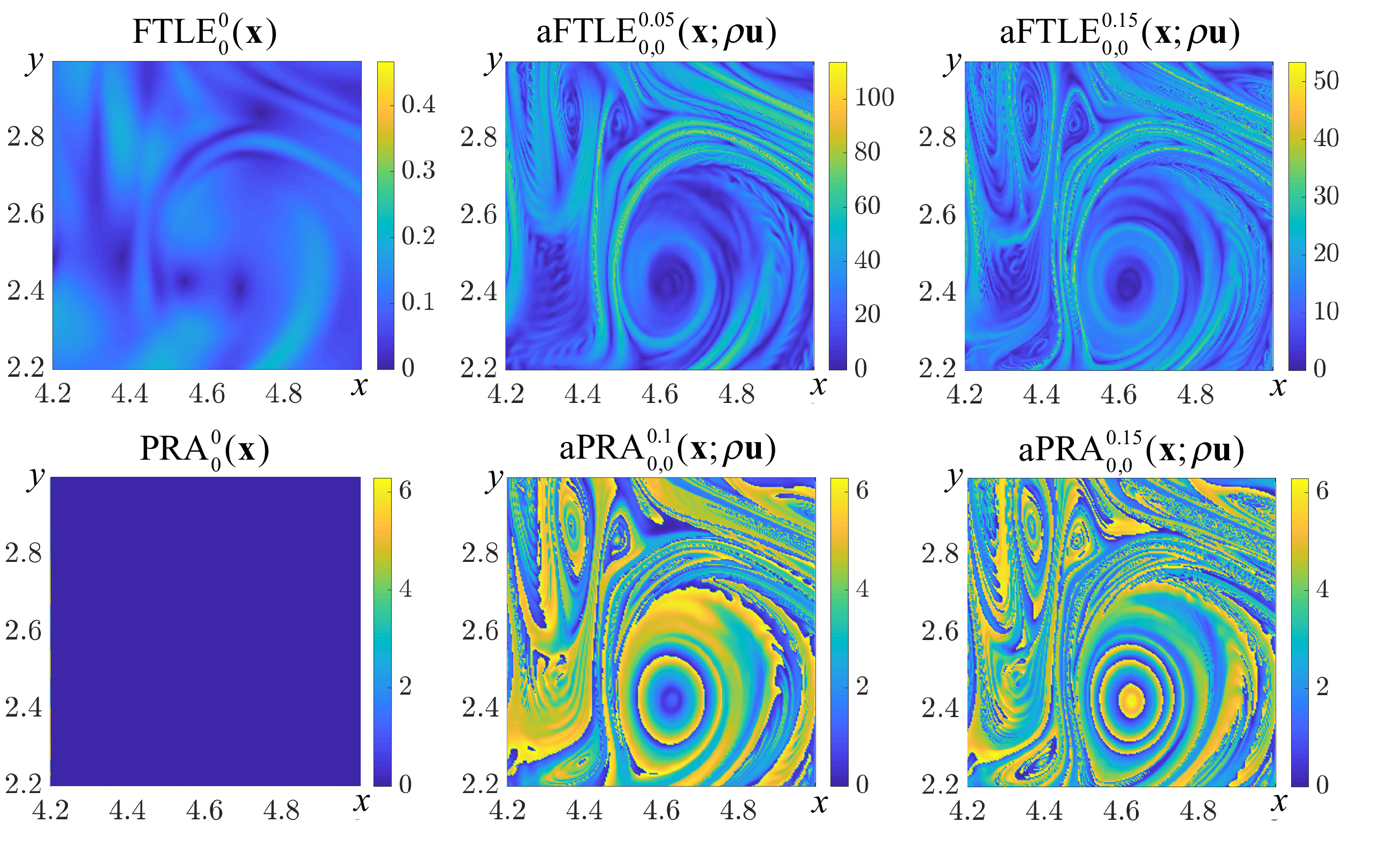}
		\par\end{centering}
	\caption{Comparison between the instantaneous limit of the passive FTLE (PRA)
		and the instantaneous limit of the momentum-based aFTLE  (PRA) fields  for   $s=0.05$ and for  $s=0.15$ in one of the vortical
		regions of our 2D turbulence example at $t=0$. \label{fig: 2D turbulence - FTLE-aFTLE Eulerian comparison_momentum_closeup}}
\end{figure}

\subsubsection{Lagrangian active barriers}

For the Lagrangian computations in this example, we use
the same, slightly oversampled grid of Katsanoulis et al. (2019) with
$1100\times1100$ equally spaced initial conditions and we advect
them over the time interval $[0,25]$ using all the available velocity
snapshots. To compute the required Lagrangian averages along trajectories,
we use $25$ snapshots of the appropriate quantities as using more
snapshots does not bring any noticeable changes to the resulting barrier
fields. Based on that, we compute the expressions for the active barrier
fields from eqs. \eqref{eq:2D incompressible NS Lagrangian momentum barrier eq}
and \eqref{eq:2D incompressible NS Lagrangian vorticity barrier eq}, which we then use for the evaluation of the aFTLE and aPRA.

Comparisons between these scalar diagnostic fields and the
passive FTLE and PRA are shown in figure \ref{fig: 2D turbulence - FTLE-aFTLE Lagrangian comparison}. We  observe that the momentum-based aFTLE and aPRA reveal structures
inside the vortical regions in much finer detail, as they do not
rely on substantial fluid particle separation. In agreement with our arguments in section \ref{relationship between active and passive LCS diagnostics}, aFTLE and aPRA consistently refine the same coherent vortices indicated by FTLE and PRA. In the mixing regions surrounding those vortices, however, active and passive LCS diagnostics tend to identify different barriers. As in the case of our
Eulerian barrier calculations in section \ref{subsec:Eulerian-barriers-to},
the vorticity-based aFTLE and aPRA provide a more moderate enhancement,
because they rely on second derivatives of the velocity data. 
\noindent 
\begin{figure}
	\begin{centering}
		\includegraphics[width=1\textwidth]{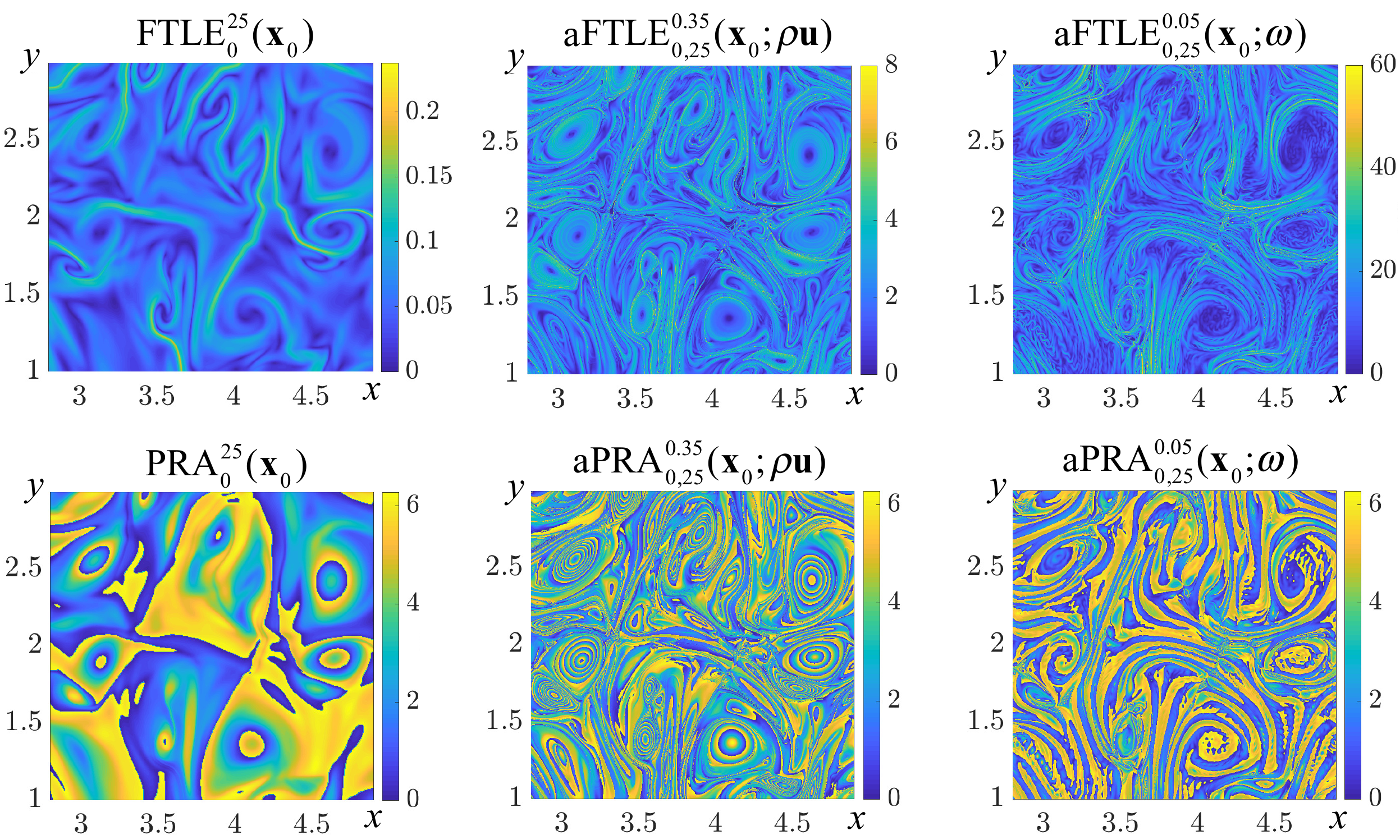}
		\par\end{centering}
	\caption{Comparison between  the passive FTLE (PRA) and the momentum- and vorticity-based aFTLE (aPRA).
		All computations were performed over the time interval $[t_{0},t_{1}]=[0,25]$
		on the domain $\left[2.8,4.9\right]\times\left[1,3\right]$. \label{fig: 2D turbulence - FTLE-aFTLE Lagrangian comparison}}
\end{figure}

 Next, we illustrate the extraction of active barriers to
the transport of momentum and vorticity as parametric curves. This is possible in 2D incompressible flows because the active barrier equations are Hamiltonian, and hence the barriers are level curves of a scalar
function (see Theorems \ref{thm:2D momentum barriers} and \ref{thm:2D vorticity barriers},
as well as Remark \ref{rem:used active LCS diagnostics instead of level curves}).
To perform this extraction, we follow the method presented in Haller
et al. (2016) for the extraction of coherent Lagrangian vortex boundaries
as outermost level sets of the Lagrangian-averaged vorticity deviation
(LAVD). We will use the notation $H_{t_{0}}^{t_{1}}(\mathbf{x}_0)$
to denote the relevant Hamiltonian from section \ref{subsec:active barriers in 2D-Navier=002013Stokes-flows}. The algorithm is the same for all those Hamiltonians, but
we will restrict our computations here to the Hamiltonian governing Lagrangian momentum-barriers in 2D, given by $H_{t_0}^{t_1}(\mathbf{x}_{0})=\nu\rho\, \overline{\omega\left(\mathbf{F}_{t_{0}}^{t}\left(\mathbf{x}_{0}\right),t\right)}$ (see eq.  \eqref{eq:2D incompressible NS Lagrangian momentum barrier eq}).

In all our computations, we will focus on finding almost
convex structurally stable level sets of $H_{t_{0}}^{t_{1}}(\mathbf{x}_0)$
that encircle a single local maximum of $\left|H_{t_{0}}^{t_{1}}(\mathbf{x}_0)\right|$.
The need for relaxation of the strict convexity requirement in discrete
data sets is discussed extensively in Haller et al. (2016), so we
will skip it here. Along these lines, we introduce the convexity deficiency
of a closed curve in the plane as the ratio of the area between the
curve and its convex hull to the area enclosed by the curve, which
we denote with $d_{max}$. The maximum $d_{max}$ we used for the
different extracted barriers was $5\times10^{-2}$.

Small-scale local maxima of $\left|H_{t_{0}}^{t_{1}}(\mathbf{x}_0)\right|$
may appear either due to non-accurate resolution of these scales or
because of computational noise. To address this issue, we only considered
boundaries with arclength larger than a threshold $l_{min}$. This
threshold was set to $0.4$ for all our  computations
because below this limit, boundaries contain too few grid points to
be considered well-resolved.

The main steps of the extraction procedure are delineated
in Algorithm 1. All the MATLAB
codes used for the extraction of the barriers of this section can
be found in the on-line supplementary materials.

\noindent \begin{algorithm}[H] 	\caption{Coherent Lagrangian and Eulerian vortex boundaries for two-dimensional flows} 	\label{alg:2DBarriers}{} 	\textbf{Input:} A time-resolved two-dimensional velocity field (or a snapshot thereof in the Eulerian case). 	\begin{enumerate} 		\item For a grid of initial conditions $\mathbf{x}_{0}$, compute the $H_{t_{0}}^{t_{1}}(\mathbf{x}_0)$. 		\item Find local maxima of  $\left|H_{t_{0}}^{t_{1}}(\mathbf{x}_0)\right|$. 		\item Detect initial vortex boundaries as outermost closed contours of $H_{t_{0}}^{t_{1}}(\mathbf{x}_0)$ satisfying the following: 		\begin{enumerate} 			\item The boundary encircles a local maximum of $\left|H_{t_{0}}^{t_{1}}(\mathbf{x}_0)\right|$. 			\item The boundary has convexity deficiency less than a bound $d_{max}$. 			\item The boundary has arclength exceeding a threshold $l_{min}$. 		\end{enumerate} 	\end{enumerate} 	\textbf{Output:} Initial positions of coherent Lagrangian or Eulerian vortex boundaries. \end{algorithm} 

We apply this algorithm to extract an active material
barrier to the transport of momentum with high precision as a parametrized
curve. This closed active barrier is shown in red in figure \ref{fig: 2D turbulence - extracted momentum barrier, Eulerian}. We also show the impact of this barrier on the momentum landscape
in Eulerian and Lagrangian coordinates, respectively, for the initial and final
times in $\left[0,25\right]$. Furthermore, for reference,  we show an elliptic LCS (black) extracted as a closed level curve of the passive PRA through a selected point of the active barrier. As expected from our discussion in section \ref{relationship between active and passive LCS diagnostics}, these active and passive elliptic barriers are very close to each other at the initial time and remain equally close during their material evolution. In the Eulerian
frame, we observe that the extracted active and passive  barriers shows no sign of filamentation throughout their whole extraction time. This
is in agreement with the general expectation we stated earlier for
diffusion-minimizing material curves. Furthermore, when viewed in
the Lagrangian frame, we note the organizing role of the extracted
barrier in the momentum landscape. Indeed, the barrier keeps encapsulating
small values of the momentum norm for the entire extraction time.
\begin{figure}
	\begin{centering}
		\includegraphics[width=1\textwidth]{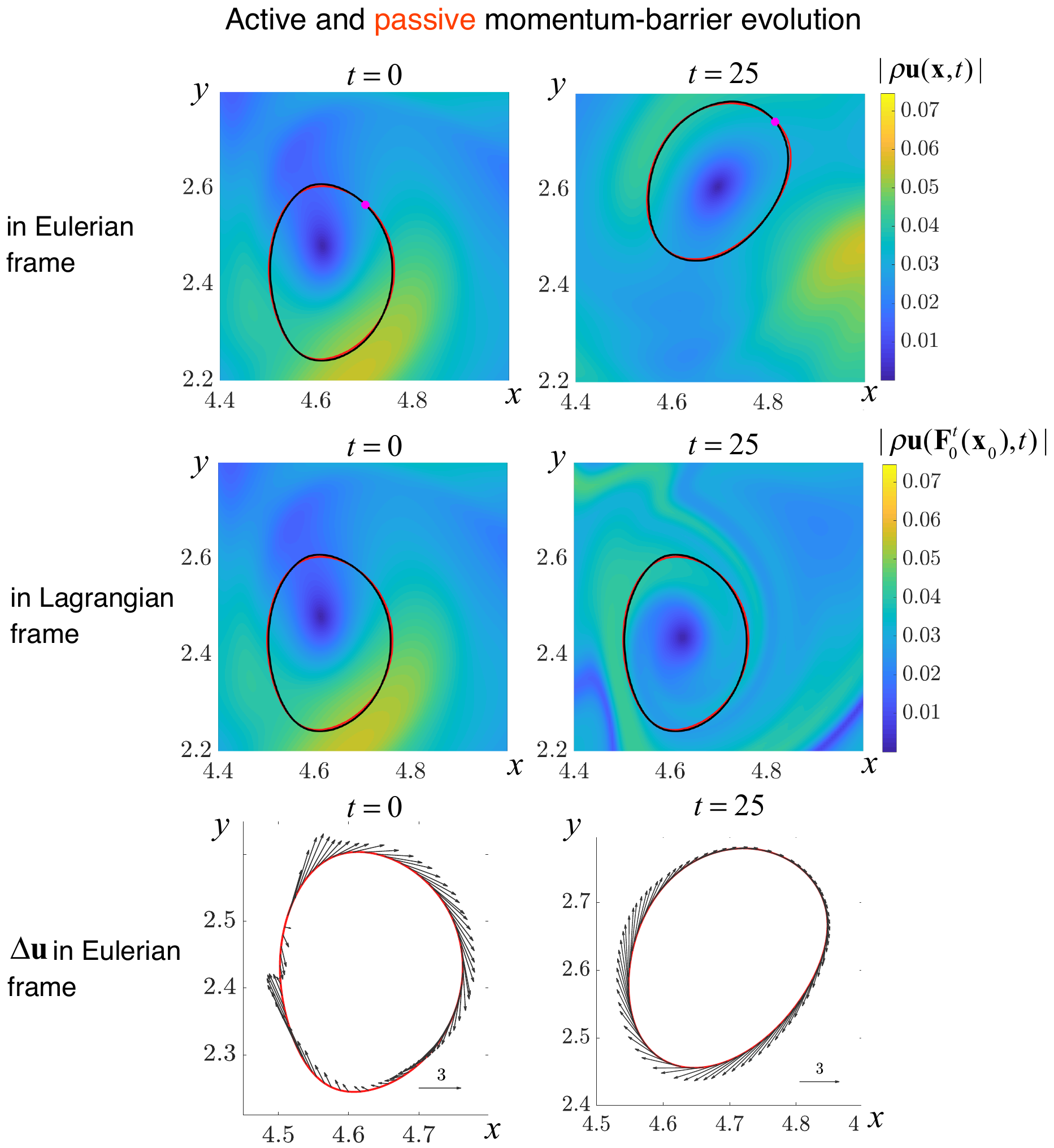}
		\par\end{centering}
	\caption{Evolution of an extracted active material barrier (red) to the diffusive transport
		of momentum in the Eulerian and the Lagrangian frame, superimposed on the distribution of the norm of the linear momentum. This barrier was identified as a level curve of the Hamiltonian $H_{t_0}^{t_1}(\mathbf{x}_{0})=\nu\rho\, \overline{\omega\left(\mathbf{F}_{t_{0}}^{t}\left(\mathbf{x}_{0}\right),t\right)}$. Black curve indicates an elliptic LCS extracted as a level curve of the passive PRA, launched from the highlighted point of the active barrier. Also shown are the instantaneous viscous forces (normalized by $\rho\nu$) acting on the evolving barrier.\label{fig: 2D turbulence - extracted momentum barrier, Eulerian}}
\end{figure}

 Figure \ref{fig: 2D turbulence - extracted momentum barrier, Eulerian}
also shows the instantaneous viscous force (normalized by $\rho\nu$) along the extracted active momentum
barrier. Note that this force remains
almost tangent to the barrier for the most part. There are, however,
some notable exceptions, illustrating that these barriers are not
constructed to be tangent to the viscous forces at every time instance.
Rather, the viscous forces are tangent to the barriers in a time-averaged
sense after being pulled back under the flow map to the initial configuration.

\subsection{Three-dimensional turbulent channel flow\label{subsec:channelflow}}

We consider now the 3D incompressible, turbulent flow of a Newtonian
fluid in a doubly periodic channel, a well-studied physical setting for 3D coherent structure studies.

Our analysis relies on velocity snapshots from a mixed-discretization parallel 
solver of the incompressible Navier--Stokes equations in the wall-normal
velocity and vorticity formulation, developed by Luchini \& Quadrio
(2006). The equations of motions are discretized via a Fourier--Galerkin
approach along the two statistically homogeneous streamwise $\left(x\right)$ and
spanwise $\left(z\right)$ directions. Fourth-order compact finite differences
(Lele 1992) based on a five-point computational stencil are adopted
in the wall-normal direction $\left(y\right)$.

The governing equations are integrated forward in time at constant
power input (Hasegawa et al. 2014) with a partially-implicit approach, combining the
three-step, low-storage Runge\textendash Kutta (RK3) scheme with the
implicit Crank\textendash Nicolson scheme for the viscous terms. The
friction Reynolds number is $Re_{\tau}=u_{\tau}h/\nu=200$, based
on the friction velocity $u_{\tau}$, the channel half height $h$
and the kinematic viscosity $\nu$, which corresponds to a bulk Reynolds
number $Re_{b}=U_{b}h/\nu=3177$, where $U_{b}$ is the bulk velocity. 
The computational domain is $L_{x}=4\pi h$ long
and $L_{z}=2\pi h$ wide. The number of Fourier modes is 256 both
in the streamwise and spanwise direction; the number of points in
the wall-normal direction is 256, unevenly spaced in order to decrease
the grid size near the walls. The corresponding spatial resolution
in the homogeneous directions is $\Delta x^{+}=9.8$ and $\Delta z^{+}=4.9$ 
(without accounting for the additional modes required for dealiasing
according to the 3/2 rule); the 
wall-normal resolution increases from $\Delta y^{+}=0.4$ near the walls to 
$\Delta y^{+}=2.6$ at the centreline,
while the temporal resolution is kept constant at $\Delta t=0.005 h/U_b$, corresponding to $\Delta t^{+}=0.063$.
The superscript $+$ denotes nondimensionalization in viscous units,
i.e. with $u_{\tau}$ and $\nu$. At each DNS timestep and thus with the same 
temporal resolution, a three-dimensional flow snapshot is stored for a total 
of 1500 snapshots. The 750$^\mathrm{th}$ snapshot in the series is stored at time
$t=0$. This is the instant at which we compute the Eulerian barriers to active transport. The last 750 snapshots are utilised for the computation 
of the active barriers and passive forward FTLE, while the first 750 ones are used 
for calculating the passive backward FTLE. The integration time for the Lagrangian calculations
has been chosen based on pair-dispersion statistics of Lagrangian tracers (see, for instance, Pitton et al.
2012). The averaging time for the 
bulk flow statistics is  8100 $U_b / h$. In the following, 
all quantities are nondimensionalized using  $U_{b}$ and $h$ unless stated
otherwise.

The active barriers are computed in a two-step procedure. First, the active 
barrier field $\mathbf{b}_{t_0}^{t_1}\left(\mathbf{x}_0\right)$, appearing at the right-hand 
side of the barrier equation \eqref{eq:barrier equation}, is computed. Then, the barrier 
ODE is solved and visualised via the FTLE and PRA diagnostics. 

\begin{figure}
  \centering
  \includegraphics[width=0.9\textwidth]{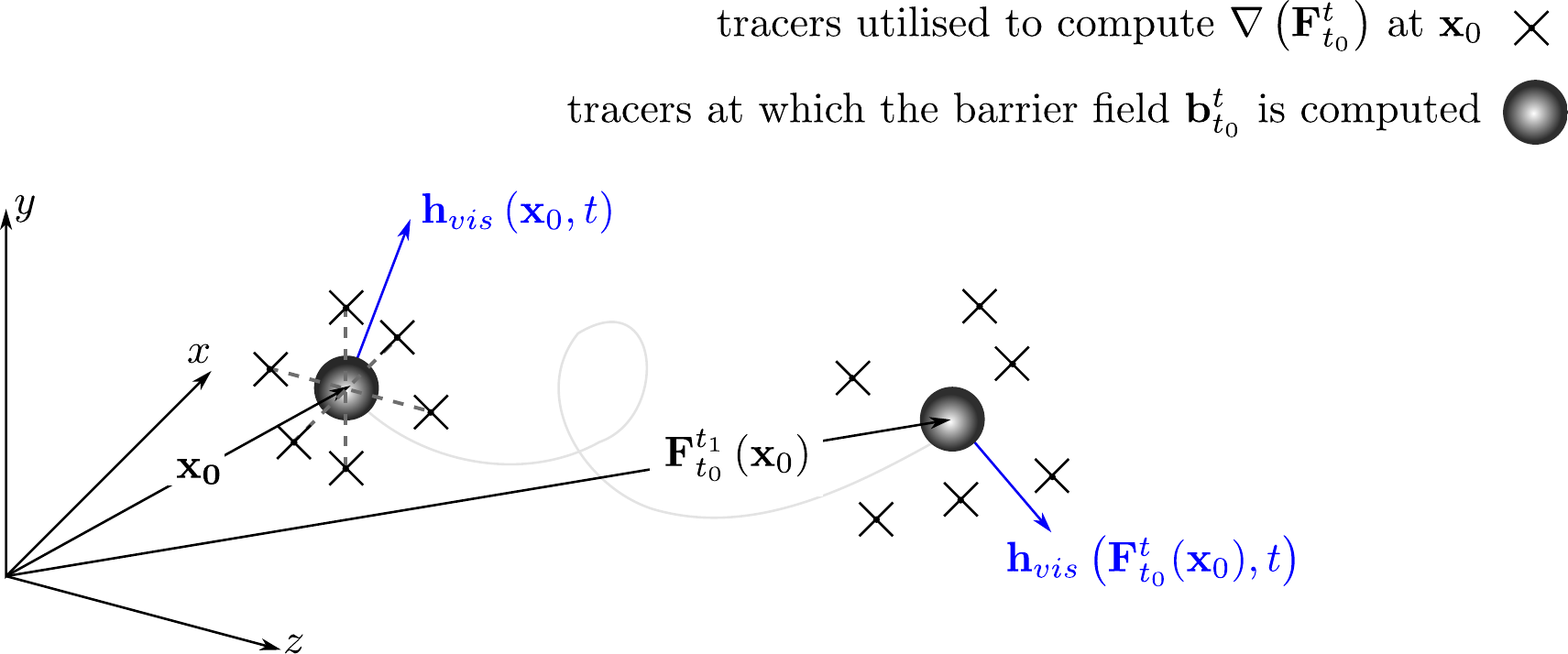}
  \caption{Sketch of the computational molecule utilised for the computation of the active barrier field $\mathbf{b}_{t_0}^{t_1}$. The large circle denotes the Lagrangian tracer at the center of the molecule, where the vector field  $\mathbf{h}_{vis}$ is computed. The cross denotes the further six tracers utilised to compute $\nabla\mathbf{F}_{t_0}^{t_1}$. }
  \label{fig:seed}
\end{figure}
For Eulerian barriers, the barrier vector field appearing in the instantaneous (or Eulerian) barrier equation \eqref{eq:Eulerian barrier equation} is readily computed from the velocity field data
as $\mathbf{b}_{t}^{t}=\mathbf{h}_{vis}$. Differentiation of the velocity field is  performed with 
the same discrete operators used during DNS. For material barriers, $\mathbf{b}_{t_0}^{t_1}\left(\mathbf{x}_0\right)$ is simply obtained as the temporal average of 
 $\left(\mathbf{F}_{t_0}^{t_1}\right)^\ast \mathbf{h}_{vis}$, because $\det\nabla \mathbf{F}_{t_0}^{t_1}\left(\mathbf{x}_0\right)\equiv 1$ by incompressibility. In this case, the vector field $\mathbf{b}_{t_0}^{t_1}\left(\mathbf{x}_0\right)$ is discretised 
on a Cartesian grid similar to the one used for the velocity field; the only difference is that the
number of collocation points along the $x$- and $z$-directions is increased to 
384 via Fourier interpolation. 

At time $t_0$, a set of tracers is seeded in the 
neighbourhood of each point $\mathbf{x}_0$ at which $\mathbf{b}_{t_0}^{t_1}\left(\mathbf{x}_0\right)$ 
needs to be computed. Each set (see figure~\ref{fig:seed}) is composed by  7 tracers. The central 
tracer is exactly located at $\mathbf{x}_0$ and is the only one along which the vector field $\mathbf{h}_{vis}$ is also computed.
The other tracers are shifted by $\epsilon_i$ along the positive and negative $i$th spatial direction and are
utilised to compute $\nabla \mathbf{F}_{t_0}^{t_1}\left(\mathbf{x}_0\right)$ with second-order central finite differences. 
The shift $\epsilon_i$ is defined as $1/100$ of the minimum grid spacing along the $i$th spatial direction. 
A total of $2.64\times 10^8$ particles are seeded into the flow. The evolving positions of these
tracers, which are images of their initial positions under the flow
map $\mathbf{F}_{t_0}^{t_1}$, are advanced in time by integrating the
$\mathbf{u}$ field with a third-order, four-stage Runge--Kutta algorithm. The vector fields $\mathbf{u}$ and 
$\mathbf{h}_{vis}$ required at the intermediate stages are obtained via linear interpolation of two consecutive flow snapshots 
and are evaluated at the particle position through a sixth-order, three-dimensional
Lagrangian interpolation (van Hinsberg et al. 2012, Pitton et al.
2012) of the underlying vector field, which reduces to fourth-order
only between the wall and the first grid point above it.

Once the (Lagrangian or Eulerian) barrier equation is available, its active flow map $\boldsymbol{\mathcal{F}}_{t_0,t_1}^{s}$
is computed by solving the steady barrier ODE up to  $s_{\mathrm{max}}=31.0$ and  $s_{\mathrm{max}}=0.62$
 for the momentum and vorticity barriers, respectively. We have chosen these $s$-times large enough for
the computed barrier trajectories to reveal enough detail in the underlying barrier vector field but small enough
to avoid accumulation of the integration error. The effect of changing the parameter $s_{\mathrm{max}}$ is shown in the
additional material \textsf{Movie 1.mp4} and \textsf{Movie 2.mp4}
for Eulerian momentum and vorticity barriers, respectively. The seeds for the flow map are
arranged in a Cartesian grid identical to the one of the active barrier field for 3D 
computations of aFTLE/aPRA diagnostics, while the spatial resolution is increased by a factor 6 when only two-dimensional
slices are computed. The comparison between the two resolutions has been utilised to verify the grid-independence of the results. The aFTLE and aPRA diagnostics are then computed according to equations \eqref{eq:aFTLE -- Lagrangian} and \eqref{eq:cos theta-1},  respectively. 

\subsubsection{Eulerian active barriers} \label{sec:i3d}

Instantaneous aFTLE and aPRA are presented for $t=0$ in figures \ref{fig:channel_iFTLE} and 
\ref{fig:channel_iPRA}, respectively, and compared against their passive variants. Even though the results are computed for 
the complete three-dimensional field, only two-dimensional cross sections are presented in the following. 
In figures \ref{fig:channel_iFTLE} and \ref{fig:channel_iPRA}, a $\left(y-z\right)$ cross-section located at $x=2 \pi h$ is shown.
The 3D visualisation of the FTLE and PRA fields poses challenges in its own, which are subjects of ongoing research
in computer visualisation (see, e.g., Sadlo \& Peikert 2009, Schindler et al. 2012) and are outside the scope of the present study. 
The 2D visualization in different cross sections results in different local flow structures which are, nevertheless, all reminiscent of classically known structures in  channel flows. These include
low-speed streaks, quasi-streamwise and hairpin vortices and packets thereof (Robinson 1991). The supplementary materials 
\textsf{Movie 3.mp4} and \textsf{Movie 4.mp4} show how figure \ref{fig:channel_iFTLE}(b-c) change throughout the channel length. 

\begin{figure}
  \includegraphics[]{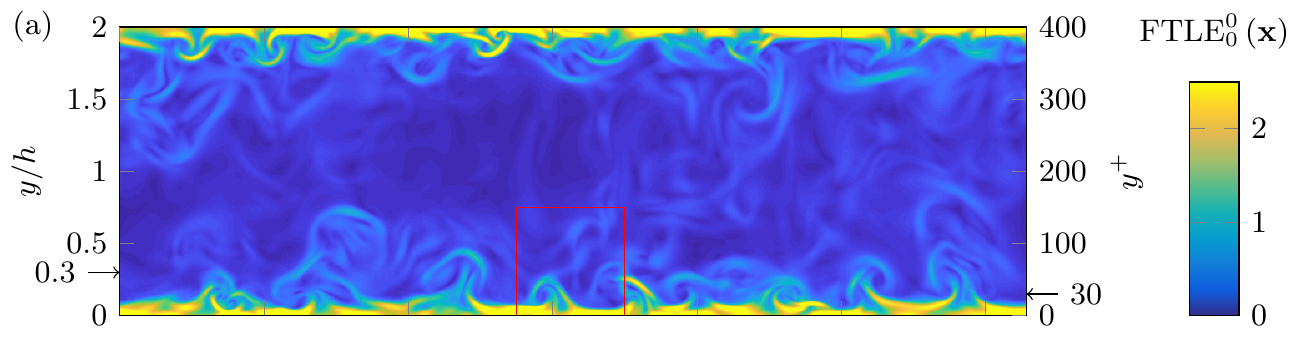}\\
  \includegraphics[]{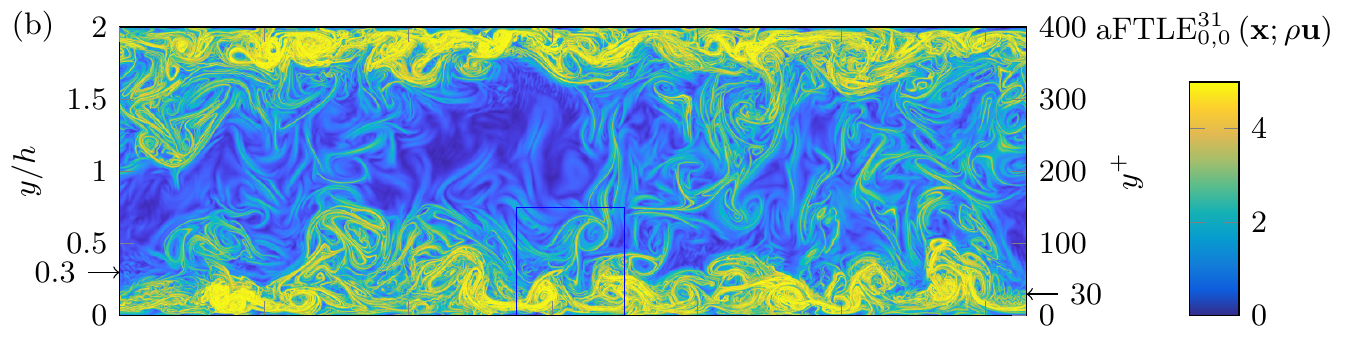}\\
  \includegraphics[]{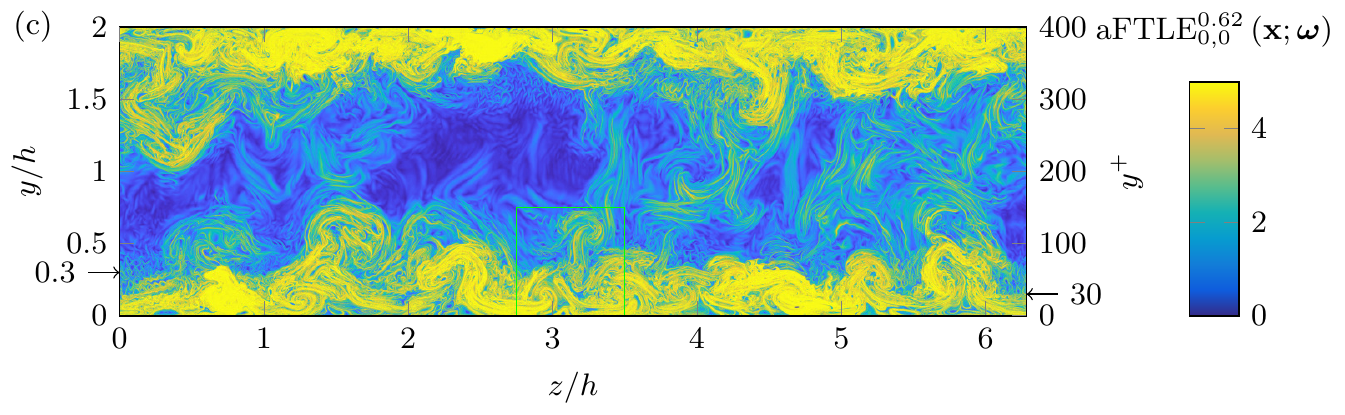}\\
  \hfill\includegraphics[]{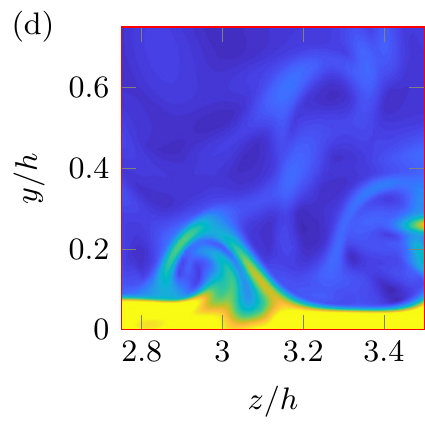}\hfill\includegraphics[]{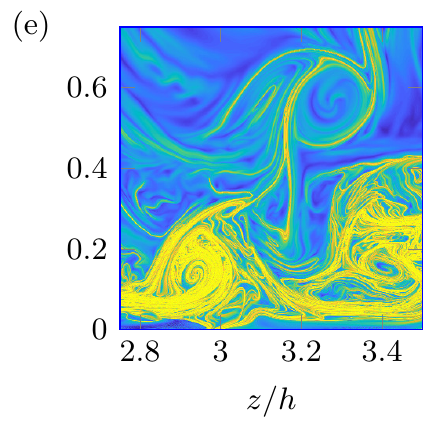}\hfill\includegraphics[]{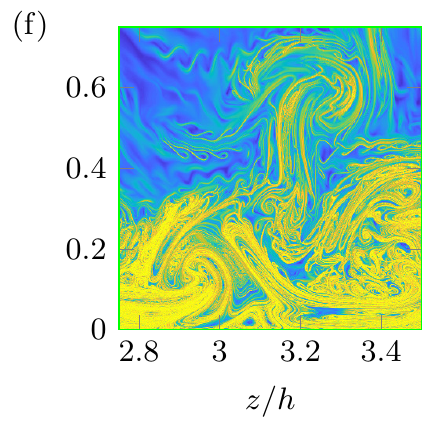}\hfill
  \caption{Comparison between the instantaneous limit of (a,d) the passive FTLE, 
           (b,e) the aFTLE with respect to $\rho \mathbf{u}$ and (c,f) the 
           aFTLE with respect to $\bm{\omega}$ at $t=0$ in a cross-sectional
           plane at $x/h = 2\pi$. The panels (d-f) magnify the region denoted with a rectangle in panels (a-c). All computations in the figure were performed on the same snapshot of the velocity field at $t=0$. }
  \label{fig:channel_iFTLE}
\end{figure}
\begin{figure}
  \includegraphics[]{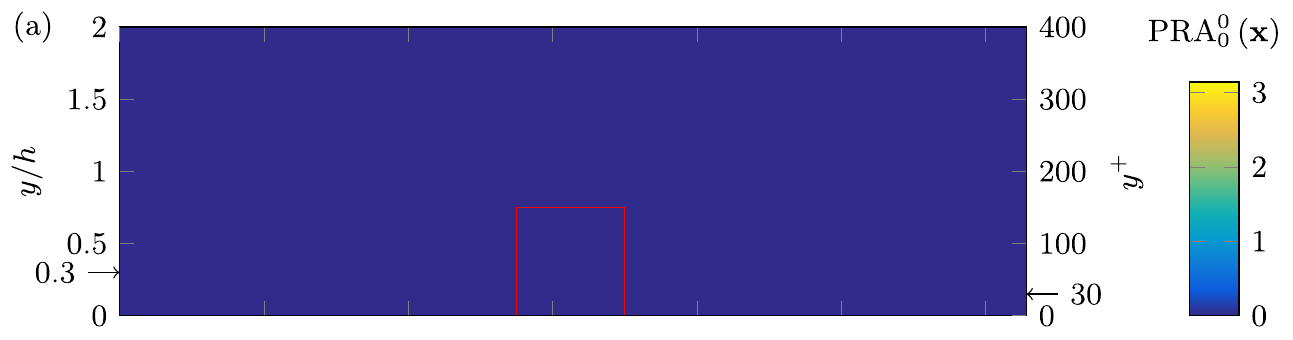}\\
  \includegraphics[]{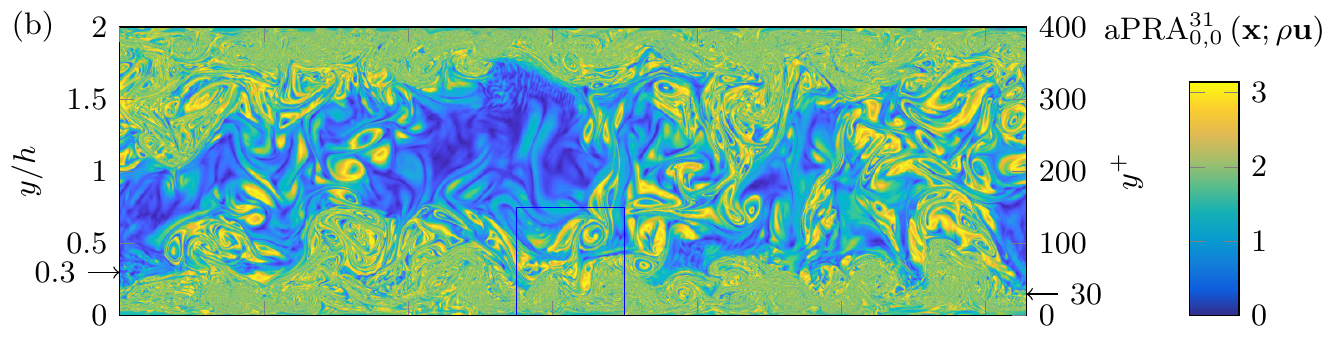}\\
  \includegraphics[]{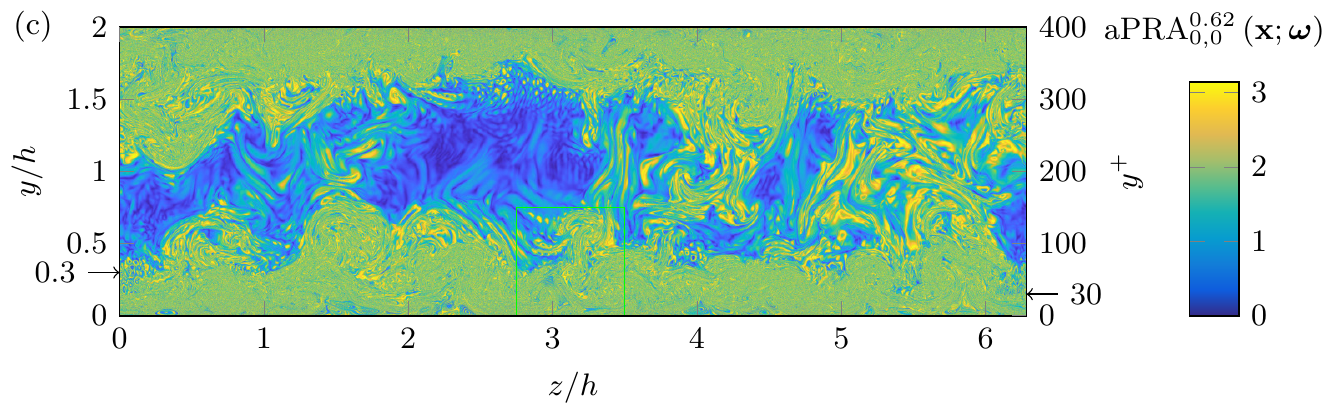}\\
  \hfill\includegraphics[]{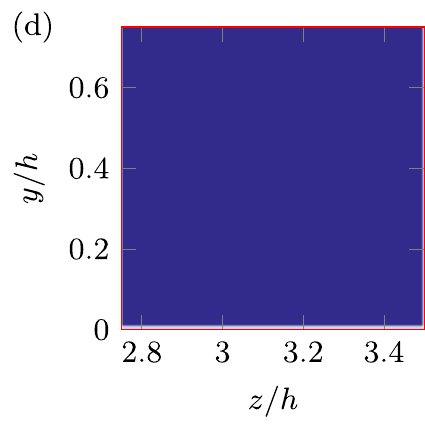}\hfill\includegraphics[]{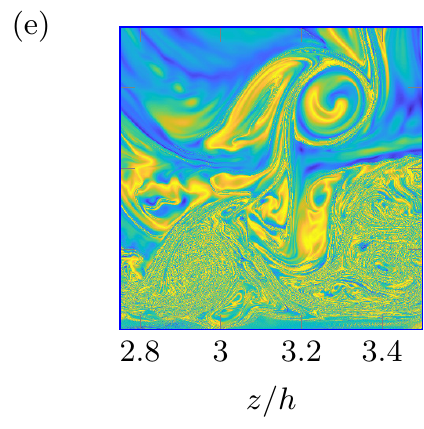}\hfill\includegraphics[]{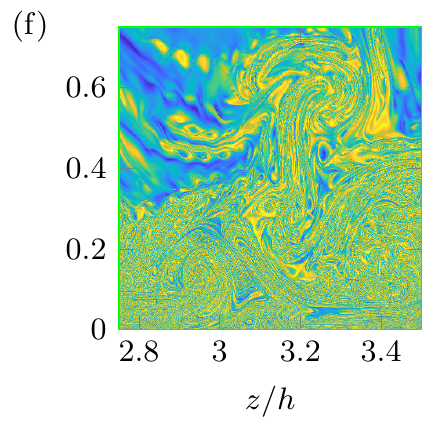}\hfill
  \caption{Comparison between the instantaneous limit of (a,d) the passive PRA (which is identically zero), 
           (b,e) the aPRA with respect to $\rho \mathbf{u}$ and (c,f) the 
           aPRA with respect to $\bm{\omega}$ at $t=0$ in a cross-sectional
           plane at $x/h = 2\pi$. The panels (d-f) magnify the region denoted with a rectangle in panels (a-c). All computations in the figure were performed on the same snapshot of the velocity field at $t=0$. }
  \label{fig:channel_iPRA}
\end{figure}

As already seen for 2D turbulence in \S{\ref{subsec:Two-dimensional-turbulence}},
the aFTLE and aPRA  highlight
a broader range of structures in more detail from the same velocity data when compared to their passive variants. (Recall that the instantaneous limit of the passive PRA, in fact, vanishes identically, and hence reveals no elliptic coherent structures from a single velocity snapshot.)
The Eulerian active barriers revealed by aFTLE and aPRA appear in figures \ref{fig:channel_iFTLE} and 
\ref{fig:channel_iPRA} as an abundance of intersections of 2D surfaces with the selected cross section. Limiting to visual inspection, 
we recognise several open (or hyperbolic) barriers as ridges of the aFTLE fields. 
Given the quasi-streamwise nature of turbulent 
structures in wall-bounded flows, vortical (or elliptic) barriers to transport 
are often observed in cross-sectional planes as aFTLE ridges wrapping around closed regions, which are 
also captured
as level sets of the corresponding aPRA fields. Example of such regions
are shown in the magnifications of panels (d-f) in figures \ref{fig:channel_iFTLE} and 
\ref{fig:channel_iPRA}.

The results also reveal that large prominent aFTLE ridges penetrate into and span the bulk flow region, sometimes
connecting the channel halves, as visible in \ref{fig:channel_iFTLE}(b) between $3 \leq z/h \leq 3.5$ and $0.3 \leq y/h \leq 1.5$.
Other regions, such as between $2 \leq z/h \leq 3$ and $0.5 \leq y/h \leq 1.5$ in the same figure, display practically no discernible barriers  
and are bounded by the envelopes of filamented open (hyperbolic)
transport barriers, which are finite-time generalizations of infinite-time
classic stable and unstable manifolds. Unlike in previous approaches,
however, these finite-time invariant manifolds are constructed here
as perfect material barriers to active transport, rather than as Lagrangian
coherent structures acting as backbones of advected fluid-mass patterns
(Haller 2015).

\begin{figure}
  \centering
  \includegraphics[]{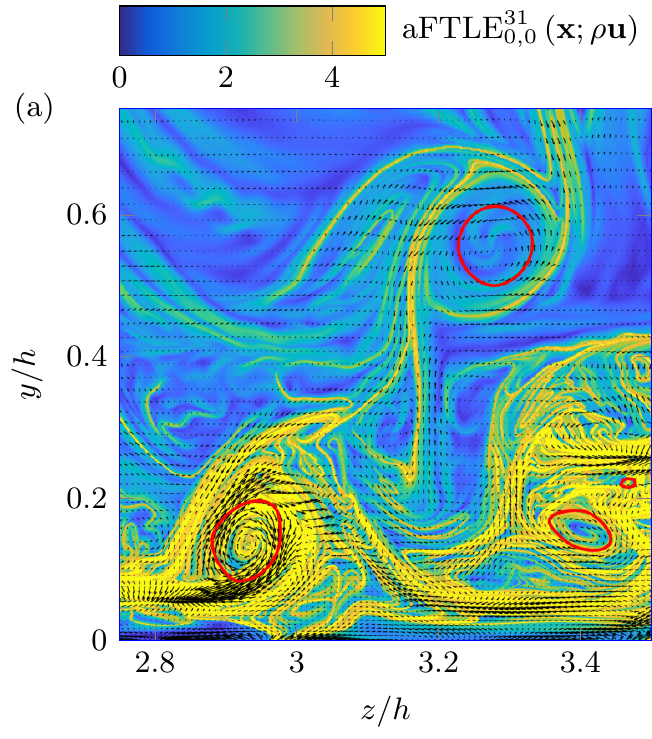}\includegraphics[]{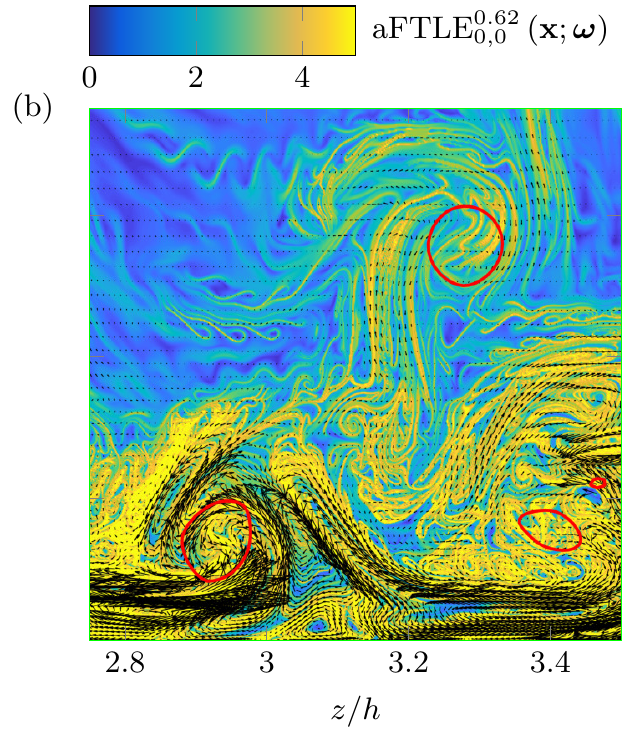}\\
  \caption{Eulerian active barriers of (a) $\rho \mathbf{u}$ and (b) $\bm{\omega}$ at $t=0$ in a cross-sectional
           plane located at $x/h = 2\pi$. The colormap shows the respective aFTLE fields, the vectors show the cross-sectional 
           components of the underlying active barrier field, while the red lines are level-set curves $\lambda_{2}^{+}(\mathbf{x},t)=-0.015$ 
           of the Eulerian vortex identification criterion proposed by Jeong et al. (1997)}.
  \label{fig:ibfield}
\end{figure}

Figure \ref{fig:ibfield} shows the Eulerian active barrier vector field of (a) $\rho\mathbf{u}$ and (b) $\bm{\omega}$
superimposed to the respective aFTLEs already shown in figure \ref{fig:channel_iFTLE}(e-f).  
Level-set curves of the $\lambda_{2}(\mathbf{x},t)=-0.015$ field (Jeong \&
Hussein 1995), a common visualization tool for coherent vortical structures in 
wall-bounded turbulence, are also shown. The scalar
field  $\lambda_{2}(\mathbf{x},t)$  is defined as the instantaneous intermediate
eigenvalue of the tensor field $\mathbf{S}^{2}(\mathbf{x},t)+\mathbf{W}^{2}(\mathbf{x},t)$,
with \textbf{$\mathbf{S}$ } and $\mathbf{W}$ defined in eq. \eqref{eq:spin and rate of strain}. This choice follows the 
heuristic convention to select a $\lambda_{2}^{+}$
value slightly below the negative of the r.m.s. peak of $\lambda_{2}(\mathbf{x},t)$
across the channel (Jeong et al. 1997), which is approximately $0.0125$
in our case.

Compared to the passive material barriers shown in figure \ref{fig:channel_iFTLE}(d), we observe that 
the active barriers not only yield a remarkably more complex flow structure but also carry a 
completely different physical meaning. Since the active barriers minimise the diffusive transport of, in this case,
 linear momentum or vorticity, we find that the cross-sectional components of the active
barrier vector field $\mathbf{b}_t^t$ are parallel to the aFTLE ridges. 
This indicates that the resultant force 
of the viscous stresses is tangential to Eulerian active barriers of momentum transport. As noted previously, momentum barriers in $\left(y-z\right)$ cross-sections
can roll-up into spiral patterns or form closed surfaces. In regions where this occurs, 
figure \ref{fig:ibfield}(a) shows that closed level-set curves of $\lambda_2$, typically used as indicators for 
the presence of quasi-streamwise vortices, tend to be found. This suggests that boundaries of 
quasi-streamwise vortices act as Eulerian active barriers to the transport of linear momentum. 
Interestingly, we find that the circulation of the active momentum barrier field in such areas is of opposite sign than the one of the velocity field.  This indicates that viscous forces oppose the vortical motion that is observed in the analyzed snapshots.  
In addition, Eulerian active barriers of vorticity tend to enter regions of closed  
momentum barriers or level-set curves of $\lambda_2$, thus highlighting regions in which vorticity diffuses into the vortex 
or is dissipated by viscosity. 

Despite some similarities, it is important to note the  practical and fundamental differences between the 
Eulerian momentum barriers and level-set surfaces of $\lambda_2$. On the fundamental side, we mention that 
the active barriers are objective, they have clear implications for the viscous transport of the active vector field and, most importantly, 
that they are extensible by definition to material barriers, thus accounting for the Lagrangian coherence of the barriers themselves. 
On the practical side, Eulerian active barriers do not require the convenient but arbitrary choice of a threshold and deliver information on the full active transport geometry, rather than just providing a few isolated curves.

\subsubsection{Lagrangian active barriers} \label{sec:l3d}

\begin{figure}
  \centering
  \includegraphics[]{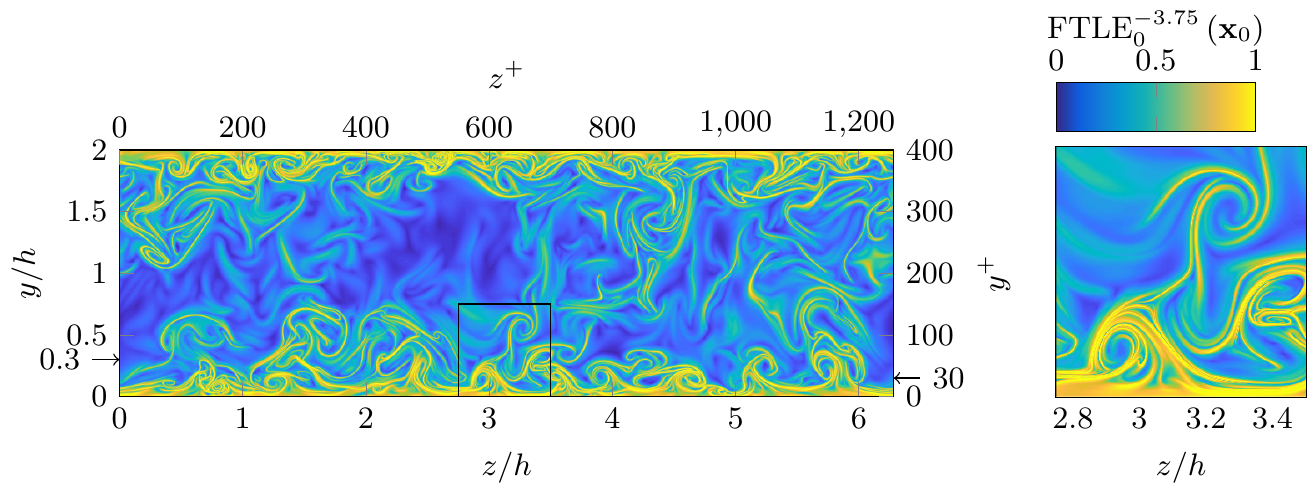}
  \caption{Passive backward $\mathrm{FTLE}_0^{-3.75}\left(\mathbf{x}_0 \right)$ in a cross-sectional
           plane at $x/h = 2\pi$. The right panel magnifies the region denoted with a rectangle in the left panel. }
  \label{fig:bFTLE}
\end{figure}
Figure \ref{fig:bFTLE} shows attracting material surfaces as passive backward  $\mathrm{FTLE}_0^{-3.75}\left(\mathbf{x}_0 \right)$ at the same $\left(y-z\right)$ cross-section located at $x=2 \pi h$ discussed in \S\ref{sec:i3d}. These attracting material surfaces, forming the cores of experimentally observed fluid trajectory patterns at time $t=0$, show a striking resemblance to the Eulerian active barriers to linear momentum indicated in figure \ref{fig:channel_iFTLE}(b,e) by the $\mathrm{aFTLE}_{0,0}^{31}\left(\mathbf{x}_0, \rho\mathbf{u} \right)$ field. The close similarity between the two is not fully surprising. At the present low value of Reynolds number viscous effects dominate throughout a significant portion of the channel, and thus determine both the characteristics of the Eulerian momentum barriers and the finite-time dynamics of particle motion. The temporal horizon, over which the analogy between  $\mathrm{FTLE}_0^{-3.75}\left(\mathbf{x}_0 \right)$ and $\mathrm{aFTLE}_{0,0}^{31}\left(\mathbf{x}_0, \rho\mathbf{u} \right)$ is observed, is expected to decrease with increasing Reynolds number, as the viscosity-dominated inner layer shrinks compared to the channel height. Whether the observed similarity holds at higher values of $Re$ is to be verified in later studies with high-$Re$ data. However, it is remarkable that the Eulerian momentum barriers, which are computed utilising a single flow snapshot, reproduce the same features of material surfaces obtained from a Lagrangian computation, which requires storing the temporal evolution of the flow. 

\begin{figure}
  \includegraphics[]{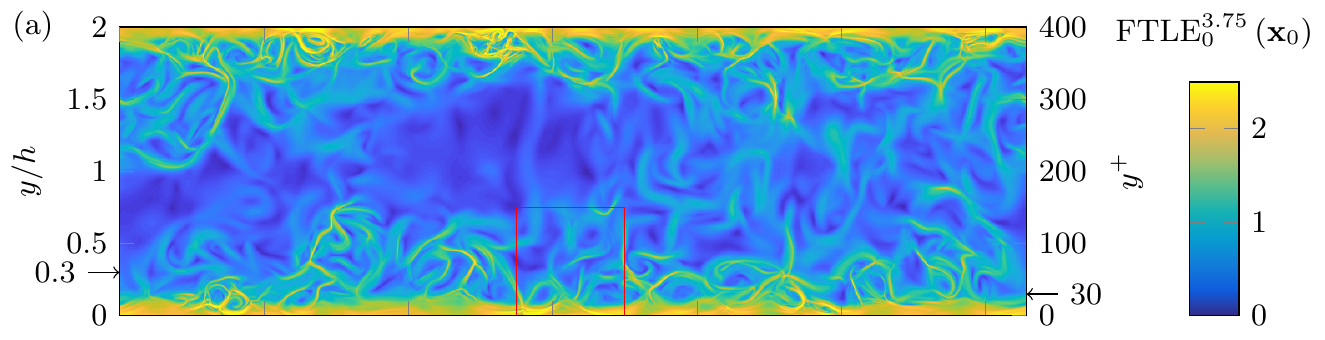}\\
  \includegraphics[]{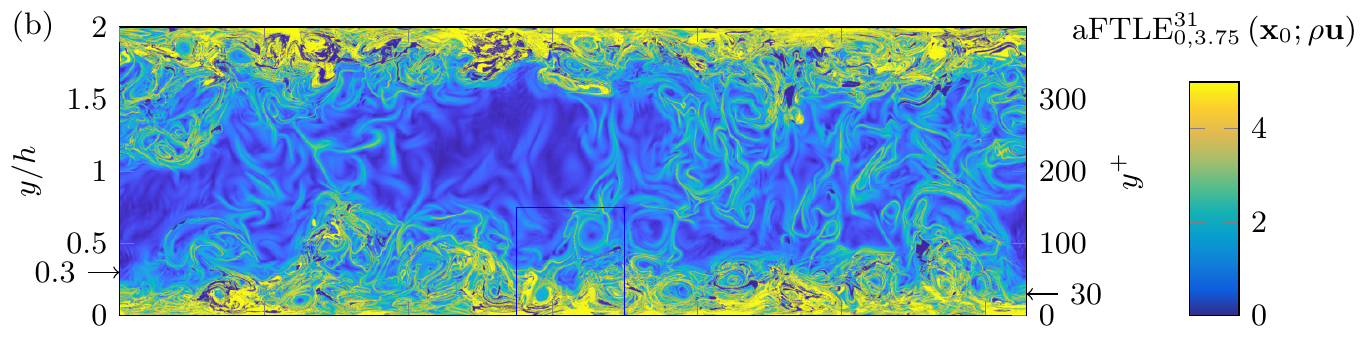}\\
  \includegraphics[]{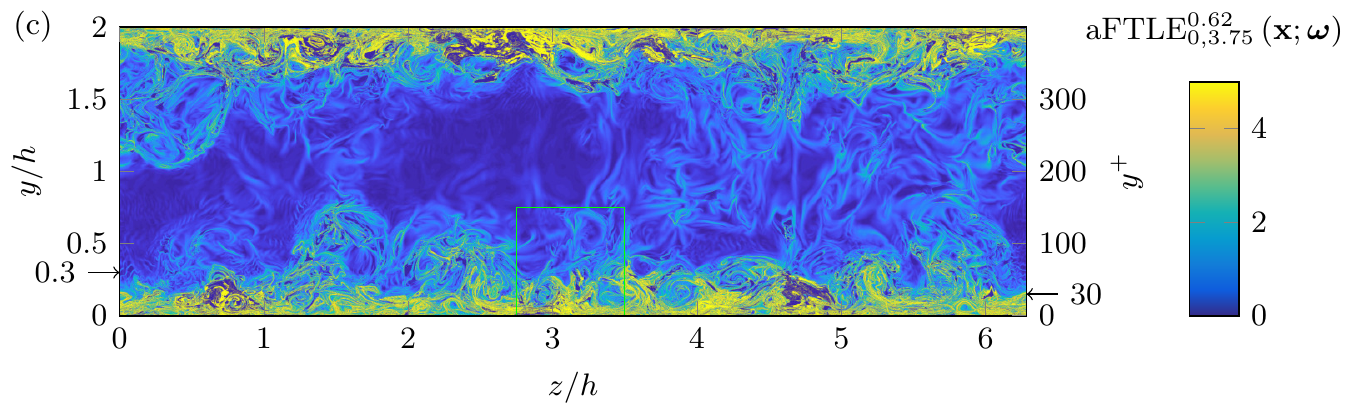}\\
  \hfill\includegraphics[]{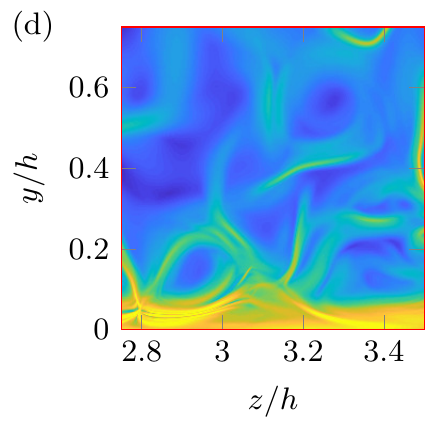}\hfill\includegraphics[]{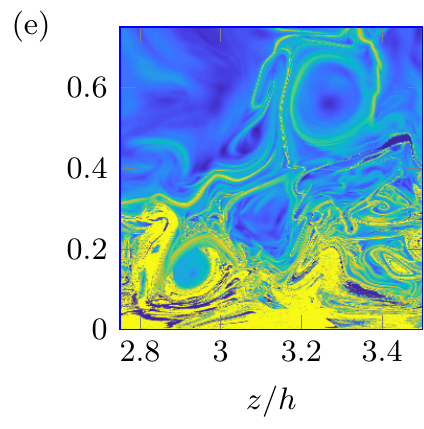}\hfill\includegraphics[]{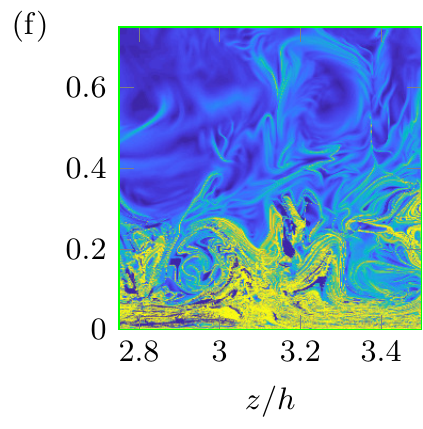}\hfill
  \caption{ Comparison between (a,d) the passive FTLE, 
            (b,e) the aFTLE with respect to $\rho \mathbf{u}$ and (c,f) the 
            aFTLE with respect to $\bm{\omega}$ in a cross-sectional
            plane at $x/h = 2\pi$. The integration interval is for all cases between $t_0=0$ and $t_1=3.75$. 
            The panels (d-f) magnify the region denoted with a rectangle in panels (a-c). }
  \label{fig:channel_lFTLE}
\end{figure}
\begin{figure}
  \includegraphics[]{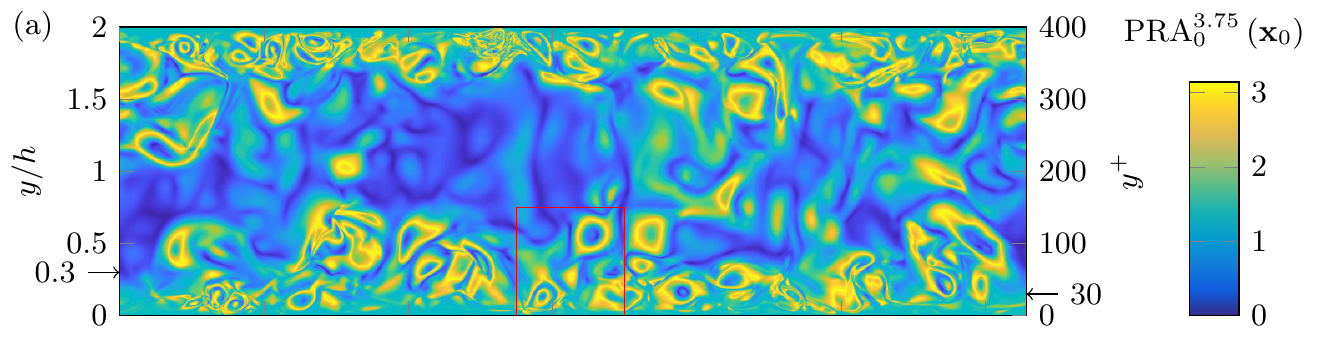}\\
  \includegraphics[]{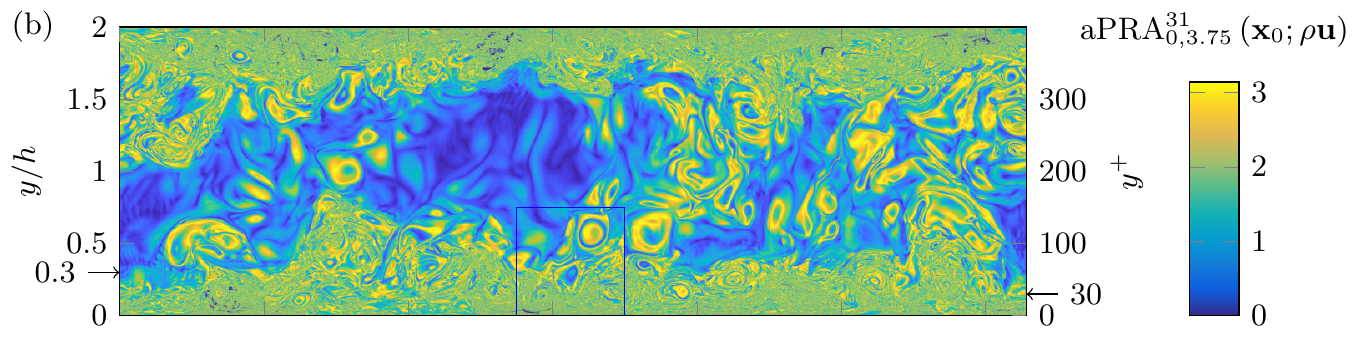}\\
  \includegraphics[]{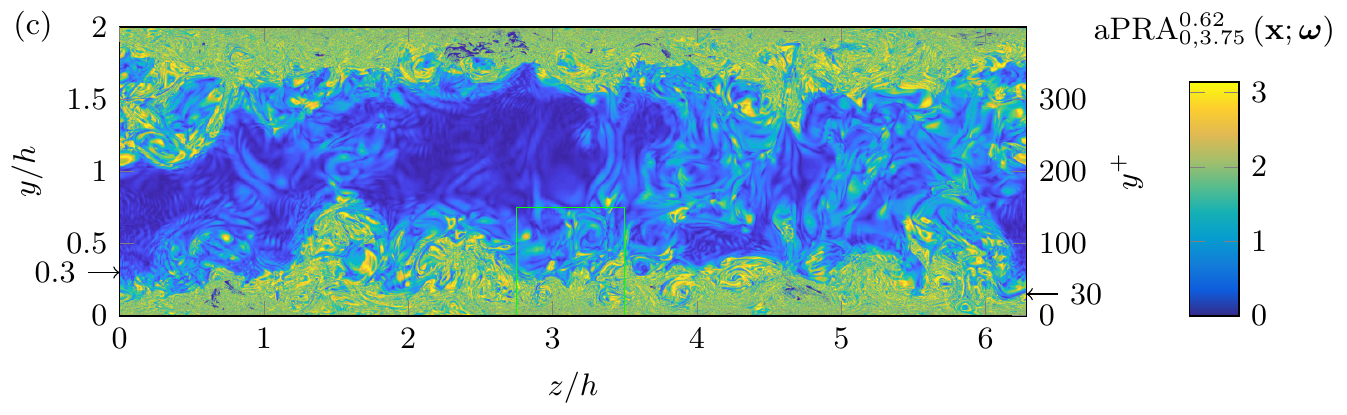}\\
  \hfill\includegraphics[]{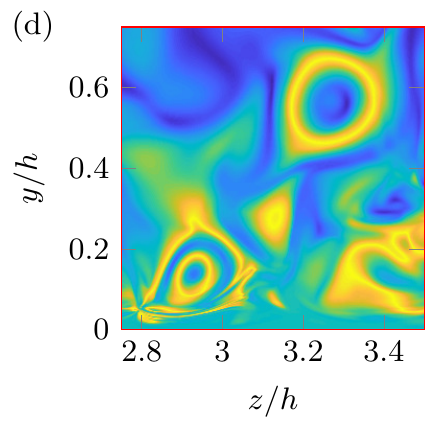}\hfill\includegraphics[]{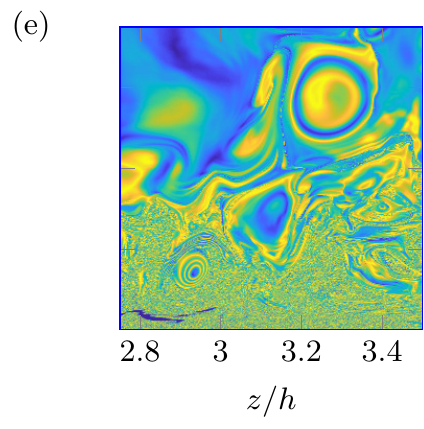}\hfill\includegraphics[]{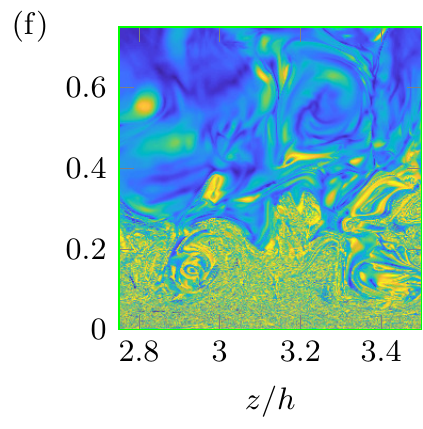}\hfill
  \caption{ Comparison between (a,d) the passive PRA, 
            (b,e) the aPRA with respect to $\rho \mathbf{u}$ and (c,f) the 
            aPRA with respect to $\bm{\omega}$ in a cross-sectional
            plane at $x/h = 2\pi$. The integration interval is for all cases between $t_0=0$ and $t_1=3.75$. 
            The panels (d-f) magnify the region denoted with a rectangle in panels (a-c).}
  \label{fig:channel_lPRA}
\end{figure}
Figures \ref{fig:channel_lFTLE} and \ref{fig:channel_lPRA} show aFTLE and aPRA computed for momentum- and vorticity-based material barriers in a $\left(y-z\right)$ cross-section located at $x=2 \pi h$ and compare them against their passive variants. The integration interval is for all cases between $t_0=0$ and $t_1=3.75$ which corresponds to a time interval of 750 viscous units. The figures clearly show that some features of the Eulerian active barriers discussed in \S\ref{sec:i3d}, such as the spiralling or closed patterns of $\mathrm{aFTLE}_{0,0}^{31}\left(\mathbf{x}, \rho \mathbf{u} \right)$, do have a material character, since they persist almost unchanged over the temporal interval which we have considered.
Examples are shown in the magnification of figures \ref{fig:channel_lFTLE}(e) and \ref{fig:channel_lPRA}(e), showing promise for active LCS diagnostics in studying the lifetime of vortical structures in wall-bounded turbulence (Quadrio \& Luchini 2003).  In the vicinity of the wall, characterised by the strong intermittent turbulent events rapidly evolving with the viscous timescale, less detail is visible in the barriers, due to the lack of material coherence for the considered time frame. Consistent with the general principle discussed in section \ref{relationship between active and passive LCS diagnostics}, passive and active LCS diagnostics tend to highlight the same vortical regions, but tend to differ in the mixing regions surrounding the vortices. As in our 2D turbulence example, while the vorticity-based aFTLE and aPRA plots show a major enhancement over passive FTLE and PRA, some of their details are less clearly defined in comparison to their momentum-based counterparts. Again, this is due to the additional spatial differentiation involved in computing active LCS diagnostics for the vorticity compared to the same computation for the linear momentum.  


\begin{figure}
  \centering
  \includegraphics[]{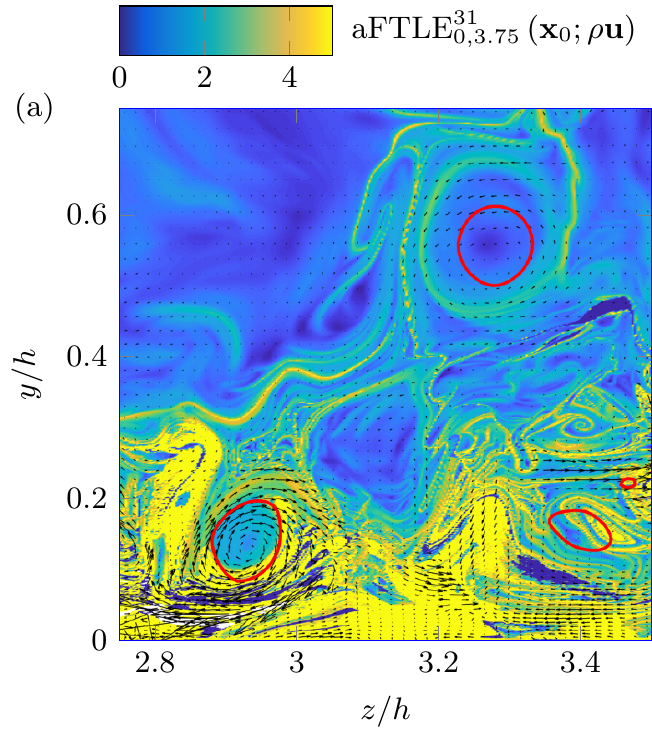}\includegraphics[]{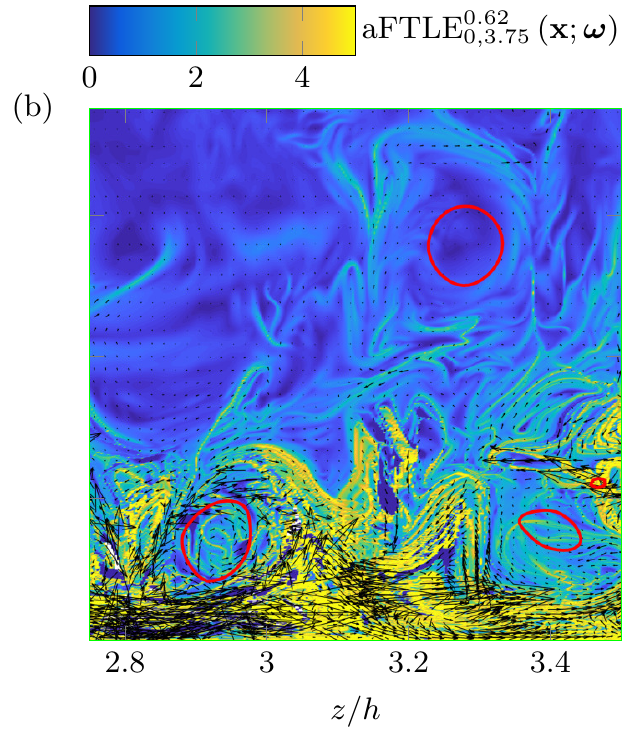}\\
  \caption{Material active barriers of (a) $\rho \mathbf{u}$ and (b) $\bm{\omega}$ at $t=0$ in a cross-sectional
           plane located at $x/h = 2\pi$ and for an integration interval between $t_0=0$ and $t_1=3.75$. 
           The colormap shows the respective aFTLE fields, the vectors show the cross-sectional 
           components of the underlying active barrier field, while the red lines are level-set curves $\lambda_{2}^{+}(\mathbf{x},t)=-0.015$ 
           of the Eulerian vortex identification criterion proposed by Jeong et al. (1997) at the instant $t_0=0$.}
  \label{fig:lbfield}
\end{figure}
Figure \ref{fig:lbfield} shows the material active barrier vector field of (a) $\rho\mathbf{u}$ and (b) $\bm{\omega}$
superimposed to the respective aFTLEs already shown in figure \ref{fig:channel_lFTLE}(e-f).  
Level-set curves of the $\lambda_{2}(\mathbf{x},t)=-0.015$ field at the temporal instant $t=t_0=0$ are also shown. It is confirmed that with the present definition of active barriers, the active vector field is tangent to the detected barriers visualised here as aFTLE ridges, in a temporally-averaged sense. Figure \ref{fig:lbfield}(a) shows that closed $\mathrm{aFTLE}_{t_0,t_1}^{s}\left(\mathbf{x}_0, \rho\mathbf{u} \right)$ ridges can be in some instances close to level-set curves of the $\lambda_{2}$ criterion, as for instance at $(y/h, z/h) \approx (0.15,2.9)$ and $(0.55,3.3)$. In this sense, the material momentum barriers can be utilised as means to objectively identify vortical structures which play a role in inhibiting momentum transport and preserve material coherence over the considered time frame, without resorting to arbitrary choices of level-sets of $\lambda_{2}$. In the present example, we find streamwise vortices that are bounded by active momentum barriers for a time period of 750 viscous units.

\section{Conclusions}

We have developed an approach to identify coherent structure boundaries
as material surfaces that minimize the diffusive transport of active
physical quantities intrinsic to the flow. We have also argued that
instantaneous limits of these active Lagrangian transport barriers
provide objective Eulerian barriers to the short-term redistribution
of active vector fields. 

Our analysis shows that in incompressible Navier\textendash Stokes
flows, active material barriers to transport evolve from structurally
stable 2D stream-surfaces of an associated steady vector field, the
barrier vector field $\mathbf{b}_{t_{0}}^{t_{1}}(\mathbf{x}_{0})$.
This vector field is the time-averaged pull-back of the viscous
terms in the evolution equation of the active vector field. For $t_{0}=t_{1}$,
instantaneous limits of these material barriers to linear momentum
are surfaces to which the viscous forces acting on the fluid are tangent.
Similarly, instantaneous limits to active barriers to vorticity are
surfaces tangent to the curl of viscous forces.

We have obtained that material and Eulerian active barriers in
3D unsteady Beltrami flows coincide exactly with invariant manifolds
of the Lagrangian particle motion. This is noteworthy because all
prior LCS methods applied to Beltrami flows would locate these barriers,
at best, approximately for large enough extraction times, rather than
exactly from arbitrary short extraction times, as the present approach
does. The reason is that the present approach to material barriers
does not rely on quantifying fluid particle separation or lack thereof,
as purely advective LCS-approaches do. Instead, this approach seeks
material surfaces that are most resistant to the diffusive transport
of intrinsic physical quantities, such as momentum and vorticity.
This dynamical extremum problem can be solved without the need for
fluid particles to show large separation.

We have argued and numerically verified that, in comparison to their purely advective versions, active LCS reveal  coherent vortices in much larger detail. Indeed, we have found
the momentum-based aFTLE and the aPRA to outperform the purely advective
FTLE and PRA significantly on vortices of the same finite-time velocity data set.
In contrast, active and passive barriers are expected to differ significantly in mixing regions surrounding those vortices, as we have indeed found in our 2D and 3D turbulence examples.
The refinement of vortical regions from vorticity-based aFTLE and aPRA is also tangible
but more modest, as that computation involves one more spatial derivative
and hence is more prone to numerical error. In addition, \textbf{$\mathrm{aFTLE}_{t_{0},t_{1}}^{s}$}
and $\mathrm{aPRA}_{t_{0},t_{1}}^{s}$ converge as the barrier-time
$s$ increases, whereas $\mathrm{FTLE}_{t_{0}}^{t}$ and $\mathrm{PRA}_{t_{0}}^{t}$
generally do not converge in unsteady flows as the physical time $t$
increases. The convergence of \textbf{$\mathrm{aFTLE}_{t_{0},t_{1}}^{s}$}
and $\mathrm{aPRA}_{t_{0},t_{1}}^{s}$ enables a scale-dependent exploration
of active barriers, with smaller spatial scales gradually revealed
under increasing barrier times $s$. 

A further advantage of the dynamically
active approach to transport-barrier analysis is that an active Poincar\'e
map (i.e., Poincar\'e map applied to the barrier equations $\mathbf{x}_{0}^{\prime}=\mathbf{b}_{t_{0}}^{t_{1}}(\mathbf{x}_{0})$)
is a well-defined, time-independent map that can be iterated for visualization
if barrier trajectories return to the Poincar\'e section. In contrast,
no time-independent return map can be defined and iterated for the
unsteady fluid-particle equation of motion $\dot{\mathbf{x}}=\mathbf{u}(\mathbf{x},t)$,
because each subsequent return to a Poincar\'e section is governed by
a different map.

The 2D versions of our results provide the simplest available objective
LCS criteria, identifying barriers to active transport as level curves
of appropriate Hamiltonians that are functions of the scalar vorticity.
This follows from the fact that the 2D barrier equations turn out to be autonomous, planar Hamiltonian systems, and hence are, in principle, integrable. We have found, however, that active LCS diagnostics applied to these autonomous but highly complex planar Hamiltonian systems
give a more robust and detailed localization of coherent vortex boundaries
than level-curve identification of their numerically generated Hamiltonians. 

Eulerian active barriers (identified from the steady dynamical
system $\mathbf{x}^{\prime}=\mathbf{b}_{t}^{t}(\mathbf{x})$) provide
an objective and parameter-free alternative to currently used, observer-dependent
flow-visualization tools, such as level surfaces of the velocity norm,
of the velocity components and of the $Q$-, $\Delta$- and $\lambda_{2}$-fields.
Undoubtedly, the implementation of the latter tools is appealingly
simple via automated level-surface visualization packages. Yet such
evolving surfaces are observer-dependent and non-material, thereby
lacking any experimental verifiability. In addition, beyond the simplicity
of generating coherent structure boundaries as level sets of these
scalar fields, the physical meaning of such level sets remains unclear.

The objectivity of the barrier vector field $\mathbf{b}_{t_{0}}^{t_{1}}$ implies that any Galilean-invariant vortex criterion mentioned in the Introduction becomes automatically objective when applied to $\mathbf{b}_{t_{0}}^{t_{1}}$, as opposed to the velocity field $\mathbf{u}$. This fact does not eliminate the heuristic nature of these criteria but at least makes the structures they return independent of the observer. The physical rationale for applying vortex- or LCS-criteria to the barrier vector field instead of the velocity field is that active barriers have a well-defined and readily quantifiable role in the viscous force field due to their transport-minimizing property, even over infinitesimally short times.  In contrast, coherence structures in the velocity field can be approached from a multitude of different principles, most of which are qualitative (i.e., lack a well-defined optimization argument) and require substantial fluid particle separation to be effective.

A physical take-away message from our 3D channel flow example is that Eulerian active barriers for momentum (or vorticity) visualize the instantaneous landscape of the viscous forces, which are everywhere tangent to those barriers and hence induce zero instantaneous diffusive transport of momentum (or vorticity) across them. Several Lagrangian active barriers are small perturbations of their Eulerian counterparts, suggesting that those Eulerian barriers have a strong material character over a significant period of time. As a second notable finding, several (but not all) momentum barriers are well approximated by quasi-streamwise tubular $\lambda_2$ level surfaces (often called streamwise vortices),  which are considered crucial elements in the regeneration cycle of near-wall turbulence (Hamilton, Kim \& Waleffe 1995; Jimenez \& Pinelli 1999). Active momentum barriers, therefore, offer a threshold-independent identification of the intrinsic, observer-independent subset of  $\lambda_2$-vortices. Such objective streamwise vortices  are bounded by material surfaces across which viscous momentum transport is minimal, while vorticity diffuses across them. A third physical finding from our analysis is that the low-Reynolds-number turbulent channel flow considered here contains active coherent structure boundaries that penetrate and span the bulk flow. Notably, active barriers spanning across the entire channel height are present in some regions of the channel cross section but absent in others. This indicates possible large-scale coherent features in this specific flow that deserve further investigation.

Finally, the objective momentum-barrier theory described here should
be able to contribute to the understanding and identification of various
turbulent flow structures that have only been described so far in
an observer- and threshold-dependent fashion under a number of assumptions
and approximations. Specifically, our future work will seek to uncover
experimentally identifiable material signatures of uniform momentum
zones (Adrian, Meinhart \& Tomkins 2000, De Silva Hutchins \& Marusic
2016) and turbulent superstructures (Marusic, Mathis \& Hutchins 2018
and Pandey, Scheel \& Schumacher 2018) based on the notion of diffusive
momentum barriers developed in this paper.

\vskip 1 true cm

\textbf{Acknowledgment}

The authors acknowledge financial support from Priority Program SPP
1881 (Turbulent Superstructures) of the German National Science Foundation
(DFG). We are grateful to Prof. Mohammad Farazmand for providing us
with the 2D turbulence data set he originally generated for the analysis
in Katsanoulis et al. (2019). We are also grateful to Prof. Charles
Meneveau for his helpful comments and for pointing out the reference
Meyers \& Meneveau (2013) to us. Finally, G.H. is thankful to Prof. Andrew Majda for his inspirational remarks, made about 25 years ago,
on the importance of dynamically active transport relative to purely
advective transport.
\vskip 0.5 true cm
The Authors report no conflict of interest.

\appendix

\section{A motivating example}

A simple example underlying the challenges of defining barriers to
momentum and vorticity transport is a planar, unsteady Navier\textendash Stokes
vector field representing an unsteady, decaying channel-flow between
two walls at $x_{2}=\pm\frac{1}{4}$ (see. Fig. \ref{fig: steady 2D channel flow example-0}).
\begin{figure}
	\centering{}\includegraphics[width=0.4\textwidth]{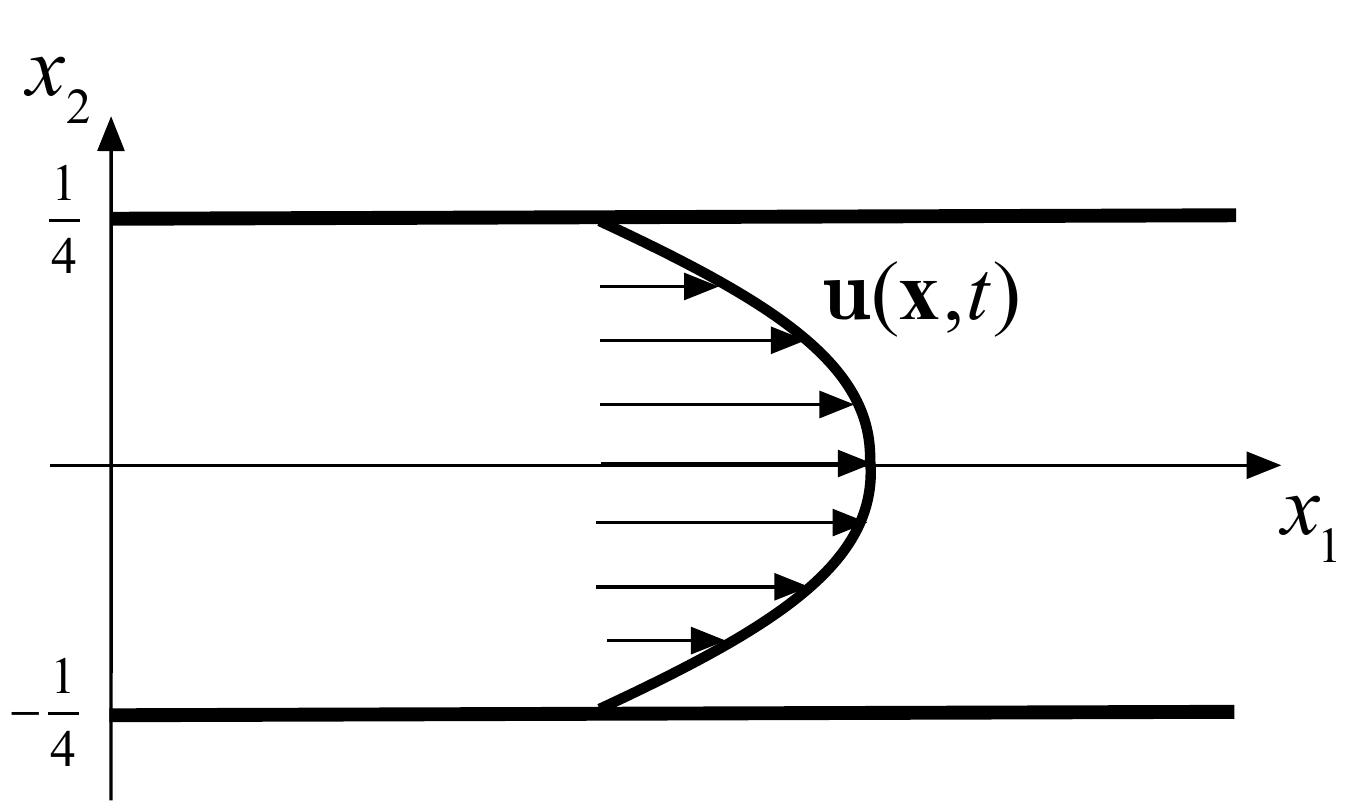}\caption{Decaying planar Navier\textendash Stokes flow in a channel with no-slip
		walls at $x_{2}=\pm\frac{1}{4}$.}
	\label{fig: steady 2D channel flow example}
\end{figure}
The corresponding velocity and scalar vorticity fields are
\begin{equation}
\mathbf{u}(\mathbf{x},t)=e^{-4\pi^{2}\nu t}\left(a\cos2\pi x_{2},0\right),\qquad\omega(\mathbf{x},t)=2\pi ae^{-4\pi^{2}\nu t}\sin2\pi x_{2}.\label{eq:horizontal shear jet}
\end{equation}
Normalized by their instantaneous global maxima, the normalized linear
momentum $\rho\mathbf{u}^{0}=\left(\cos2\pi x_{2},0\right)$ and vorticity
$\omega^{0}=\sin2\pi x_{2}$ are both constant in time. There is,
therefore, no structural reorganization in the topology of the momentum
and vorticity fields. Instead, for all times, horizontal lines act
as level curves for both the horizontal momentum and the vorticity,
forming material barriers between higher and lower values of these
scalars. Indeed, the theory developed in this paper identifies all
horizontal lines as materiel barriers to the diffusive transport of
both momentum and vorticity (see Example 1 of section \ref{subsec:active barriers in 2D-Navier=002013Stokes-flows}).

Haller et al. (2019) obtain an ODE family describing the time $t_{0}$
position of uniform barriers to the diffusive (passive) transport
of the scalar vorticity over a finite time interval $[t_{0},t_{1}].$
With the notation $y_{0}=2\pi x_{2}$, with the constants
\begin{equation}
A=\frac{a^{2}}{\nu}\sin2y_{0}\left[\frac{1}{2}e^{-2\nu t_{1}}+\frac{1}{2}e^{-2\nu t_{0}}-e^{-\nu\left(t_{1}+t_{0}\right)}\right],\qquad B=a\left(e^{-\nu t_{0}}-e^{-\nu t_{1}}\right),
\end{equation}
and with the vector field 
\begin{equation}
\bar{\mathbf{q}}_{t_{0}}^{t_{1}}(\mathbf{x}_{0})=\frac{1}{2\nu\left(t_{1}-t_{0}\right)}\left(\begin{array}{c}
A\sin2y_{0}\\
B\cos y_{0}
\end{array}\right),
\end{equation}
the ODE family describing the time $t_{0}$ position of uniform constrained
barriers is given by
\begin{equation}
\mathbf{x}_{0}^{\prime}=\frac{1}{2\nu\left(t_{1}-t_{0}\right)}\left\{ \frac{\sqrt{\left|\mathbf{\bar{q}}_{t_{0}}^{t_{1}}\left(\mathbf{x}_{0}\right)\right|^{2}-\mathcal{T}_{0}^{2}}}{\left|\mathbf{\bar{q}}_{t_{0}}^{t_{1}}\left(\mathbf{x}_{0}\right)\right|^{2}}\left(\begin{array}{c}
A\sin2y_{0}\\
B\cos y_{0}
\end{array}\right)+\frac{\mathcal{T}_{0}}{\left|\mathbf{\bar{q}}_{t_{0}}^{t_{1}}\left(\mathbf{x}_{0}\right)\right|^{2}}\left(\begin{array}{c}
B\cos y_{0}\\
-A\sin2y_{0}
\end{array}\right)\right\} \label{eq:2ODE}
\end{equation}
for some value of the transport density constant $\mathcal{T}_{0}\in\mathbb{R}$.
For the choice
\begin{equation}
\mathcal{T}_{0}=\left|\mathbf{\bar{q}}_{t_{0}}^{t_{1}}\left(\mathbf{x}_{0}\right)\right|_{y_{0}=0}=\frac{B}{2\nu\left(t_{1}-t_{0}\right)},\label{eq:T_0 for jet}
\end{equation}
the ODE \eqref{eq:2ODE} becomes
\begin{align}
\mathbf{x}_{0}^{\prime} & \vert_{y_{0}=0}=\frac{B}{2\nu\left(t_{1}-t_{0}\right)}\left(\begin{array}{c}
B\\
0
\end{array}\right)\parallel\mathbf{\bm{\Omega}}\mathbf{\bar{q}}_{t_{0}}^{t_{1}}\left(\mathbf{x}_{0}\right)\vert_{y_{0}=0},\label{eq:jet parallel}
\end{align}
showing that $x_{02}=0$ is an invariant line for equation \eqref{eq:2ODE}
for the parameter value $\mathcal{T}_{0}$ selected as in \eqref{eq:T_0 for jet}.
Consequently, the center line of the channel at $x_{02}=0$ is a uniform,
constrained barrier to vorticity-diffusion along which the pointwise
diffusive transport of vorticity is equal to \eqref{eq:2ODE}. Choosing
the constant $\mathcal{T}_{0}=0$ in eq. \eqref{eq:2ODE} gives
\begin{equation}
\mathbf{x}_{0}^{\prime}=\frac{1}{2\nu a\left(t_{1}-t_{0}\right)\left|\mathbf{\bar{q}}_{t_{0}}^{t_{1}}\left(\mathbf{x}_{0}\right)\right|}\left(\begin{array}{c}
A\sin4\pi x_{02}\\
B\cos2\pi x_{02}
\end{array}\right),
\end{equation}
for which $x_{02}=\pm1/4$ are invariant lines, and hence the channel
walls at $x_{02}=\pm1/4$ are perfect constrained barriers to diffusive
transport. Therefore, the variational theory of Haller et al. (2019)
identifies the center line of the channel at $x_{2}=0$ and the upper
and lower walls as barriers to vorticity transport, but finds an infinite
family of non-straight barrier curves for the rest of the channel,
given by general integral curves of the vector field family \eqref{eq:2ODE}
(see Fig.  \ref{fig: steady 2D channel flow example-0}). Only in
the limit of $t_{1}\to\infty$ do the latter, curved variational barriers
align with horizontal lines, which is suboptimal, given that these
horizontal barriers prevail already in any finite-time observation
of the vorticity field. The objective of the present paper is to strengthen
these results by considering vorticity transport as an active, vectorial
transport problem consistent with the 3D Navier\textendash Stokes
equation, rather than a passive scalar transport problem in the 2D
Navier\textendash Stokes equation.

In contrast, Meyers \& Meneveau (2013) define a momentum flux vector
field $\bar{\mathbf{F}}_{m}^{\boldsymbol{\zeta}}\left(\mathbf{x},t\right)$
with respect to a unit reference direction vector $\boldsymbol{\zeta}\in\mathbb{R}^{3}$
as
\begin{equation}
\bar{\mathbf{F}}_{m}^{\boldsymbol{\zeta}}=\left(\bar{\mathbf{u}}\cdot\boldsymbol{\zeta}\right)\bar{\mathbf{u}}+\overline{\mathbf{u}^{\prime}\otimes\mathbf{u}^{\prime}}\boldsymbol{\zeta}-2\nu\bar{\mathbf{S}}\boldsymbol{\zeta},\label{eq:M-M momentum flux vector}
\end{equation}
where overbar refers to Reynolds-averaging, prime refers to the fluctuating
part of the velocity field, $\otimes$ denotes the dyadic product
and $\mathbf{S}=\frac{1}{2}\left[\boldsymbol{\nabla}\mathbf{u}+\mathbf{\left(\boldsymbol{\nabla}\mathbf{u}\right)}^{T}\right]$
is the rate-of-strain tensor. The flux vector $\bar{\mathbf{F}}_{m}^{\boldsymbol{\zeta}}$
is obtained by Meyers \& Meneveau (2013) after averaging the unsteady
terms out of the Navier\textendash Stokes equations, projecting these
averaged equations into the $\boldsymbol{\zeta}$ direction, identifying
all terms that are divergences of some vector field in these projected
equations, and summing up all three vector fields identified in this
fashion. For the laminar velocity field \eqref{eq:horizontal shear jet},
we have $\bar{\mathbf{u}}\equiv\mathbf{u}$, $\bar{\mathbf{S}}\equiv\mathbf{S}$,
$\mathbf{u}^{\prime}\equiv\mathbf{0}$, and $\bar{\mathbf{F}}_{m}^{\boldsymbol{\zeta}}\equiv\mathbf{F}_{m}^{\boldsymbol{\zeta}}$.
Following the choice of Meyers \& Meneveau (2013) for planar parallel
shear flows, we select $\boldsymbol{\zeta}=\left(1,0\right)^{T}$.
Using the relation 
\begin{equation}
\mathbf{S}=ae^{-4\pi^{2}\nu t}\left(\begin{array}{cc}
0 & -\pi\sin2\pi x_{2}\\
-\pi\sin2\pi x_{2} & 0
\end{array}\right),
\end{equation}
we obtain from eq. \eqref{eq:M-M momentum flux vector} the momentum-flux
vector
\begin{align}
\mathbf{F}_{m}^{\boldsymbol{\zeta}}&=a^{2}e^{-8\pi^{2}\nu t}\left(\begin{array}{c}
\cos^{2}2\pi x_{2}\\
0
\end{array}\right)-2\nu ae^{-4\pi^{2}\nu t}\left(\begin{array}{c}
0\\
-\pi\sin2\pi x_{2}
\end{array}\right)\nonumber\\
&=ae^{-4\pi^{2}\nu t}\left(\begin{array}{c}
ae^{-4\pi^{2}\nu t}\cos^{2}2\pi x_{2}\\
2\nu\pi\sin2\pi x_{2}
\end{array}\right).
\end{align}

The $x_{2}=0$ line is an integral curve of $\mathbf{F}_{m}^{\boldsymbol{\zeta}}$,
correctly conveying the fundamental role of the centerline of the
channel in blocking linear momentum transfer. All other integral curves
of $\mathbf{F}_{m}^{\boldsymbol{\zeta}}\left(\mathbf{x},t\right)$,
however, curl either upwards or downwards, running eventually into
the two horizontal walls perpendicularly. These curves turn very slowly
towards to channel walls for small values of the viscosity. For easy
illustration over a shorter horizontal domain, we select the time
$t^{*}=-\frac{1}{4\pi^{2}\nu}\log\left[2\nu\pi/a\right]$ so that
$\mathbf{F}_{m}^{\boldsymbol{\zeta}}$ becomes
\begin{equation}
\mathbf{F}_{m}^{\boldsymbol{\zeta}}\left(\mathbf{x},t^{*}\right)=2\nu\pi ae^{-4\pi^{2}\nu t^{*}}\left(\begin{array}{c}
\cos^{2}2\pi x_{2}\\
\sin2\pi x_{2}
\end{array}\right),
\end{equation}
whose integral curves are shown in Fig.  \ref{fig: steady 2D channel flow example-0}).
These integral curves do not delineate observable structures governing
the rearrangement of momentum within this flow. In the limit of $t\to\infty,$
they limit on vertical lines.

\section{Reynolds transport theorem and the convective flux through the boundary
	of a material volume }

The Reynolds transport theorem for an arbitrary vector field $\mathbf{f}(\mathbf{x},t)$
and an arbitrary, time-varying volume $V(t)$ in a velocity field
$\mathbf{u}(\mathbf{x},t)$ is of the form
\begin{equation}
\frac{d}{dt}\int_{V(t)}\mathbf{f}\,dV=\int_{V(t)}\frac{\partial\mathbf{f}}{\partial t}\,dV+\int_{\partial V(t)}\mathbf{f}\left(\mathbf{u}_{\partial V(t)}\cdot\mathbf{n}\right)\,dA.\label{eq:Reynolds}
\end{equation}
Here $\mathbf{u}_{\partial V(t)}$ denotes the local velocity of the
boundary surface $\partial V(t)$ of $V(t)$, therefore we have $\mathbf{u}_{\partial V(t)}=\mathbf{u}$
when $V(t)$ is a material volume. The identity \eqref{eq:Reynolds}
merely gives a formal partition of $\frac{d}{dt}\int_{V(t)}\mathbf{f}\,dV$
into two terms, yet it is tempting to conclude that the second term,
$\int_{\partial V(t)}\mathbf{f}\left(\mathbf{u}\cdot\mathbf{n}\right)\,dA$,
is the convective flux of $\mathbf{f}$ through the boundary $\partial V(t)$
of $V(t)$. We will now illustrate on a specific example that this
is generally not the case.

Consider the scalar version of \ref{eq:Reynolds} for a passive scalar
field $c\left(\mathbf{x},t\right)$:
\begin{equation}
\frac{d}{dt}\int_{V(t)}c\,dV=\int_{V(t)}\frac{\partial c}{\partial t}\,dV+\int_{\partial V(t)}c\left(\mathbf{u}\cdot\mathbf{n}\right)\,dA.\label{eq:Reynolds6}
\end{equation}
Assume that \textbf{$\mathbf{u}$} is incompressible and $c$ is
a passive scalar field that is a solution of the advection-diffusion
equation
\begin{equation}
\frac{Dc}{Dt}=\partial_{t}c+\boldsymbol{\nabla}c\cdot\mathbf{u}=\kappa\Delta c,\label{eq:adv diff for Reynolds section}
\end{equation}
with diffusivity $\kappa>0$. The surface integral in \eqref{eq:Reynolds6}
gives a formal convective flux for the passive scalar field $\mathbf{c}$
across $\partial V(t)$ even though no convective scalar transport
can occur through the material surface $\partial V(t)$. 

The (purely diffusive) flux of $c$ out of $V(t)$ can be computed
directly as
\begin{equation}
\frac{d}{dt}\int_{V(t)}c\,dV=\int_{V(t_{0})}\frac{Dc}{Dt}\,dV_{0}=\int_{V(t_{0})}\kappa\Delta c\,dV_{0}=\kappa\int_{V(t)}\boldsymbol{\nabla}\cdot\left(\boldsymbol{\nabla}c\right)\,dV=\int_{\partial V(t)}\kappa\mathbf{\bm{\nabla}}c\cdot\mathbf{n}\,dA,
\end{equation}
showing that the vector describing the correct pointwise diffusive
flux vector of the passive scalar $c(\mathbf{x},t)$ through the material
surface $\partial V(t)$ is the well-known flux vector, $\kappa\mathbf{\bm{\nabla}}c$
rather than the vector $c\mathbf{u}$ appearing in the surface integral
term in \eqref{eq:Reynolds6}. This is because the volume integral
term $\int_{V(t)}\frac{\partial c}{\partial t}\,dV$ on the right-hand
side of the transport theorem \eqref{eq:Reynolds6} also contributes
to the flux through $\partial V(t)$. Indeed, using eq. \eqref{eq:adv diff for Reynolds section},
we can rewrite this term as
\begin{align}
\int_{V(t)}\frac{\partial c}{\partial t}\,dV & =\int_{V(t)}\left(\kappa\Delta c-\boldsymbol{\nabla}c\cdot\mathbf{u}\right)\,dV=\int_{V(t)}\mathbf{\bm{\nabla}}\cdot\left(\kappa\boldsymbol{\nabla}c-c\mathbf{u}\right)\,dV\nonumber \\
& =\int_{\partial V(t)}\left(\kappa\boldsymbol{\nabla}c-c\mathbf{u}\right)\cdot\mathbf{n}dA.
\end{align}
Therefore, $\int_{V(t)}\frac{\partial c}{\partial t}\,dV$ yields
a nonzero flux through the boundary and a part of this flux cancels
out the second integral in \eqref{eq:Reynolds6} that incorrectly
suggests nonzero convective flux for $c$. 

More generally, the partition of $\frac{d}{dt}\int_{V(t)}\mathbf{f}\,dV$
in \eqref{eq:Reynolds} into two terms is somewhat arbitrary from
the point of view of transport through the boundary of a material
volume. Indeed, the volume integral on the right-hand-side of \eqref{eq:Reynolds}
will also contribute to the flux of \textbf{$\mathbf{f}$ }through
the boundary of $V(t)$.

\section{Poofs of Theorems \ref{thm:2D momentum barriers} and \ref{thm:2D vorticity barriers}}

\subsection{Proof of Theorem \ref{thm:2D momentum barriers}}

For a Navier-Stokes velocity field $\mathbf{u}$ of the form \eqref{eq:2D ansatz}-\eqref{eq:w is the vorticity},
we have 
\begin{equation}
\Delta\mathbf{u}\left(\mathbf{x},t\right)=\left(\begin{array}{c}
\Delta_{\hat{\mathbf{x}}}\hat{\mathbf{u}}\\
\Delta_{\hat{\mathbf{x}}}\hat{\omega}
\end{array}\right).
\end{equation}
Therefore, 
\begin{align}
\left(\mathbf{F}_{t_{0}}^{t}\right)^{*}\Delta\mathbf{u}(\mathbf{x}_{0}) & =\left[\nabla_{\mathbf{x}_{0}}\mathbf{F}_{t_{0}}^{t}\left(\mathbf{x}_{0}\right)\right]^{-1}\left(\begin{array}{c}
\Delta_{\hat{\mathbf{x}}}\hat{\mathbf{u}}\\
\Delta_{\hat{\mathbf{x}}}\hat{\omega}(\hat{\mathbf{x}},t)
\end{array}\right)\nonumber \\
& =\left(\begin{array}{cc}
\boldsymbol{\nabla}_{\mathbf{\hat{\mathbf{x}}}}\hat{\mathbf{F}}_{t}^{t_{0}}\left(\mathbf{\hat{\mathbf{x}}}\right) & \mathbf{0}\\
\int_{t}^{t_{0}}\boldsymbol{\nabla}_{\mathbf{\hat{\mathbf{x}}}}\hat{\omega}\left(\hat{\mathbf{F}}_{t}^{s}\left(\hat{\mathbf{x}}\right),s\right)ds & 1
\end{array}\right)\left(\begin{array}{c}
\Delta_{\hat{\mathbf{x}}}\hat{\mathbf{u}}\\
\Delta_{\hat{\mathbf{x}}}\hat{\omega}(\hat{\mathbf{x}},t)
\end{array}\right)\nonumber \\
& =\left(\begin{array}{c}
\boldsymbol{\nabla}_{\mathbf{\hat{\mathbf{x}}}}\hat{\mathbf{F}}_{t}^{t_{0}}\left(\mathbf{\hat{\mathbf{x}}}\right)\Delta_{\hat{\mathbf{x}}}\hat{\mathbf{u}}\\
\int_{t}^{t_{0}}\boldsymbol{\nabla}_{\mathbf{\hat{\mathbf{x}}}}\hat{\omega}\left(\hat{\mathbf{F}}_{t}^{s}\left(\hat{\mathbf{x}}\right),s\right)ds\cdot\Delta_{\hat{\mathbf{x}}}\hat{\mathbf{u}}+\Delta_{\hat{\mathbf{x}}}\hat{\omega}(\hat{\mathbf{x}},t)
\end{array}\right).\label{eq:pullback for 2D flows}
\end{align}
With these expressions, the barrier equation \eqref{eq:barrier equation}
becomes

\begin{align}
\hat{\mathbf{x}}_{0}^{\prime} & =\nu\rho\overline{\left(\hat{\mathbf{F}}_{t_{0}}^{t}\right)^{*}\Delta_{\hat{\mathbf{x}}}\hat{\mathbf{u}}(\hat{\mathbf{x}}_{0})},\nonumber \\
x_{03}^{\prime} & =\nu\rho A(\mathbf{\hat{\mathbf{x}}}_{0},t_{1},t_{0}),\label{eq:momentum barrier eq for 2D flow}
\end{align}
for an appropriate smooth function $A(\mathbf{\hat{\mathbf{x}}}_{0},t_{1},t_{0}).$
Two-dimensional invariant manifolds of this dynamical system are of
the form $\left\{ \hat{\mathbf{x}}_{0}(s)\right\} _{s\in\mathbb{R}}\times\mathbb{R}$,
i.e., topological products of trajectories of the $\hat{\mathbf{x}}_{0}$-component
of the \eqref{eq:momentum barrier eq for 2D flow-separable case}
with a line in the $x_{03}$ direction. As trajectories $\left\{ \hat{\mathbf{x}}_{0}(s)\right\} _{s\in\mathbb{R}}$
are contained in the streamlines of the steady 2D velocity field $\overline{\left(\hat{\mathbf{F}}_{t_{0}}^{t}\right)^{*}\Delta_{\hat{\mathbf{x}}}\hat{\mathbf{u}}(\hat{\mathbf{x}}_{0})}$,
Eulerian barriers to momentum transport are, structurally stable
streamlines of the vector field $\mathbf{\Delta_{\hat{\mathbf{x}}}\hat{\mathbf{u}}}(\hat{\mathbf{x}},t)$.
By incompressibility, we have
\begin{equation}
\Delta_{\hat{\mathbf{x}}}\hat{\mathbf{u}}=\left(\begin{array}{c}
\partial_{x_{1}x_{1}}^{2}v_{1}+\partial_{x_{2}x_{2}}^{2}v_{1}\\
\partial_{x_{1}x_{1}}^{2}v_{2}+\partial_{x_{2}x_{2}}^{2}v_{2}
\end{array}\right)=\left(\begin{array}{c}
-\partial_{x_{1}x_{2}}^{2}v_{2}+\partial_{x_{2}x_{2}}^{2}v_{1}\\
\partial_{x_{1}x_{1}}^{2}v_{2}-\partial_{x_{1}x_{2}}^{2}v_{1}
\end{array}\right)=\left(\begin{array}{c}
\partial_{x_{2}}\hat{\omega}\\
-\partial_{x_{1}}\hat{\omega}
\end{array}\right),\label{eq:2D velocity Laplacian}
\end{equation}
and hence these streamlines are structurally stable level curves
of the stream function $\hat{\omega}(\boldsymbol{\hat{\mathbf{x}}},t)$,
as claimed.

Using formula \eqref{eq:2D velocity Laplacian} and the canonical
symplectic matrix $\mathbf{J=}\left(\begin{array}{cc}
0 & 1\\
-1 & 0
\end{array}\right)$, we also find that 
\begin{equation}
\Delta_{\hat{\mathbf{x}}}\hat{\mathbf{u}}\left(\hat{\mathbf{F}}_{t_{0}}^{t}\left(\mathbf{x}_{0}\right),t\right)=\mathbf{J}\boldsymbol{\nabla}\omega\left(\hat{\mathbf{F}}_{t_{0}}^{t}\left(\mathbf{x}_{0}\right),t\right)=\mathbf{J}\left[\boldsymbol{\nabla}_{0}\hat{\mathbf{F}}_{t_{0}}^{t}\left(\mathbf{x}_{0}\right)\right]^{-T}\boldsymbol{\nabla}_{0}\hat{\omega}\left(\hat{\mathbf{F}}_{t_{0}}^{t}\left(\mathbf{x}_{0}\right),t\right),
\end{equation}
where $\boldsymbol{\nabla}_{0}\hat{\omega}\left(\hat{\mathbf{F}}_{t_{0}}^{t}\left(\mathbf{x}_{0}\right),t\right)$
denotes the derivative of the Lagrangian vorticity $\omega\left(\hat{\mathbf{F}}_{t_{0}}^{t}\left(\mathbf{x}_{0}\right),t\right)$
with respect to the initial condition $\mathbf{x}_{0}$. This last
equation implies
\begin{align}
\left(\hat{\mathbf{F}}_{t_{0}}^{t}\right)^{*}\Delta_{\hat{\mathbf{x}}}\hat{\mathbf{u}}(\mathbf{x}_{0}) & =\left[\boldsymbol{\nabla}_{0}\hat{\mathbf{F}}_{t_{0}}^{t}\left(\mathbf{x}_{0}\right)\right]^{-1}\Delta_{\hat{\mathbf{x}}}\hat{\mathbf{u}}\left(\hat{\mathbf{F}}_{t_{0}}^{t}\left(\mathbf{x}_{0}\right),t\right)\nonumber \\
& =\left[\boldsymbol{\nabla}_{0}\hat{\mathbf{F}}_{t_{0}}^{t}\left(\mathbf{x}_{0}\right)\right]^{-1}\mathbf{J}\left[\boldsymbol{\nabla}_{0}\hat{\mathbf{F}}_{t_{0}}^{t}\left(\mathbf{x}_{0}\right)\right]^{-T}\boldsymbol{\nabla}_{0}\hat{\omega}\left(\hat{\mathbf{F}}_{t_{0}}^{t}\left(\mathbf{x}_{0}\right),t\right)\nonumber \\
& =\det\left[\boldsymbol{\nabla}_{0}\hat{\mathbf{F}}_{t_{0}}^{t}\left(\mathbf{x}_{0}\right)\right]^{-1}\mathbf{J}\boldsymbol{\nabla}_{0}\hat{\omega}\left(\hat{\mathbf{F}}_{t_{0}}^{t}\left(\mathbf{x}_{0}\right),t\right)=\mathbf{J}\boldsymbol{\nabla}_{0}\hat{\omega}\left(\hat{\mathbf{F}}_{t_{0}}^{t}\left(\mathbf{x}_{0}\right),t\right),
\end{align}
given that $\det\left[\boldsymbol{\nabla}_{0}\hat{\mathbf{F}}_{t_{0}}^{t}\left(\mathbf{x}_{0}\right)\right]^{-1}\equiv1$
holds due to incompressibility. Here, we have also used the fact here for any constants $a,b,c,d\in\mathbb{R}$
satisfying $ad-bc=1,$ we have
\begin{equation}
\left(\begin{array}{cc}
a & b\\
c & d
\end{array}\right)\left(\begin{array}{cc}
0 & 1\\
-1 & 0
\end{array}\right)\left(\begin{array}{cc}
a & c\\
b & d
\end{array}\right)=\left(\begin{array}{cc}
0 & ad-bc\\
bc-ad & 0
\end{array}\right).\label{eq:intermediate-1}
\end{equation}
Consequently, we have 
\begin{equation}
\overline{\left(\hat{\mathbf{F}}_{t_{0}}^{t}\right)^{*}\Delta_{\hat{\mathbf{x}}}\hat{\mathbf{u}}}(\hat{\mathbf{x}}_{0})=\mathbf{J}\boldsymbol{\nabla}_{0}\overline{\hat{\omega}\left(\hat{\mathbf{F}}_{t_{0}}^{t}\left(\hat{\mathbf{x}}_{0}\right),t\right)},
\end{equation}
and hence the averaged Lagrangian vorticity $\overline{\hat{\omega}\left(\hat{\mathbf{F}}_{t_{0}}^{t}\left(\mathbf{x}_{0}\right),t\right)}$
acts as an autonomous Hamiltonian (or steady stream function) for
the $\hat{\mathbf{x}}_{0}$-component of eq. \eqref{eq:momentum barrier eq for 2D flow},
as claimed in formula \eqref{eq:2D incompressible NS Lagrangian momentum barrier eq}.
Consequently, initial positions of material barriers to momentum transport
are level curves of the time-averaged Lagrangian vorticity $\overline{\omega\left(\hat{\mathbf{F}}_{t_{0}}^{t}\left(\mathbf{x}_{0}\right),t\right)}$,
as claimed. Furthermore, the instantaneous limit of eq. \eqref{eq:2D incompressible NS Lagrangian momentum barrier eq}
is \eqref{eq:2D incompressible NS Eulerian momentum barrier eq} and,
accordingly, Eulerian barriers to momentum transport are level
curves of the Hamiltonian $\hat{\omega}\left(\mathbf{x},t\right)$

\subsection{Poof of Theorem \ref{thm:2D vorticity barriers}}

For $\mathbf{u}$ defined in \eqref{eq:2D ansatz} and \eqref{eq:w is the vorticity},
the full vorticity of the 3D flow is given by 
\begin{equation}
\boldsymbol{\mathbf{\omega}}(\mathbf{x},t)=\left(\partial_{x_{2}}\hat{\omega}(\boldsymbol{\hat{\mathbf{x}}},t),-\partial_{x_{1}}\hat{\omega}(\boldsymbol{\hat{\mathbf{x}}},t),\hat{\omega}(\boldsymbol{\hat{\mathbf{x}}},t)\right),
\end{equation}
implying
\begin{equation}
\Delta\boldsymbol{\omega}=\left(\begin{array}{c}
\partial_{x_{2}}\Delta_{\boldsymbol{\hat{\mathbf{x}}}}\hat{\omega}\\
-\partial_{x_{1}}\Delta_{\boldsymbol{\hat{\mathbf{x}}}}\hat{\omega}\\
\Delta_{\boldsymbol{\hat{\mathbf{x}}}}\hat{\omega}
\end{array}\right).
\end{equation}
In all $x_{3}=const.$ planes, therefore, the vector field $\Delta\boldsymbol{\omega}$
admits the same reduced Hamiltonian dynamics, with the Hamiltonian
$H=\Delta_{\boldsymbol{\hat{\mathbf{x}}}}\hat{\omega}=\frac{1}{\nu}\frac{D}{Dt}\hat{\omega}$
acting as the stream function in that plane. With the notation $\mathbf{J=}\left(\begin{array}{cc}
0 & 1\\
-1 & 0
\end{array}\right)$ , we use the calculations in \eqref{eq:pullback for 2D flows} to
obtain
\begin{align}
\left(\mathbf{F}_{t_{0}}^{t}\right)^{*}\Delta\bm{\omega}(\mathbf{x}_{0}) & =\left(\begin{array}{c}
\boldsymbol{\nabla}_{\mathbf{\hat{\mathbf{x}}}}\hat{\mathbf{F}}_{t}^{t_{0}}\left(\mathbf{\hat{\mathbf{x}}}\right)\mathbf{J}\boldsymbol{\nabla}_{\mathbf{\hat{\mathbf{x}}}}\frac{1}{\nu}\frac{D}{Dt}\hat{\omega}\left(\hat{\mathbf{F}}_{t_{0}}^{t}\left(\hat{\mathbf{x}}_{0}\right),t\right)\\
\int_{t}^{t_{0}}\boldsymbol{\nabla}_{\mathbf{\hat{\mathbf{x}}}}\hat{\omega}\left(\hat{\mathbf{F}}_{t}^{s}\left(\hat{\mathbf{x}}\right),s\right)ds\cdot\mathbf{J}\boldsymbol{\nabla}_{\mathbf{\hat{\mathbf{x}}}}\frac{1}{\nu}\frac{D}{Dt}\hat{\omega}\left(\hat{\mathbf{F}}_{t_{0}}^{t}\left(\hat{\mathbf{x}}_{0}\right),t\right)+\frac{1}{\nu}\frac{D}{Dt}\hat{\omega}\left(\hat{\mathbf{F}}_{t_{0}}^{t}\left(\hat{\mathbf{x}}_{0}\right),t\right)
\end{array}\right).
\end{align}

As a consequence, the first two components of the vorticity barrier
equation \eqref{eq:incompressible NS Lagrangian vorticity barrier eq}
are
\begin{align}
\tilde{\mathbf{x}}_{0}^{\prime} & =\nu\overline{\boldsymbol{\nabla}_{\mathbf{\hat{\mathbf{x}}}}\hat{\mathbf{F}}_{t}^{t_{0}}\left(\mathbf{\hat{\mathbf{x}}}\right)\mathbf{J}\boldsymbol{\nabla}_{\mathbf{\hat{\mathbf{x}}}}\frac{D}{Dt}\hat{\omega}\left(\hat{\mathbf{F}}_{t_{0}}^{t}\left(\hat{\mathbf{x}}_{0}\right),t\right)}\nonumber \\
& =\nu\overline{\boldsymbol{\nabla}_{\mathbf{\hat{\mathbf{x}}}}\hat{\mathbf{F}}_{t}^{t_{0}}\left(\mathbf{\hat{\mathbf{x}}}\right)\mathbf{J}\left[\boldsymbol{\nabla}_{\mathbf{\hat{\mathbf{x}}}}\hat{\mathbf{F}}_{t}^{t_{0}}\left(\mathbf{\hat{\mathbf{x}}}\right)\right]^{T}\boldsymbol{\nabla}_{\hat{\mathbf{x}}_{0}}\frac{D\hat{\omega}}{Dt}\left(\hat{\mathbf{F}}_{t_{0}}^{t}\left(\hat{\mathbf{x}}_{0}\right),t\right)}.\label{eq:2D prelim}
\end{align}
Using formula \eqref{eq:intermediate-1} again, we obtain from \eqref{eq:2D prelim}
that 2D Lagrangian vorticity-diffusion barriers must satisfy
\begin{equation}
\hat{\mathbf{x}}_{0}^{\prime}=\nu\mathbf{J}\boldsymbol{\nabla}_{\hat{\mathbf{x}}_{0}}H_{t_{0}}^{t_{1}}\left(\hat{\mathbf{x}}_{0}\right),\qquad H_{t_{0}}^{t_{1}}\left(\hat{\mathbf{x}}_{0}\right)=\frac{\delta\hat{\omega}\left(\hat{\mathbf{x}}_{0},t_{0},t_{1}\right)}{t_{1}-t_{0}},\label{eq:2D final}
\end{equation}
as claimed in formula \eqref{eq:2D incompressible NS Lagrangian vorticity barrier eq},
with $H_{t_{0}}^{t_{1}}\left(\hat{\mathbf{x}}_{0}\right)$ playing
the role of a Hamiltonian for the two-dimensional $\hat{\mathbf{x}}_{0}$-component
of the full material barrier equation, which is therefore of the general
form 
\begin{align}
\hat{\mathbf{x}}_{0}^{\prime} & =\nu\mathbf{J}\boldsymbol{\nabla}_{\hat{\mathbf{x}}_{0}}H_{t_{0}}^{t_{1}}\left(\hat{\mathbf{x}}_{0}\right),\nonumber \\
x_{03}^{\prime} & =\nu B(\mathbf{\hat{\mathbf{x}}}_{0},t_{1},t_{0}),\label{eq:barrier eq for 2D flows}
\end{align}
for an appropriate scalar-valued function $G_{t_{0}}^{t_{1}}(\hat{\mathbf{x}}_{0})$.
As trajectories $\left\{ \hat{\mathbf{x}}_{0}(s)\right\} _{s\in\mathbb{R}}$
are contained in the level curves of the Hamiltonian $H_{t_{0}}^{t_{1}}\left(\hat{\mathbf{x}}_{0}\right)$,
we obtain the statement of Theorem \ref{thm:2D vorticity barriers},
using the definition of $H_{t_{0}}^{t_{1}}$ from \eqref{eq:2D final}.

\section{Proof of Theorem \ref{thm: Separable Beltrami momentum and vorticity barriers }}

To identify barrier equations for directionally steady Beltrami flows,
note that the flow map for the particle motion ODE
\begin{equation}
\dot{\mathbf{x}}=\alpha(t)\mathbf{u}_{0}(\mathbf{x}),\qquad\alpha(t)=e^{-\nu k^{2}\left(t-t_{0}\right)},\label{eq:separable ODE}
\end{equation}
of any such flow can be computed from the flow map $\mathbf{G}_{t_{0}}^{\tau}(\mathbf{x}_{0})$
of the autonomous ODE $\dot{\mathbf{x}}=\mathbf{u}_{0}(\mathbf{x})$
as
\begin{equation}
\mathbf{F}_{t_{0}}^{t}(\mathbf{x}_{0})=\mathbf{G}_{t_{0}}^{\tau(t)}\left(\mathbf{x}_{0}\right)=\mathbf{G}_{t_{0}}^{\int_{t_{0}}^{t}\alpha(s)\,ds}\left(\mathbf{x}_{0}\right),\label{eq:flow map for separable flows}
\end{equation}
as one verifies by direct substitution of this $\mathbf{F}_{t_{0}}^{t}(\mathbf{x}_{0})$
into \eqref{eq:separable ODE}. Since $\dot{\mathbf{x}}=\mathbf{u}_{0}(\mathbf{x})$
is an autonomous ODE, $\bm{\mathbf{u}}_{0}(\mathbf{F}_{t_{0}}^{t}(\mathbf{x}_{0}))=\bm{\mathbf{u}}_{0}\left(\mathbf{G}_{t_{0}}^{\tau}\left(\mathbf{x}_{0}\right)\right)$
is a solution of its equation of variations, i.e., 
\begin{equation}
\bm{\mathbf{u}}_{0}\left(\mathbf{G}_{t_{0}}^{\tau}\left(\mathbf{x}_{0}\right)\right)=\boldsymbol{\nabla}\mathbf{G}_{t_{0}}^{\tau}\left(\mathbf{x}_{0}\right)\bm{\mathbf{u}}_{0}\left(\mathbf{x}_{0}\right).
\end{equation}
This implies 
\begin{equation}
\bm{\mathbf{u}}_{0}\left(\mathbf{G}_{t_{0}}^{\int_{t_{0}}^{t}\alpha(s)\,ds}\left(\mathbf{x}_{0}\right)\right)=\boldsymbol{\nabla}\mathbf{G}_{t_{0}}^{\int_{t_{0}}^{t}\alpha(s)\,ds}\left(\mathbf{x}_{0}\right)\bm{\mathbf{u}}_{0}\left(\mathbf{x}_{0}\right),
\end{equation}
or, equivalently, by \eqref{eq:flow map for separable flows},
\begin{equation}
\bm{\mathbf{u}}_{0}\left(\mathbf{F}_{t_{0}}^{t}(\mathbf{x}_{0})\right)=\boldsymbol{\nabla}\mathbf{F}_{t_{0}}^{t}(\mathbf{x}_{0})\bm{\mathbf{u}}_{0}\left(\mathbf{x}_{0}\right).
\end{equation}
Multiplying both sides of this equation by $\alpha(t)$ leads to the
identity.
\begin{equation}
\left[\boldsymbol{\nabla}\mathbf{F}_{t_{0}}^{t}(\mathbf{x}_{0})\right]^{-1}\left(\alpha(t)\bm{\mathbf{u}}_{0}(\mathbf{F}_{t_{0}}^{t}(\mathbf{x}_{0}))\right)=\alpha(t)\bm{\mathbf{u}}_{0}(\mathbf{x}_{0}).\label{eq:pullback for separable flows}
\end{equation}

As a consequence of the relation \eqref{eq:pullback for separable flows},
for a directionally steady, strong Beltrami flow, the linear momentum
barrier equation \eqref{eq:incompressible NS Lagrangian momentum barrier eq}
takes the specific form
\begin{align}
\mathbf{x}_{0}^{\prime} & =\mathbf{b}_{t_{0}}^{t_{1}}=\nu\rho\,\overline{\left(\mathbf{F}_{t_{0}}^{t}\right)^{*}\Delta\mathbf{u}}=-\nu\rho\overline{\left(\mathbf{F}_{t_{0}}^{t}\right)^{*}\alpha k^{2}\mathbf{u}_{0}}\nonumber \\
& =-\frac{\nu\rho}{t_{1}-t_{0}}\int_{t_{0}}^{t_{1}}k^{2}\left[\boldsymbol{\nabla}\mathbf{F}_{t_{0}}^{t}(\mathbf{x}_{0})\right]^{-1}\left(\alpha(t)\bm{\mathbf{u}}_{0}(\mathbf{F}_{t_{0}}^{t}(\mathbf{x}_{0}))\right)\,dt\nonumber \\
& =-\frac{\nu\rho\int_{t_{0}}^{t_{1}}k^{2}\alpha(t)\,dt}{t_{1}-t_{0}}\bm{\mathbf{u}}_{0}(\mathbf{x}_{0}).
\end{align}
After rescaling the independent variable $s$ in this ODE as $s\to s\frac{t_{0}-t_{1}}{\nu\rho\int_{t_{0}}^{t_{1}}k^{2}\alpha(t)\,dt}$,
we obtain the Lagrangian and Eulerian momentum barrier equations 
\begin{align}
\mathbf{x}_{0}^{\prime} & =\bm{\mathbf{u}}_{0}(\mathbf{x}_{0}),\nonumber \\
\mathbf{x}^{\prime} & =\bm{\mathbf{u}}_{0}(\mathbf{x}).\label{eq:diffusive Eulerian barrier equation for Bernoulli flows}
\end{align}
Note that all invariant manifolds of this barrier equation coincide
with invariant manifolds of the particle motion \eqref{eq:separable ODE}
of the directionally steady Beltrami flow defined by \eqref{eq:separable ODE},
which proves the statement of Theorem \ref{thm: Separable Beltrami momentum and vorticity barriers }
for linear momentum barriers.

With the relation \eqref{eq:pullback for separable flows}, the vorticity
barrier equation \eqref{eq:incompressible NS Lagrangian vorticity barrier eq}
for directionally steady Beltrami flows takes the specific form
\begin{align}
\mathbf{x}_{0}^{\prime} & =\mathbf{b}_{t_{0}}^{t_{1}}=\nu\,\overline{\left(\mathbf{F}_{t_{0}}^{t}\right)^{*}\Delta\boldsymbol{\omega}}=-\nu\,\overline{\left(\mathbf{F}_{t_{0}}^{t}\right)^{*}\boldsymbol{\nabla}\times\left(\boldsymbol{\nabla}\times\boldsymbol{\omega}\right)}=-\nu\overline{\left(\mathbf{F}_{t_{0}}^{t}\right)^{*}\alpha(t)k^{3}\mathbf{u}_{0}},\nonumber \\
& =-\frac{\nu}{t_{1}-t_{0}}\int_{t_{0}}^{t_{1}}k^{3}\left[\nabla_{0}\mathbf{F}_{t_{0}}^{t}(\mathbf{x}_{0})\right]^{-1}\left(\alpha(t)\bm{\mathbf{u}}_{0}(\mathbf{F}_{t_{0}}^{t}(\mathbf{x}_{0}))\right)\,dt\nonumber \\
& =-\frac{\nu\int_{t_{0}}^{t_{1}}k^{3}\alpha(t)\,dt}{t_{1}-t_{0}}\bm{\mathbf{u}}_{0}(\mathbf{x}_{0}).
\end{align}
Again, an appropriate rescaling of time shows that all invariant manifolds
of this barrier equation coincide with invariant manifolds of the
particle motion \eqref{eq:separable ODE} of the directionally steady
Beltrami velocity field $\mathbf{u}(\mathbf{x},t)$, which proves
the statement of Theorem \ref{thm: Separable Beltrami momentum and vorticity barriers }
for vorticity barriers.


\begin{thebibliography}{10}

	\bibitem[Arnold, V.I.], Arnold, V.I., \emph{Mathematical Methods of Classical
		Mechanics,} Springer, New York (1978).
		
	\bibitem[Adrian, R. J., Meinhart, C. D. \& Tomkins, C. D.], Adrian, R. J., Meinhart, C. D. \& Tomkins, C. D.
	2000 Vortex organization in the outer region of the turbulent boundary
	layer. \emph{J. Fluid Mech}. \textbf{422}, 1\textendash 54.

	\bibitem[Anghan, C., Dave, S., Saincher, S. \& Banerjee,
	J.], Anghan, C., Dave, S., Saincher, S. \& Banerjee,
	J. 2014 Direct numerical simulation of transitional and turbulent
	round jets: Evolution of vortical structures and turbulence budget,
	\emph{Phys. Fluids} \textbf{31, }053606.

	\bibitem[Antuono, M.], Antuono, M. 2020 Tri-periodic fully three-dimensional
	analytic solutions for the Navier\textendash Stokes equations. \emph{J.
	Fluid Mech.} \textbf{890}, A23.

	\bibitem[Arnold, V.I.b], Arnold, V.I. 1978 \emph{Mathematical Methods of
	Classical Mechanics,} Springer, New York.

	\bibitem[Arnold, V.I., and Keshin, B.A.], Arnold, V.I., and Keshin, B.A. 1998\emph{ Topological
	Methods in Hydrodynamics}. Springer, New York. 

	\bibitem[Aref, H., Blake, J.R., Budisic, M., Cardoso, S.S.S.,
	Cartwright, J.H.E., Clercx, H.J.H., El Omari, K., Feudel, U., Golestanian,
	L., Gouillart, E., van Heijst, G.F., Krasnopolskaya, T.S., Le Guer,
	Y., MacKay, R.S., Meleshko, V.V., Metcalfe, G.G., Mezic, I., de Moura,
	A.P.S., Piro, O., Speetjens, M.F.M., Sturman, R., Thiffeault, J.-C.,
	Tuval, I.], Aref, H., Blake, J.R., Budisic, M., Cardoso, S.S.S.,
	Cartwright, J.H.E., Clercx, H.J.H., El Omari, K., Feudel, U., Golestanian,
	L., Gouillart, E., van Heijst, G.F., Krasnopolskaya, T.S., Le Guer,
	Y., MacKay, R.S., Meleshko, V.V., Metcalfe, G.G., Mezic, I., de Moura,
	A.P.S., Piro, O., Speetjens, M.F.M., Sturman, R., Thiffeault, J.-C.,
	Tuval, I. 2017 Frontiers of chaotic advection. \emph{Rev. Modern Phys.
	}\textbf{\emph{89, }}025007
	
	\bibitem[Balasuriya, S., Ouellette, N. T.,  \& Rypina, I.], Balasuriya, S., Ouellette, N. T.,  \& Rypina, I. 2018  Generalized Lagrangian Coherent Structures, \emph{Physica D}. \textbf{372}, 31\textendash 51.

	\bibitem[Batchelor, G. K.], Batchelor, G. K.\emph{ }2000\emph{ An Introduction
	to Fluid Mechanics. }Cambridge University Press, Cambridge.

	\bibitem[Barbato, D., Berselli, L.-C., \& Grisanti, C.R.], Barbato, D., Berselli, L.-C., \& Grisanti, C.R.
	2007 Analytical and numerical results for the rational large eddy
	simulation model. \emph{J. Math. Fluid Mech}. \textbf{9}, 44\textendash 74.

	\bibitem[Bird, R. B., Stewart, W. E., \& Lightfoot, E. N.], Bird, R. B., Stewart, W. E., \& Lightfoot, E. N. 2007 \emph{
	Transport Phenomena.} John Wiley \& Sons, New York.

	\bibitem[Childress, S.], Childress, S., 2009\emph{ A Theoretical Introduction
	to Fluid Mechanics}. AMS, Providence.

	\bibitem[De Silva, C., Hutchins, N., \& Marusic, I.], De Silva, C., Hutchins, N., \& Marusic, I. 2016
	Uniform momentum zones in turbulent boundary layers. \emph{J.} \emph{Fluid
	Mech}., \textbf{786}, 309-331.

	\bibitem[Dinklage, A. Klinger, T., Marx, G., \& Schweikhard,
	L.], Dinklage, A. Klinger, T., Marx, G., \& Schweikhard,
	L. 2005 \emph{Plasma Physics -{}- Confinement, Transport and Collective
	Effects.} Springer, Berlin.

	\bibitem[Dombre, T., Frisch, U., Greene, J.M., H\'enon, M.,
	Mehr, A.], Dombre, T., Frisch, U., Greene, J.M., H\'enon, M.,
	Mehr, A., \& Soward, A.M. 1986 Chaotic streamlines in ABC flows, \emph{J.
	Fluid Mech.} \textbf{167, }353\textendash 391.

	\bibitem[Dubief, Y. \& Delcayre, F.], Dubief, Y. \& Delcayre, F. 2000 On coherent-vortex
	identification in turbulence, \emph{J. Turbulence} \textbf{1,} 011.

	\bibitem[Ethier, R.C., and Steinman, D.A.], Ethier, R.C., and Steinman, D.A. 1994 Exact fully
	3D Navier\textendash Stokes solutions for benchmarking. \emph{Int.
	J. Num. Methods. Fluids.} \textbf{19, }369-375.

	\bibitem[Epps, B.], Epps, B. 2017 Review of vortex identification methods.
	in\emph{ AIAA SciTech Forum}, 9-13 January 2017, Grapevine, Texas,
	55th AIAA Aerospace Sciences Meeting, 1-22.

	\bibitem[Falahatpisheha, A., Kheradvarc, A.], Falahatpisheha, A., Kheradvarc, A., 2015 A measure
	of axisymmetry for vortex rings. \emph{European Journal of Mechanics
	B/Fluids }\textbf{49, }264\textendash 271.

	\bibitem[Farazmand, M., Kevlahan, N. \& Protas, B.], Farazmand, M., Kevlahan, N. \& Protas, B. 2011
	Controlling the dual cascades of two-dimensional turbulence, \emph{J.
	Fluid Mech.} \textbf{668}, 202-222.

	\bibitem[Farazmand, M. \& Haller, G.], Farazmand, M. \& Haller, G. 2016 Polar rotation
	angle identifies elliptic islands in unsteady dynamical systems. \emph{Physica
	D} \textbf{315, }1-12.

	\bibitem[Fraenkel, L. E.], Fraenkel, L. E. 1970 On steady vortex rings of
	small cross-section in an ideal fluid. \emph{Proc. Roy. Soc. A.} \textbf{316},
	29. 

	\bibitem[Fraenkel, L. E.b], Fraenkel, L. E. 1972 Examples of steady vortex
	rings of small cross-section in an ideal fluid. \emph{J . Fluid Mech}.\textbf{
	51} (1972) 119.

	\bibitem[Gao, F., Ma, W., Zambonini, G., Boudet, J., Ottavy,
	X., Lu, L. \& Shao, L.], Gao, F., Ma, W., Zambonini, G., Boudet, J., Ottavy,
	X., Lu, L. \& Shao, L. 2015 Large-eddy simulation of 3-D corner separation
	in a linear compressor cascade, \emph{Phys. Fluids} \textbf{27, }085105.

	\bibitem[Guckenheimer, J. \& Holmes, P.], Guckenheimer, J. \& Holmes, P. 1983 \emph{Nonlinear
	Oscillations, Dynamical Systems and Bifurcations of Vector Fields}.
	Springer, New York.

	\bibitem[Gurtin, M.E., Fried, E., \& Anand, L.], Gurtin, M.E., Fried, E., \& Anand, L. 2013 \emph{The
	Mechanics and Thermodynamics of Continua}. Cambridge University Press,
	Cambridge..

	\bibitem[G\"unther, T. \& Theisel, H.], G\"unther, T. \& Theisel, H. 2018 The state of the
	art in vortex extraction. \emph{Comp. Graphics Forum.}\textbf{\emph{
	}}\textbf{37, }149-173.

	\bibitem[Hadjighasem, A., Farazmand, M., Blazevski,
	D. Froyland, G. \& Haller, G.], Hadjighasem, A., Farazmand, M., Blazevski,
	D. Froyland, G. \& Haller, G. 2017 A critical comparison of Lagrangian
	methods for coherent structure detection. \emph{Chaos} \textbf{27,
	}053104.

	\bibitem[Haller, G.], Haller, G.2001 Distinguished material surfaces
	and coherent structures in 3D fluid flows. \emph{Physica D} \textbf{149,
	}248-277.

	\bibitem[Haller, G.b], Haller, G. 2015 Lagrangian Coherent Structures.
	\emph{Annual Rev. Fluid. Mech}, \textbf{47, }137-162.

	\bibitem[Haller, G., Hadjighasem, A., Farazamand, M., \&
	Huhn, F.], Haller, G., Hadjighasem, A., Farazamand, M., \&
	Huhn, F. 2016 Defining coherent vortices objectively from the vorticity.
	\emph{J. Fluid Mech.} \textbf{795, }136-173. 

	\bibitem[Haller, G., Karrasch, D., \& Kogelbauer, F.], Haller, G., Karrasch, D., \& Kogelbauer, F. 2018
	Material barriers to diffusive and stochastic transport. \emph{Proc.
	Natl. Acad. Sci. U.S.A.},\textbf{115}/37, 9074-9079.

	\bibitem[Haller, G., Karrasch, D., \& Kogelbauer, F.b], Haller, G., Karrasch, D., \& Kogelbauer, F., 2020
	Barriers to the transport of diffusive scalars in compressible flows,
	\emph{SIAM J. Appl. Dyn. Sys}., 85\textendash 123.

\bibitem[Hamilton, J., Kim, J., \& Waleffe, F.], Hamilton, J., Kim, J., \& Waleffe, F. 1995 Regeneration mechanisms of near-wall turbulence structures.  \emph{J. Fluid Mech.},  \textbf{287}, 317\textendash 348. 


	
	\bibitem[Hasegawa, Y., Quadrio, M. \& Frohnapfel, B.], Hasegawa, Y., Quadrio, M. \& Frohnapfel, B., 2014
	Numerical simulation of turbulent duct flows at constant power input,
	\emph{J. Fluid Mech.} \textbf{750, }191\textendash 209.
	
		\bibitem[Hunt, J. C. R., Wray, A. \& Moin, P.], Hunt, J. C. R., Wray, A. \& Moin, P. 1988 Eddies,
	stream, and convergence zones in turbulent flows. \emph{Center for
	Turbulence Research Report} CTR-S88.
	
	\bibitem[Hutchins, N. \& Marusic, I.2007], Hutchins, N. \& Marusic, I. 2007 Large-scale influences in near-wall turbulence. \emph{Phil. Trans.R. Soc. A} \textbf{365}, 647\textendash 664.

	\bibitem[Jantzen, R.T., Taira, K., Granlund, K.O. \& Ol,
	M.V.], Jantzen, R.T., Taira, K., Granlund, K.O. \& Ol,
	M.V. 2019 Vortex dynamics around pitching plates, \emph{Phys. Fluids}
	\textbf{26, }065105.

	\bibitem[Jeong, J. \& Hussain, F.], Jeong, J. \& Hussain, F. 1995 On the identification
	of a vortex. \emph{J. Fluid Mech.} \textbf{285, }69\textendash 94.

	\bibitem[Jeong, J., Hussain, F., Schoppa, W. \& Kim, J.], Jeong, J., Hussain, F., Schoppa, W. \& Kim, J. 1997
	Coherent structures near the wall in a turbulent channel flow. \emph{J.
	Fluid Mech.} \textbf{332, }185\textendash 214.
	
	\bibitem[Jim{\'e}nez, J., \& Pinelli, A.], Jim{\'e}nez, J., \& Pinelli, A. 1999 The autonomous cycle of near-wall turbulence. \emph{J.  Fluid Mech.} \textbf{389}, 335\textendash 359.

	\bibitem[Katsanoulis, S., Farazmand, M., Serra, M.
	\& Haller, G.], Katsanoulis, S., Farazmand, M., Serra, M.
	\& Haller, G. 2020 Vortex boundaries as barriers to diffusive vorticity
	transport in two-dimensional flows, \emph{Phys. Rev. Fluids }\textbf{5}\emph{,
	}024701.

	\bibitem[Kirkwood, J.], Kirkwood, J. 2018 \emph{Mathematical Physics
	with Partial Differential Equations}. Academic Press, London.

	\bibitem[Lele, S.K.], Lele, S.K. 1992 Compact finite-difference schemes
	with spectral-like resolution. \emph{J. Comp. Phys. }\textbf{103,
	}16\textendash 42.

	\bibitem[Luchini, P. \& Quadrio, M.], Luchini, P. \& Quadrio, M. 2006 A low-cost parallel
	implementation of direct numerical simulation of wall turbulence.
	\emph{J. Comp. Phys}. \textbf{211,} 551-571.

	\bibitem[Lugt, H.J.], Lugt, H.J. 1979 The dilemma of defining a vortex\textquotedbl ,
	in \emph{Recent Developments in Theoretical and Experimental Fluid
	Mechanics,} U. Muller, K. G. Riesner, and B. Schmidt, (eds.) \textbf{13,
	}309-321.

	\bibitem[MacKay, R.S.], MacKay, R.S. 1994 Transport in 3D volume-preserving
	flow. \emph{J. Nonlinear Sci.} \textbf{4, }329-354.

	\bibitem[Majda, A.J., and Bertozzi, A.L.], Majda, A.J., and Bertozzi, A.L.\emph{ }2002\emph{
	Vorticity and Incompressible Flow}. Cambridge University Press, Cambridge. 

	\bibitem[Marusic, I., Mathis, R.], Marusic, I., Mathis, R. \& Hutchins, N. 2010 Predictive
	model for wall-bounded turbulent flow. \emph{Science} \textbf{329},
	193\textendash 196.

	\bibitem[Mackay, R.S., Meiss, J. D.], Mackay, R.S., Meiss, J. D., \& Percival, I.C. 1984
	Transport in Hamiltonian systems, \emph{Physica D} \textbf{13} (1984)
	55-81.

	\bibitem[McMullan, W.A. \& Page, G.J.], McMullan, W.A. \& Page, G.J. 2012 Towards Large
	Eddy Simulation of gas turbine compressors, \emph{Progr. in Aerospace.
	Sci}. \textbf{52}, 30-{}-47.

	\bibitem[Meiss, J. D.], Meiss, J. D. 1992 Symplectic maps, variational principles,
	and transport. \emph{Rev. Mod. Phys.} \textbf{64}, 795-848.

	\bibitem[Meyers, J. \& Meneveau, C.], Meyers, J. \& Meneveau, C. 2013 Flow visualization
	using momentum and energy transport tubes and applications to turbulent
	flow in wind farms. \textbf{715,} 335-{}-358 
	
	\bibitem[Nolan, P.J., Serra, M. \& Ross, S.D.], Nolan, P.J., Serra, M. \& Ross, S.D. 2020 Finite-time Lyapunov exponents in the instantaneous limit and material transport. \emph{Nonlinear Dyn.} https://doi.org/10.1007/s11071-020-05713-4

	\bibitem[Norbury, J.], Norbury, J. 1973 A family of steady vortex rings.
	\emph{J. Fluid Mech.} \textbf{57,} 417-431. 

	\bibitem[Ogden, R.W.], Ogden, R.W. 1984 \emph{Non-linear Elastic Deformations},
	Ellis Horwood, Chichester..

	\bibitem[Ottino, J.M.], Ottino, J.M. 1989 \emph{The Kinematics of Mixing:
	Stretching, Chaos and Transport,} Cambridge University Press, Cambridge.

	\bibitem[Pandey, A. Scheel, J.D. \& Schumacher, J.], Pandey, A. Scheel, J.D. \& Schumacher, J. 2018
	Turbulent superstructures in Rayleigh-B\'enard convection, \emph{Nature
	Comm.} \textbf{9}, 2118

	\bibitem[Pedergnana, T., Oettinger, D., Langlois, G.
	P. \& Haller, G.], Pedergnana, T., Oettinger, D., Langlois, G.
	P. \& Haller, G. 2020 Explicit unsteady Navier-Stokes solutions and
	their analysis via local vortex criteria \emph{Phys. Fluids} \textbf{32,
	}046603.

	\bibitem[Pitton, E., Marchioli, C., Lavezzo, C., Soldati,
	A. \& Toschi, F.], Pitton, E., Marchioli, C., Lavezzo, C., Soldati,
	A. \& Toschi, F. 2012 Anisotropy in pair dispersion of inertial particles
	in turbulent channel, \emph{Phys. Fluids }\textbf{24, }073305.
	
	
	\bibitem[Quadrio, M. \& Luchini, P.], Quadrio, M. \& Luchini, P. 2003 Integral space–time scales in turbulent wall flows. \emph{Phys. Fluids} \textbf{24} 2219-2227

	
	\bibitem[Robinson, S. K.], Robinson, S. K. 1991 Coherent motions in the turbulent boundary layer, \emph{Ann. Rev. Fluid Mech.} \textbf{23,} 601--639.

	\bibitem[Rosner, D.], Rosner, D. 2000 \emph{Transport Processes in Chemically	Reacting Flow Systems}. Dover Publications.
	
	\bibitem[Saffman, P. G.,  Ablowitz, M. J., Hinch, E., Ockendon, J. R., \& Olver, P. J.], Saffman, P. G.,  Ablowitz, M. J., Hinch, E., Ockendon, J. R., \& Olver, P. J. 1992. \emph{Vortex Dynamics.} Cambridge University Press, Cambridge.
	
	\bibitem[Sadlo, F. \& Pikert, R.], Sadlo, F. \& Pikert, R.,], Sadlo, F. \& Pikert, R.], Sadlo, F. \& Pikert, R., 2009 Visualizing Lagrangian coherent structures and comparison to vector field topology, in \emph{Topology-Based Methods in Visualization II}, Springer, Berlin Heidelberg, 15-29.
	
	\bibitem[Schindler, B., Peikert, R., Ruchs, R. \& Theisel, H.], Schindler, B. Peikert, R., Ruchs, R. \& Theisel, H. 2012 Ridge concepts for the visualization of Lagrangian coherent structures, in \emph{Topological Methods in Data Analysis and Visualization II: Theory, Algorithms, and Applications}, Springer, Berlin Heidelberg, 221-235.

	\bibitem[Serra, M. \& Haller, G.], Serra, M. \& Haller, G. 2016 Objective Eulerian
	coherent structures. \emph{Chaos} \textbf{26,} 053110.

	\bibitem[Surana, A., Grunberg, O. \& Haller, G.], Surana, A., Grunberg, O. \& Haller, G. 2006 Exact
	theory of three-dimensional flow separation. Part I. Steady separation
	\emph{J. Fluid. Mech}. \textbf{564,} 57-103.

	\bibitem[van Hinsberg, M.A.T., Ten Thije Boonkamp,
	J.H.M., Toschi, F. \& Clercx], van Hinsberg, M.A.T., Ten Thije Boonkamp,
	J.H.M., Toschi, F. \& Clercx, H.J.H. 2012 On the efficiency and accuracy
	of interpolation methods for spectral codes, \emph{J. Sci. Comput.
	}\textbf{34}\textbf{\emph{, }}479-{}-498.

	\bibitem[Wang, C.Y.], Wang, C.Y. 1990 Exact solutions of the Navier-Stokes
	equations- the generalized Beltrami flows, review and extension. \emph{Acta
	Mechanica }\textbf{81, }69-74.

	\bibitem[Weiss, J.B. and Provenzale, A.], Weiss, J.B. and Provenzale, A. 2008 \emph{Transport
	and Mixing in Geophysical Flows: Creators of Modern Physics} \textbf{744},
	Springer, Berlin.

\end{thebibliography}
\end{document}